%% file: main.tex
\pdfoutput=1
\documentclass[final]{article}

\PassOptionsToPackage{numbers, compress}{natbib}
\usepackage{formatting/neurips_2024}
\usepackage{fix-cm}

\usepackage{url}            %
\usepackage{booktabs}       %
\usepackage{amsfonts}       %
\usepackage{nicefrac}       %
\usepackage{microtype}      %
\usepackage{xcolor}         %
\usepackage{colortbl}
\usepackage{float}%
\usepackage{titletoc}
\usepackage{standalone} %
\usepackage{adjustbox} %
\usepackage{overpic}    %

\usepackage[colorlinks=true,
    linktoc=all]{hyperref}       %
\hypersetup{
    urlcolor=[rgb]{0,0,0.5},
    citecolor=[rgb]{0,0,0.5},
    linkcolor=[rgb]{0,0,0.5}
}

\usepackage{packages}
\usepackage{nips-private}

\title{
Efficient Federated Learning against Heterogeneous
and Non-stationary Client Unavailability
}

\author{%
Ming Xiang$^1$ \quad
Stratis Ioannidis$^1$ \quad 
Edmund Yeh$^1$ \quad 
Carlee Joe-Wong$^2$ \quad 
Lili Su$^1$ \\
$^1$Northeastern University, Boston, MA \quad 
$^2$Carnegie Mellon University, Pittsburgh, PA\\
\texttt{\{xiang.mi,l.su\}@northeastern.edu} \\
\texttt{\{ioannidis,eyeh\}@ece.neu.edu}\\
\texttt{cjoewong@andrew.cmu.edu}
}

\begin{document}

\newpage

\maketitle

\begin{abstract}
Addressing intermittent client availability is critical for the real-world deployment of federated learning algorithms. 
Most prior work either overlooks the potential non-stationarity in the dynamics of client unavailability %
or requires substantial memory/computation overhead.
We study federated learning in the presence of 
heterogeneous and non-stationary client availability, which may occur when the deployment environments are uncertain, or the clients are mobile. 
The impacts of heterogeneity and non-stationarity on client unavailability can be significant, as we illustrate using~\FedAvg, the most widely adopted federated learning algorithm. 
We propose~\FedAPM, which includes novel algorithmic structures that (i) compensate for missed computations due to unavailability with only $O(1)$ additional memory and computation with respect to standard \FedAvg, and (ii) evenly diffuse local updates within the federated learning system through implicit gossiping, despite being agnostic to non-stationary dynamics.
We show that~\FedAPM~converges to a stationary point of even non-convex 
objectives while achieving the desired linear speedup property.
We corroborate our analysis with numerical experiments over diversified client unavailability dynamics on real-world data sets. 
\end{abstract}

\section{Introduction}
\label{sec: intro}
Federated learning is a %
distributed machine learning 
approach that enables training global models without disclosing raw local data 
\cite{mcmahan2017communication,kairouz2021advances}. 
It has been adopted in commercial applications such as autonomous vehicles 
\cite{chen2021bdfl,zeng2022federated,peng2023privacy}, the Internet of things \cite{nguyen2019diot}, and natural language processing \cite{yang2018applied,ramaswamy2019federated}.

Heterogeneous data 
and massive client populations
are two of the defining characteristics of cross-device federated learning systems \cite{mcmahan2017communication,kairouz2021advances}.
Despite intensive efforts \cite{mcmahan2017communication,Li2020,yuan2022,ruan2021towards,kairouz2021advances}, 
several key challenges that arise from the involvement of large-scale
client populations
are often overlooked in the existing literature \cite{perazzone2022communication}.
One of the primary hurdles is the issue of client unavailability.
Intuitively, more active clients drive the global model to their local optima by overfitting their local data,
which biases the training. 
In addition, the higher 
the uncertainty
in client unavailability, the larger the performance degradation. 
Concrete examples that confirm these intuitions in the context of~\FedAvg~- the most widely adopted federated learning algorithm - can be found in Section \ref{sec: obj inconsistency}.  
Client unavailability issues
can arise from internal factors such as
different working schedules %
and heterogeneous hardware/software constraints.
External factors, 
such as
poor network coverage and
frequent handovers of base stations due to fast movements, 
only exacerbate these problems \cite{tse2005fundamentals,wen2024communication,Ye2022TSP,bonawitz2019towards,kairouz2021advances}.
The intricate interplay of internal and external factors results in the {\em non-stationarity} and {\em heterogeneity} of client unavailability.

\begin{wrapfigure}[11]{r}{0.35\textwidth} 
\centering
\resizebox{.9\linewidth}{!}{
\includegraphics[width=\linewidth, trim=0 0 0 1.3cm,clip]{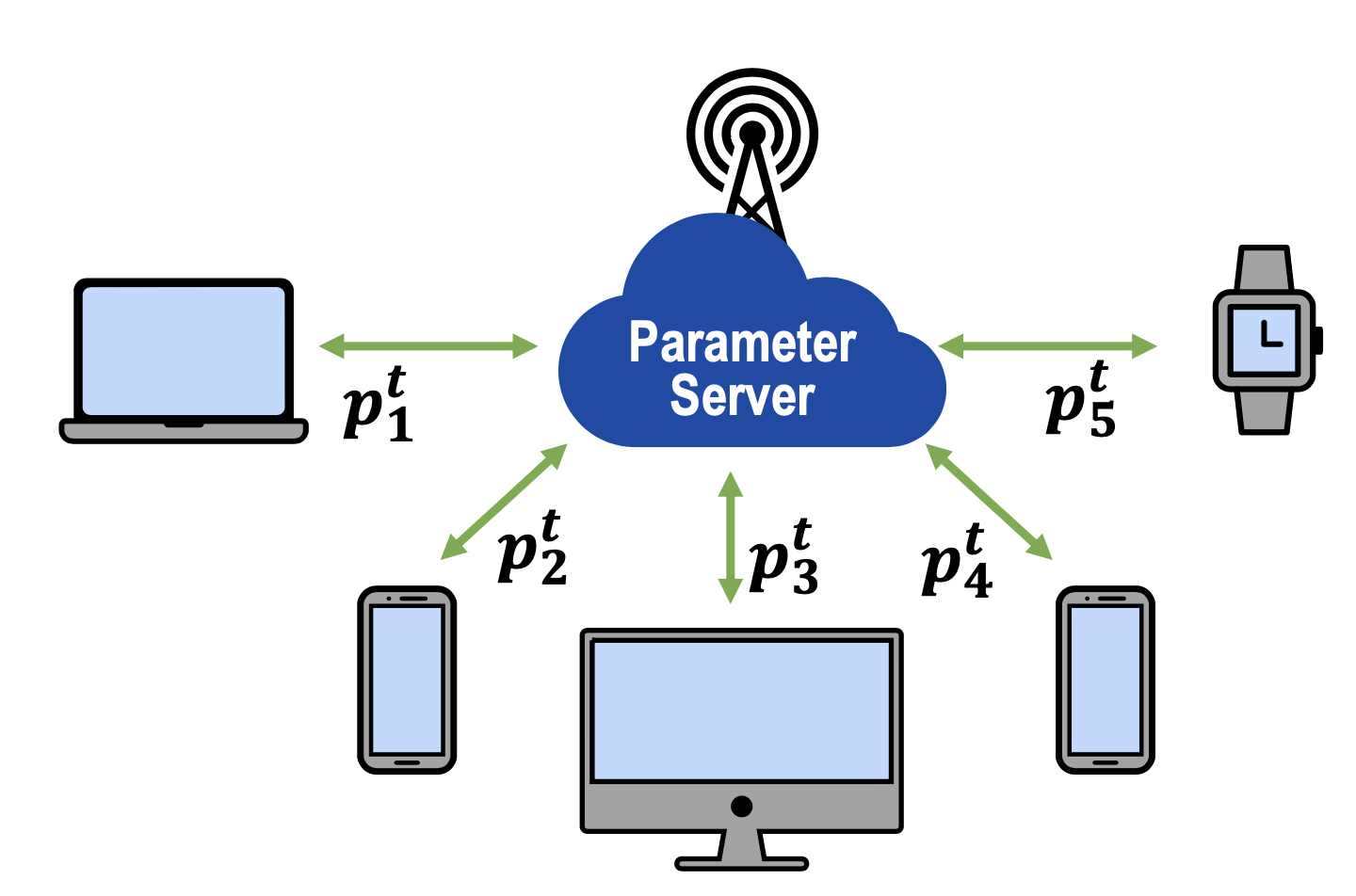} }
\caption{\footnotesize 
Client $i$'s available probabilities $p_i^t$'s are heterogeneous and are subject to {\em non-stationary} dynamics. 
}
\label{fig: system setup}
\end{wrapfigure} 
Most prior work either assumes
exact knowledge of the clients' available dynamics or
requires their dynamics to be benignly stationary
\cite{mcmahan2017communication,li2020federated,perazzone2022communication,wang2023lightweight,wang2022}. 
A related line of work studies asynchronous 
federated learning wherein clients are vulnerable to delays in 
message transmission
and the reported 
model updates
may be stale \cite{xie2019asynchronous,nguyen2022federated,toghani2022unbounded,koloskova2022sharper}.
The proposed methods therein
assume the availability of all clients or uniformly sampled clients, making them inapplicable to our settings. 
A few recent works \cite{ribero2022federated,xiang2023towards} study non-stationary dynamics. 
Ribero et al.~\cite{ribero2022federated} consider the settings where the available probabilities follow a homogeneous Markov chain.  
Xiang et al.~\cite{xiang2023towards} require that clients be capable of continuous local optimization regardless of communication failures. 
A handful of other works \cite{gu2021fast,yan2023federated} memorize the old gradients of the unavailable clients to compensate for their unavailability.
However,
the added memory burdens the federated learning system with substantial memory proportional to the product of the number of clients and the model dimension.

{\bf Contributions}. 
In this work, we focus on stochastic client unavailability, where
client $i$ is available for federated learning model training with probability $p_i^t$
at any time $t$.   
An illustration can be found in~\prettyref{fig: system setup}.
Our contributions are four-fold:
\begin{itemize}[leftmargin=*]
\item 
In~\prettyref{sec: obj inconsistency},
via constructing concrete examples, 
we demonstrate that %
both heterogeneity and non-stationarity of $p_i^t$ will result in 
bias and thus
significant performance degradation of~\FedAvg.  
\item 
In~\prettyref{sec: algorithm FedAPM},
we propose an algorithm named~\FedAPM, 
which features computational and memory efficiency:
only $O(1)$ additional computation and memory per client will be used
when compared with~\FedAvg.
The design of~\FedAPM~introduces two novel algorithmic structures:
{\em adaptive innovation echoing} and {\em implicit gossiping}.
At a high level,
these novel algorithmic structures
(i) help clients catch up on the missed computation, %
and 
(ii) simultaneously enable a balanced information mixture through implicit client-client gossip,  
which ultimately corrects the remaining bias. 
Notably, no direct neighbor information exchanges are used, 
and the client unavailability dynamics remains unknown to all clients and the parameter server. 
\item 
In~\prettyref{sec: convergence analysis},
we show that~\FedAPM~converges to a stationary point of even non-convex global objective and achieves the linear speedup property
without conditions on second-order partial derivatives of the loss function in analysis.
\item 
In~\prettyref{sec: numerical}, 
we validate our analysis with numerical experiments over diversified client unavailability dynamics on real-world data sets.
\end{itemize}

\section{Related Work}
\label{sec: related work} 

\noindent{\bf Dynamical client availability.} 
There is a recent surge of efforts to study time-varying client availability
\cite{ruan2021towards,ribero2022federated,chen2022optimal,wang2022,ribero2022federated,perazzone2022communication,xiang2023towards}, which  
can be roughly classified into two categories depending on whether the parameter server can unilaterally determine the participating clients. %

(i) %
{\em Controllable participation.}
Earlier research \cite{mcmahan2017communication,Li2020} presumes that, 
in each round, 
the parameter server could select a small set of clients either uniformly at random or in proportion to the volume of local data held by clients.  
More recently, 
Cho et al.\,\cite{cho2022towards} design adaptive and non-uniform client sampling to accelerate learning convergence, albeit at the cost of introducing a non-zero residual error.
In another work,
Cho et al.\,\cite{pmlr-v202-cho23b} study the convergence of~\FedAvg~with cyclic client participation.
Yet, %
the set of available clients is sampled uniformly at random per cyclic round and is decided unilaterally by the parameter server.
Perazzone et al. \cite{perazzone2022communication} consider heterogeneous and time-varying response rates $p_i^t$ under the assumptions that $p_i^t$ is known a priori and that the stochastic gradients are bounded in expectation. 
Furthermore, the dynamics of $p_i^t$ are determined by the parameter server by solving a stochastic optimization problem.     
Chen et al. \cite{chen2022optimal} propose a client sampling scheme wherein only the clients with the most ``important" updates communicate back to the parameter server. 
This sampling method can achieve performance comparable to that of full client participation, provided that  
$p_i^t$ is globally known to both the parameter server and the clients. %
Departing from this line of literature, our setup neither assumes any side information or prior knowledge of the response rates $p_i^t$ nor assumes that the parameter server has any influence on $p_i^t$. 

(ii) %
{\em Uncontrollable participation.}
There is a handful of work on building resilience against arbitrary client availability \cite{ribero2022federated,wang2022,yan2023federated,gu2021fast,yang2022anarchic,wang2023lightweight}.  
Ribero et al. \cite{ribero2022federated} consider random client availability whose underlying response rates are also heterogeneous and time-varying with unknown dynamics. 
However,
the underlying dynamics of $p_i^t$ in \cite{ribero2022federated} are assumed to follow a homogeneous Markov chain. 
Wang et al.\,\cite{wang2022} propose a generalized~\FedAvg
~that amplifies parameter updates every $P$ rounds for some carefully tuned $P$. 
Despite its elegant unified analysis and potential to accommodate non-independent unavailability dynamics, 
to reach a stationary point, 
$p_i^t$ needs to satisfy some assumptions to ensure roughly equal availability of all clients over every $P$ rounds.
Yang et al. \cite{yang2022anarchic} analyze a setting where clients participate in the training at their will. 
Yet, their convergence is shown to be up to a non-zero residual error. 
The algorithms proposed in \cite{gu2021fast,yan2023federated} 
share the same idea of using the memorized latest updates from unavailable clients for global aggregation. 
Despite superior numerical performance, 
both algorithms demand a substantial amount of additional memory \cite{wang2023lightweight}. 
For non-convex objectives, both \cite{yan2023federated} and \cite{gu2021fast} require an absolute bounded inactive period, and share similar technical assumptions such as almost surely bounded stochastic gradients \cite{yan2023federated} or Lipschitz Hessian \cite{gu2021fast}. 
Though bounded inactive periods are relevant for applications wherein the sensors wake up on a periodic schedule, this assumption is not satisfied even for the simple stochastic setting when clients are selected uniformly at random.  
Wang and Ji consider unknown heterogeneous $p_i$' in a concurrent work~\cite{wang2023lightweight};
however, $p_i$'s are assumed to be fixed over time. 

{\bf Asynchronous federated learning.}
Another related line of work is asynchronous federated learning. 
To the best of our knowledge, 
Xie et al.\,\cite{xie2019asynchronous} initialize the study of asynchronous federated learning, 
wherein the parameter server revises the global model every time it receives an update from a client.  
Convergence is shown under some technical assumptions such as weakly-convex global objectives, bounded delay, and bounded stochastic gradients.  
Zakerinia et al.\,\cite{zakerinia2023communication} propose QuAFL which is shown to be resilient to computation asynchronicity and quantized communication yet under the bounded and stationary delay assumption.  
Nguyen et al. \cite{nguyen2022federated} propose FedBuff, which uses additional memory to buffer asynchronous aggregation to achieve scalability and privacy. 
Convergence is shown under bounded gradients and bounded staleness assumptions. 
In fact, 
most convergence guarantees in the asynchronous federated learning literature rely on bounded staleness \cite{xie2019asynchronous,nguyen2022federated,toghani2022unbounded,koloskova2022sharper}, or bounded gradients \cite{xie2019asynchronous,nguyen2022federated,koloskova2022sharper}.
Recently, 
arbitrary delay is considered in the context of distributed SGD with bounded stochastic gradients and $(0, \zeta)$-bounded inter-client heterogeneity \cite{mishchenko2022asynchronous} (see Assumption \ref{ass: bounded similarity} for the definition).  
The convergence suffers from a non-zero residual term $O(\zeta^2)$. 
In contrast,
our convergence guarantee is free from non-zero residual terms and does not require gradients 
to be bounded.

\section{Problem Formulation}
\label{sec: problem formulation}

A federated learning system consists of a parameter server and $m$ clients that collaboratively minimize
\begin{align}
\label{eq: global obj}
\min\limits_{\x\in\reals^d} F(\x) \triangleq \frac{1}{m}\sum_{i=1}^m F_i(\x),
\end{align} 
where $F_i(\x) \triangleq \expects{\ell_i(\x;\xi_i)}{\xi_i \sim \calD_i}$ is the local objective and can be non-convex,
$\calD_i$ is the local distribution,
$\xi_i$ is a stochastic sample that client $i$ has access to,
$\ell_i$ is the local loss function,
and $d$ is the model dimension.

We use~\prettyref{ass: prob lower bound}
to capture the uncertain {\em non-stationary} dynamics and heterogeneity. 
Let $\calA^t $ denote the set of active clients,
$\Indc_{\{\cdot\}}$ an indicator function,
$T$ the number of total training rounds.
\begin{assumption}
\label{ass: prob lower bound}
    There exists 
    a %
    $\delta \in (0,1]$ such that
    $ p_i^t \triangleq \mathbb{E}[\Indc_{\{i\in\calA^t\}}] \ge \delta$, 
    where the events $\{i\in\calA^t\}$ are independent across clients $i$ and across rounds $t\in [T]$.  
\end{assumption}

    \prettyref{ass: prob lower bound}
    subsumes uniform availability \cite{li2020federated,yang2022anarchic} and
    stationary availability considered in
    \cite{wang2023lightweight}.
    Independent client unavailability is widely adopted by federated learning research \cite{li2020federated,Li2020,karimireddy2020scaffold,yang2021achieving,yang2022anarchic,wang2023lightweight}.
    Analyzing non-independent unavailability,
    together with uncertain and non-stationary dynamics in~\prettyref{ass: prob lower bound}, is in general challenging. 
    Specifically, the involved entanglement of stochastic gradient and availability statistics fundamentally complicates the theoretical analysis.
    However,
    we conjecture that independence and strictly positive probabilities are only necessary for the technical convenience of our analysis. %
    Our experiments in~\prettyref{sec: numerical} suggest that our algorithm offers notable improvement even in the presence of non-independent and occasionally zero-valued probabilities. 
    Future work will investigate how to provably accommodate correlated or zero-valued probabilities of arbitrary probabilistic trajectories.

\section{Heterogeneity and Non-stationarity May Lead to Significant Bias} %
\label{sec: obj inconsistency}
\begin{wrapfigure}[13]{r}{0.38\textwidth} 
\centering
\vspace*{-\baselineskip}
\resizebox{\linewidth}{!}{
\includegraphics[width=\linewidth,trim=0 0 0 0.2cm, clip]{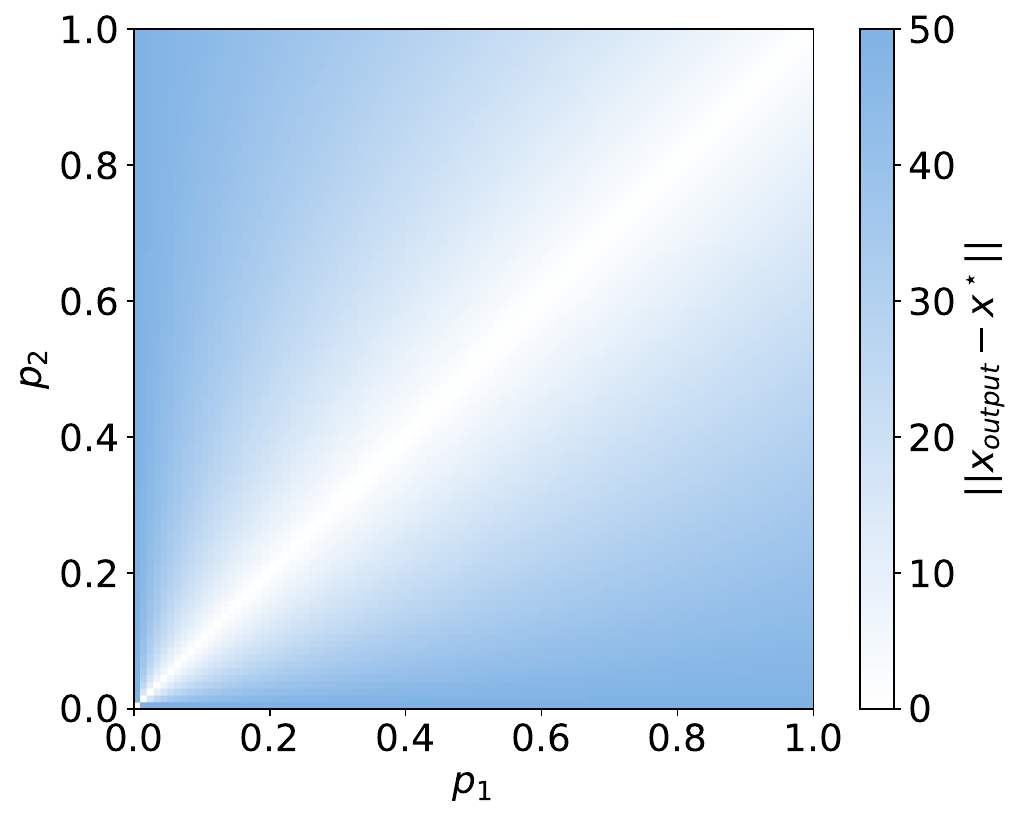} }
\vskip -0.6\baselineskip
\caption{\footnotesize 
Let $x_{\text{output}}\triangleq \lim_{t\diverge} \expect{x^t}$.
Under most of the choices of $p_1, p_2$, $x_{\text{output}}$ is far from $x^*$. %
}
\label{fig: obj inconsistency quadratic}
\end{wrapfigure}
In this section, we illustrate the impacts of heterogeneity and non-stationarity of client availability under the classic~\FedAvg. 
We use two examples to showcase the significant bias incurred.

\begin{example}[Heterogeneity]
\label{example: obj shift}
Suppose that $m=2$ and $p_i^t = p_i$ for $i \in [2]$. 
Let $F_i\pth{x} \triangleq \norm{x - u_i}^2 / 2$, where $x, u_i\in \reals$.   
The global objective~\eqref{eq: global obj} is  
\begin{align}
F\pth{x} = \frac{1}{2} (\norm{x - u_1}^2 + \norm{x - u_2}^2),
\label{eq: counterexample global objective}
\end{align}
with unique minimizer 
$x^\star = (u_1+u_2)/ 2$.  
Let $u_1=0$ and $u_2=100$. 
\prettyref{fig: obj inconsistency quadratic} 
illustrates %
how the heterogeneity 
in $p_i$ affects the expected output of~\FedAvg.
\end{example}
Example~\ref{example: obj shift} matches \cite[Theorem 1]{wang2023lightweight}, which shows that~\FedAvg~leads to a biased global objective~\eqref{eq: static biased obj} under heterogeneous $p_i$'s, 
and that~\eqref{eq: static biased obj} may be significantly away from \eqref{eq: global obj} depending on $p_i$'s.
\begin{align}
    \label{eq: static biased obj}
    \tilde{F} (\x) \triangleq 
    \sum_{i=1}^m
    \frac{p_i}{\sum_{j=1}^m p_j}
    F_i(\x).
\end{align}
When the probabilistic dynamics of $p_i^t$'s is non-stationary, obtaining an exact biased objective similar to~\eqref{eq: static biased obj} in a neat analytical form becomes challenging, if not impossible, due to the unstructured non-stationary dynamics. Fortunately, Example~2 helps us confirm that the complex interplay between $p_i^t$'s across rounds and clients will inevitably further degrade the performance of~\FedAvg~algorithm.

{\bf Example 2} (Non-stationarity){\bf .}
In~\prettyref{fig: motivating example non-stationary},
a total of $m = 100$ clients perform an image classification task on the SVHN dataset \cite{netzer2011readingdigits} under the~\FedAvg~algorithm,
whose local dataset distribution follows $\mathsf{Dirichlet}(0.1)$ \cite{hsu2019measuring}.
Clients become available with probability $p_i^t = p \cdot [\gamma \cdot \sin (0.1 \pi \cdot t) + (1 - \gamma)],~\forall i \in [m]$.
The hyperparameter details are deferred to~\prettyref{app: numerical}.
Observations can be found in the caption.
\begin{figure}[H]
    \centering
    \begin{subfigure}[b]{.45\textwidth}
    \includegraphics[width=\linewidth]{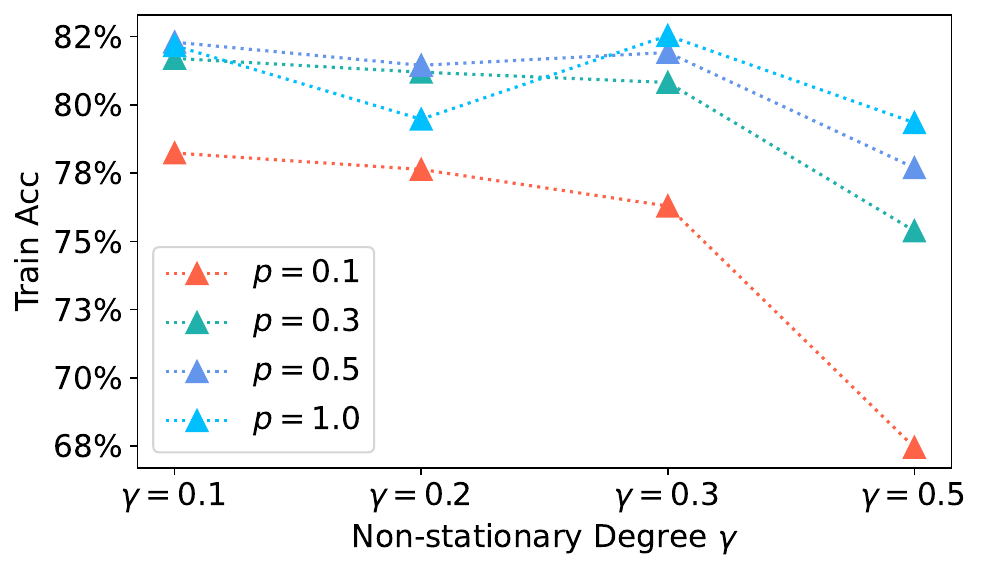}
    \vspace{-1.5em}
    \caption{\footnotesize Train accuracy.}
    \end{subfigure}
    \begin{subfigure}[b]{.45\textwidth}
    \includegraphics[width=\linewidth]{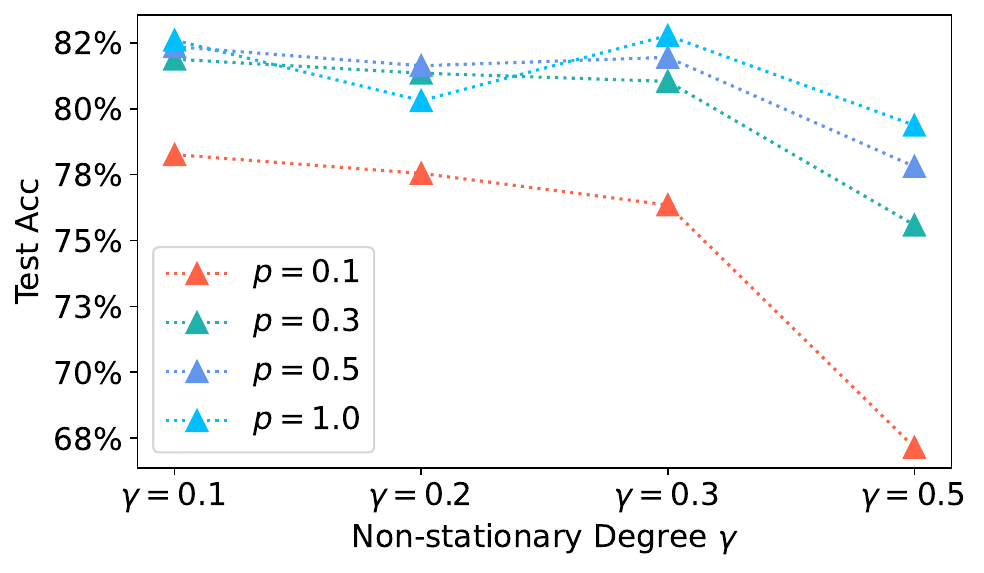}    
    \vspace{-1.5em}
    \caption{\footnotesize Test accuracy}
    \end{subfigure}
    \vspace{-0.5em}
    \caption{\footnotesize
    Train and test accuracy results in percentage (\%).
    In particular,
    the parameter $\gamma$ signifies the degree of non-stationary.
    Notice that,
    as the client availability becomes more non-stationary
    (a larger $\gamma$),
    \FedAvg~experiences a significant drop in accuracy.
    For example,
    both the train and test accuracies drop by over $10\%$ when $p=0.1$,
    and $\gamma$ increases from $0.1$ to $0.5$.
    }
    \label{fig: motivating example non-stationary}
\end{figure}

\section{
{\bf Fed}erated Agile Weight Re-Equalization (\FedAPM) 
}
\label{sec: algorithm FedAPM}

To minimize \eqref{eq: global obj}, 
one natural idea is to have the entire client population
performs the same number of local updates and mixes these updates carefully to ensure they are weighted equally.
Unfortunately,
when clients are available only intermittently,
they will miss some rounds.
A naive approach to equalizing the number of local updates is to have clients catch up by performing their missed local computations immediately when they become available. 
However, this approach requires a daunting amount of resources and may not be possible due to hardware/software constraints. 
Formally, recall that $\calA^t$ is the set of available clients at time $t$.  
Let $\tau_i(t) \triangleq \{t^{\prime}:~ t^{\prime}<t ~\text{and}~ i\in \calA^{t^{\prime}}\}$ denote the most recent (with respect to time $t$) round that client $i$ is available.
Compared with standard~\FedAvg, the naive ``catch-up'' procedure will consume $\pth{t-\tau_i(t)-1} \cdot s$ local stochastic gradient descent updates and $(t - \tau_i(t)-1)$ additional stochastic samples, where $s$ is the number of local updates per round when a client is available in standard~\FedAvg.    

In this work, we target computation-light algorithms that,  compared with~\FedAvg, only take $O(1)$ additional computation without additional stochastic samples. 
We propose {\bf Fed}erated {\bf A}gile {\bf W}eight Re-{\bf E}qualization
 (\FedAPM), which is formally described in~\prettyref{alg: fedpbc+}.
 It involves %
two novel algorithmic structures:
 {\em adaptive innovation echoing} and {\em implicit gossiping}. 
At a high level,
these novel algorithmic structures
(i) help clients catch up on the missed computation,  and 
(ii) simultaneously enable a balanced information mixture through implicit client-client gossip,  
which ultimately corrects the remaining bias.

\begin{figure}[!t]
\centering
\begin{minipage}{.9\textwidth}
    \setlength{\columnsep}{1.8cm}
\begin{algorithm}[H]
\caption{\FedAPM}
\label{alg: fedpbc+}
\textbf{Inputs:} 
$T$, 
$s$, 
$\eta_l$, 
$\eta_g$, 
$\x^0$. 

\lFor{$i\in [m]$}
{
$\x_i^{0} \gets \x^0$ and $\tau_i (0) \gets -1$ 
}

\vspace*{.1\baselineskip}
\For{$t=0, \cdots, T-1$}
{
\begin{multicols}{2}
\For{$i\in \calA^t$}
{
$\x_i^{(t,0)} \gets \x_i^{t}$\;
\For{$k=0, \cdots, s-1$}
{
\hspace{-1em}
$\small \x_i^{(t, k+1)} \gets $ 

$~~\x_i^{(t, k)} -  \eta_l \nabla \ell_i(\x_i^{(t, k)}; \xi_{i}^{(t,k)});$
} 
$\bm{G}_i^{t} \gets \x_i^{t} - \x_i^{(t, s)}$\;
$\x_i^{t\dagger} \gets \x_i^{(t, 0)} - \eta_g
(t - \tau_i(t)) \bm{G}_i^t$\;
$\tau_i(t+1) \gets t$\;
Report $\x_i^{t\dagger}$ to the parameter server\; 
} 
\columnbreak
$\x^{t+1} \gets \frac{1}{\abth{\calA^t}} 
\sum_{i\in \calA^t}\x_i^{t\dagger}
$\;
\For{$i\in [m]$}{
\uIf{$i\in \calA^t$}{$\x_i^{t+1} \gets \x^{t+1}$\;}
\uElseIf{$i\notin \calA^t$}{$\x_i^{t+1} \gets \x_i^{t}$\; $\tau_i(t+1) \gets \tau_i(t)$\;}}
\end{multicols}
}
\end{algorithm}
    
\end{minipage}
\vspace*{-\baselineskip}
\end{figure}

In~\prettyref{alg: fedpbc+}, each client keeps two local variables $\x_i$ and $\tau_i$,  along with a few auxiliary variables used in updating $\x_i$ and $\tau_i$. 
The algorithm inputs are rather standard: total training rounds $T$, local and global learning rates $\eta_{l}$ and $\eta_g$, the number of local updates per round $s$, and the initial model $\x^0$.  
In each round $t$, similar to~\FedAvg, an available client $i\in \calA^t$ 
performs $s$ steps of stochastic gradient descent
on its local model $\x_i^t$ (lines 5-8), where $\nabla \ell_i (\cdot; \xi_i^{(t,k)})$ is the stochastic gradient of sample $\xi_i^{(t,k)}$. 
Next, we describe the two novel algorithmic structures used in~\FedAPM.

{\bf Adaptive innovation echoing.}
Departing from~\FedAvg~wherein the local estimate $\x_i^t$ is updated as \\
$\x_i^{t\dagger} \gets \x_i^{(t, 0)} - \eta_g
\bm{G}_i^t$.   
In \FedAPM~ (lines 10-11), we ``echo'' the local innovation $\bm{G}_i^t$ by multiplying it by $(t - \tau_i(t))$. 
Intuitively, this simple echoing helps us approximately equalize the number of local improvements, as formally stated in~\prettyref{prop: similar speed}. 
It says that the total numbers of innovations echoing are the same for all active clients for any given round and allows the unavailable clients to catch up to the missed computations when they become available. 
\begin{proposition}
\label{prop: similar speed}
If 
$\indc{i \in \calA^{R-1}} = 1$, it holds that
$
\sum_{t=0}^{R-1} \indc{i \in \calA^t} \pth{t - \tau_i(t)} = R,~\forall~ R \ge 1.
$
\end{proposition}

\begin{wrapfigure}[5]{r}{0.42\textwidth}
    \vspace{-1.5\baselineskip}
    \begin{align}
    W_{ij}^{(t)} \triangleq 
    \begin{cases}
        \frac{1}{\abth{\calA^t}},~\text{if } i,j \in \calA^t ;&\\
        1,~\text{if } i=j ~\text{and } i \notin \calA^t ;&\\
        0,~\text{otherwise.}&
    \end{cases}
    \label{eq: W matrix elements}
    \end{align}
\end{wrapfigure}
 
\textbf{Implicit gossiping.}
In~\FedAPM, the parameter server does not send the most recent global model to the active clients at the beginning of a round. Instead, the parameter server aggregates the locally updated models $\x_i^{t\dagger}$ and sends the new global model $\x^{t+1}$ to all active clients $\calA^t$ (lines 14-15).  
By postponing multicasting the shared global model, 
the active clients in $\calA^t$ {\em implicitly gossip} 
their updated local models with each other through the parameter server \cite{xiang2023towards}.
Though the postponed multi-cast brings in staleness, simple coupling argument show that the staleness is bounded (\prettyref{lmm: geo second moment main text}). 
In addition,
our empirical results (\prettyref{tab: slowdown supp} in~\prettyref{app: numerical}) suggest that there is no significant slowdown when compared to vanilla~\FedAvg.
Gossip-type algorithms were originally proposed for peer-to-peer networks and are well-known for their agility to communication failures and asynchronous information exchange in achieving average consensus  
\cite{degroot1974reaching,boyd2006randomized,kempe2003gossip,Hajnal58,Lynch:1996:DA:2821576,nedic2009distributed}. 
Intuitively, the clients' local estimates are eventually equally weighted in the 
final algorithm output.   
Note that,
departing from the standard gossiping protocols therein \cite{kempe2003gossip,shah2009gossip},
information exchange in~\FedAPM~does not involve client-client communication.
The information mixing matrix under~\FedAPM~is defined in~\eqref{eq: W matrix elements}, which is doubly stochastic. 
Let $M^{(t)} \triangleq \mathbb{E}[(W^{(t)})^2]$,
$\rho (t) \triangleq \lambda_2 (M^{(t)})$,
$\allones = \Indc \Indc^{\top} / m$,
and $\rho \triangleq \max_t \rho(t)$,
where $\lambda_2(\cdot)$ denotes the second largest eigenvalue. 
We next characterize the information mixing error,
\ie, consensus error in~\prettyref{lmm: spectral norm}. 
\begin{lemma}[\cite{nedic2017achieving,nedic2018network,wang2021cooperative}]
\label{lmm: spectral norm}
For any matrix $B \in \reals^{d \times m}$, it holds that
$
\mathbb{E}_{W}[\fnorm{B \pth{\prod_{r=1}^{t} W^{(r)} - \allones}}^2] \le \rho^t \fnorm{B}^2,
$
where the expectation is taken with respect to randomness in $W$ matrices.    
\end{lemma}

\section{Convergence Analysis}
\label{sec: convergence analysis}

In this section,
we analyze the convergence of~\FedAPM.
All missing proofs and intermediate results are deferred to the Appendix.
Details can be found in~\hyperref[toc]{Table of Contents}.
\subsection{Assumptions}
We start by stating regulatory assumptions that are common in federated learning analysis \cite{li2020federated, wang2020tackling, karimireddy2020scaffold}.

\begin{assumption}
\label{ass: 2 smmothness}
Each local objective function
$\nabla F_{i}(\x)$ is $L$-Lipschitz,
\ie,
\[
    \norm{\nabla F_{i}(\x_1)-\nabla F_{i}(\x_2)}\le L \norm{\x_1-\x_2},
    ~\forall \x_1,~\x_2,~\text{and}~\forall~i\in [m].
\] 
\end{assumption}
\begin{assumption}
\label{ass: bounded variance client-wise}
Stochastic gradients $\nabla \ell_i(\x;\xi)$ are unbiased with bounded variance, \ie,  
\[  
    \expect{\nabla \ell_i(\x;\xi) \mid \x}=\nabla F_i(\x)
    ~\text{and}~
    \expect{\norm{\nabla \ell_i(\x;\xi)-\nabla F_i(\x)}^2 \mid \x} \le \sigma^2,
    ~\forall~
    i\in[m].
\] 
\end{assumption}
\begin{assumption}
\label{ass: bounded similarity}
The divergence between local and global gradients is bounded
for $\beta,~\zeta\ge 0$
such that 
\begin{align}
\label{eq: BG condition}
    \frac{1}{m}\sum_{i=1}^m \norm{\nabla F_i(\x)- \nabla F(\x)}^2 \le \beta^2 \norm{\nabla F(\x)}^2+ \zeta^2. 
\end{align} 
\end{assumption}
\vspace*{-\baselineskip}
When the local data sets are homogeneous, 
$\nabla F_i(\x) = \nabla F(\x)$ holds for any client $i \in [m]$, 
resulting in $\beta = \zeta = 0$. 
\prettyref{ass: bounded similarity} and its variants in~\prettyref{tab: limitation of existing work}
are often referred to as bounded gradient dissimilarity assumption to account for data heterogeneity across clients.  
It can be easily checked that our~\prettyref{ass: bounded similarity} is more relaxed or equivalent to the variants therein. 
\begin{table}[!h]
    \centering
    \vspace*{-\baselineskip}
    \caption{\small Popular variant assumptions on gradient dissimilarity.}
    \label{tab: limitation of existing work}
    \centering
    \begin{tabular}{ c  c }
    \toprule
    {\bf \centering Bounded Gradient Dissimilarity} &  
    {\bf \centering References}  \\
    \midrule     
     $\max_{\x}\norm{\nabla F_i(\x)}^2 \le \zeta^2, ~ \forall ~i\in [m]$ & \cite{Li2020,yu2019parallel,cho2023communication,cho2022towards,yan2023federated} \\ 
    \midrule 
    $ \frac{1}{m}\sum_{i=1}^m \norm{\nabla F_i(\x)}^2 \le \beta^2 \norm{\nabla F(\x)}^2 $ & \cite{li2020federated,li2019feddane} \\ 
    \midrule
    $ \frac{1}{m}\sum_{i=1}^m \norm{\nabla F_i(\x) - \nabla F(\x)}^2 \le \zeta^2$ & \cite{wang2022matcha,yu2019linear,huang2022lower,wang2019adaptive,allouah2023fixing,karimireddybyzantine22,wang2022,yang2022anarchic} \\
    \midrule
    $\frac{1}{m}\sum_{i=1}^m \norm{\nabla F_i(\x)}^2 \le \beta^2 \norm{\nabla F(\x)}^2 + \zeta^2$ & \cite{karimireddy2020scaffold,yuan2022,wang2020tackling,wang2021cooperative,gu2021fast} \\
    \bottomrule
    \end{tabular}
    \vspace*{-.5\baselineskip}
\end{table}

\subsection{Auxiliary/Imaginary update sequence construction.}
Directly analyzing the evolution of $\x^t$ and $\x_i^t$ is challenging 
due to the fact that different clients update at different rounds, 
and that different active clients echo their 
local innovation $\bm{G}_i^t$ (line 9 in~\prettyref{alg: fedpbc+}) with different strength $(t-\tau_i)$. 
As such, 
we construct an auxiliary/imaginary update sequence $\bm{z}_i^t$ for client $i\in [m]$,
whose evolution is closely coupled with $\x^t$ and $\x_i^t$
but is easier to analyze.  
Note that
the auxiliary/imaginary update sequence is never actually computed by clients
but acts as a necessary tool in building up the analysis.
\begin{definition}
\label{def: auxiliary sequence}
{\em 
The auxiliary sequence $\{\bz_i^t\}$ of client $i\in [m]$ is defined as 
\begin{align}
\label{eq: auxiliary definition form}
    \bz_i^t ~ \triangleq ~ 
    \x_i^t - \eta_l \eta_g s (t-\tau_i(t) - 1)
        \nabla F_i(\x_i^{\tau_{i}(t)+1}),~\forall~ i\in [m].
\end{align}    
}

\end{definition}
Recall that $\tau_i(0) = -1$. Thus, by definition, $\bz_i^0 = \x_i^0$ according to \eqref{eq: auxiliary definition form}. 
For general $t$, 
when client $i \in \calA^{t-1}$,
we simply have $\tau_i(t) = t-1$
and thus $t - 1 -\tau_i(t) = t- 1 - (t-1) = 0$.
That is,
the auxiliary model $\bm{z}_i^t$ and the real model $\x_i^t$ are {\em identical}
whenever the client $i$ becomes available in the previous round.  
\begin{itemize}[leftmargin=*]
    \item 
  When 
    $i\in \calA^{t-1}$, the iterate of $\bm{z}_i$ is a bit more involved: 
    \begin{align}
    \bz_i^t \overset{(\ref{eq: auxiliary active line 2}.a)}{=} 
    \x_i^t 
    &\overset{(\ref{eq: auxiliary active line 2}.b)}{=} 
    \frac{\sum_{j\in \calA^{t-1}}}{|\calA^{t-1}|}
    \pth{
    \bz_j^{t-1}
    +
    \underbrace{(\x_j^{t-1} - \bz_j^{t-1})}_{(\ref{eq: auxiliary active line 2}.c)}
    - 
    \eta_l
    \eta_g (t - 1 - \tau_{j}(t-1) )
    \bm{G}_j^{t-1}}
    ,
    \label{eq: auxiliary active line 2}
    \end{align}    
    where 
    $(\ref{eq: auxiliary active line 2}.a)$ holds because of~\prettyref{def: auxiliary sequence} and $i \in \calA^{t-1}$,
    $(\ref{eq: auxiliary active line 2}.b)$ because of line 10 in~\prettyref{alg: fedpbc+},
    addition and subtraction.
    $(\ref{eq: auxiliary active line 2}.c)$ can be expanded by~\eqref{eq: auxiliary definition form}.
    We defer the simplified form of \eqref{eq: auxiliary active line 2} to~\eqref{eq: z iterate last active} in~\prettyref{app: nomenclature} for a tidy presentation.
    \item
    When 
    $i \notin \calA^{t-1}$, $\bm{z}_i^t$ has a simple iterative relation:
    \begin{align}
    \label{eq: auxiliary inactive update}
    \bz_i^t ~ = ~ \bz_i^{t-1}  - \eta_l\eta_g s \nabla F_i(\x_i^{\tau_i(t-1)+1}).      
    \end{align}
At a high level, 
the sequence $\bm{z}_i^t$ 
approximately
mimics the ideal descent evolution at a client as if the client performs local optimizations on its local model $\x_i$ per round regardless of its availability. 
Mathematically,
the idea is that, 
if the progress per iteration of the auxiliary sequence $\bz_i^t$ is bounded,
we can show the convergence of $\x_i^t$
when $\x_i^t$ and $\bz_i^t$ are close to each other.
\end{itemize}

It is worth noting that auxiliary sequences are used in peer-to-peer distributed learning literature \cite{spiridonoff2020robust,avdiukhin2021federated,lian2017can,yuan2016convergence,stich2018local,nedic2018network}. 
Yet, existing constructions are not applicable to our problem due to
(1) the non-convexity of the global objectives,  
(2) multiple local updates per round, 
(3) possibly unbounded gradients, and 
(4) the general form of bounded gradient dissimilarity. 
Departing from the use of staled stochastic gradients for auxiliary updates therein,
we adopt the true gradient $\nabla F_i(\cdot)$ 
to avoid the complications from the involved interplay between randomness in stochastic samples
and randomness in $\tau_i(t)$.
On the technical front, it follows from~\prettyref{def: auxiliary sequence} that
 $\|\x_i^t - \bz_i^t\|_2^2 \le \eta_l^2 \eta_g^2 s^2 (t - \tau_i(t)-1)^2 \|\nabla F_i(\x_i^{\tau_{i}(t)+1})\|_2^2$, 
 whose bound appears to be quite challenging to derive
 due to the coupling of different realizations of $\tau_i(t)$ and gradients.
 As such, we bound the 
 average of $\|\x_i^t - \bz_i^t\|^2 $ across clients and rounds in~\prettyref{prop: client dis}.

\begin{lemma}[Unavailability statistics]
\label{lmm: geo second moment main text}
Under~\prettyref{ass: prob lower bound}
and $\delta$ defined therein.
It holds for $t\ge 0$ that
$\expect{t - \tau_i(t)} \le 1/\delta$
and
$\expect{\pth{t - \tau_i(t)}^2} \le 2/\delta^2.$
\end{lemma}
\prettyref{lmm: geo second moment main text} yields an upper bound on the first and second moments of a client $i$'s unavailable duration
despite the unstructured nature of clients' non-stationary and heterogeneous unavailability.
In the special case where we have clients available with the same probability $\delta$,
the duration simply follows a homogeneous geometric distribution.
It can be easily checked that our bounds trivially hold.
However,
the duration becomes a more challenging {\em non-homogeneous} geometric random variable
under our non-stationary unavailability dynamics.
\prettyref{lmm: geo second moment main text}
can be derived by using a simple coupling argument and by using tools from probability theory \cite{gut2006probability}.

\subsection{Main results.}
Let $\bar{\bm{z}}_t \triangleq \frac{1}{m} \sum_{i=1}^m \bm{z}_i^t$,
$F^\star \triangleq \min_{\x} F(\x)$,
and $\delta_{\max} \triangleq \max_{i\in [m], t \in [T]} p_i^t$.
\begin{lemma}[Descent Lemma]
\label{lmm: descent lemma}
Let $\calF^t$ define the sigma algebra generated by randomness up to round $t$.
Suppose Assumptions~\ref{ass: 2 smmothness},
\ref{ass: bounded variance client-wise} hold and 
$\eta_l \eta_g \le 9 / (100 s L)$,
it holds that
\begin{align*}
\expect{F(\bar{\bz}^{t+1}) - F(\bar{\bz}^t)~|~\calF^t}
&\le
- \frac{\eta_l \eta_g s}{4} \norm{\nabla F(\bar{\bz}^t)}^2 \\
&~~~+\frac{
2 \eta_l \eta_g s L \sigma^2
\pth{\eta_l \eta_g \delta_{\max} + 4.5 m \eta_l^2 s L}}
{m^2} 
\sum_{i=1}^m
(t - \tau_i(t))^2
\\
&~~~
+
\frac{35 \eta_g \eta_l^3 s^3 L^2}{m} 
\sum_{i=1}^m 
(t - \tau_i(t))^2
\norm{\nabla F_i(\x_i^{\tau_i(t)+1})}^2 \\
&~~~+ 
\frac{2.2 \eta_l \eta_g s L^2 }{m}
\sum_{i=1}^m
\underbrace{\norm{\x_i^t - \bz_i^t}^2}_{\textnormal{Approximation Error}}
+
\frac{\eta_l \eta_g s L^2}{2m}
\sum_{i=1}^m
\underbrace{\norm{\bz_i^t - \bar{\bz}^t}^2}_{\textnormal{Consensus Error}}
.
\end{align*}
\end{lemma}

The proof of~\prettyref{lmm: descent lemma} follows from the standard analysis for non-convex smooth objectives but
with non-trivial adaptation to account for
{\em adaptive innovation echoing}
and {\em implicit gossiping}.
In particular,
it highlights two terms unique in our derivation:
the approximation error from the auxiliary sequence
and the consensus error from the implicit gossiping procedure.

\begin{proposition}[Approximation error]
\label{prop: client dis}
Given Assumptions \ref{ass: 2 smmothness} and \ref{ass: bounded similarity},
it holds that
\begin{small}
\begin{align}
\frac{1}{m T}\sum_{t=0}^{T-1}
\sum_{i=1}^m
\expect{\norm{\x_i^t - \bz_i^t}^2} %
\le \frac{6 \eta_l^2 \eta_g^2 s^2 }{\delta^2 }
\pth{\beta^2 + 1}
\frac{1}{T}
&\sum_{t=0}^{T-1} 
\expect{
\norm{\nabla F(\bar{\bz}^{t})}^2} 
+
\frac{6 \eta_l^2 \eta_g^2 s^2 }{\delta^2 }
\zeta^2 
\nonumber
\\
+
\frac{6 L^2 \eta_l^2 \eta_g^2 s^2 }{\delta^2}
\frac{1}{m}
\sum_{i=1}^m
\frac{1}{T}
&\sum_{t=0}^{T-1} 
\expect{\norm{\bz_i^{t} - \bar{\bz}^t}^2}
.
\label{eq: approx error}
\end{align}    
\end{small}
\end{proposition}
The proof of~\prettyref{prop: client dis} starts from~\prettyref{def: auxiliary sequence}.
Although in general it is difficult to bound the error,
Assumptions~\ref{ass: 2 smmothness} and~\ref{ass: bounded similarity} 
allow us to break down the problem into 
bounding the averaged gradient norm of $\bar{\bz}^t$ and the consensus error over all randomness instead.
Next,
we analyze the consensus error.
Note that
although implicit gossiping
takes place in~\prettyref{alg: fedpbc+} for $\x_i^t$,
its analysis is technically challenging
as discussed before.
So, 
we adopt the auxiliary $\bz_i^t$ as an intermediary and apply Young's inequality to bound the actual consensus error.
Details will be specified next.
Formally,
the auxiliary models can be expressed in a compact matrix form as 
$\bm{Z}^{(t)} \triangleq [\bz_1^t, \ldots, \bz_m^t]$.
Their local parameter innovation matrix ${ \tilde{{\bm{G}}^t}}$ 
is formulated by combing~\eqref{eq: auxiliary active line 2} and~\eqref{eq: auxiliary inactive update}.
We refer the interested readers to~\eqref{eq: auxiliary update detail} in~\prettyref{app: nomenclature} for the exact formula.
Unrolling the recursion,
the consensus error
can be expanded as 
\begin{align}
\frac{1}{m}{\lnorm{\pth{\bm{Z}^{(t-1)} - \eta_l \eta_g \tilde{\bm{G}}^{(t-1)}} W^{(t-1)} \pth{\identity - \allones}}{\rm F}^2} 
\label{eq: consensus derivation}
\overset{(\ref{eq: consensus derivation}.a)}{=} 
\frac{\eta_l^2 \eta_g^2}{m} 
\lnorm{\sum_{q=0}^{t-1} \tilde{\bm{G}}^{(q)} \pth{\prod_{l=q}^{t-1} W^{(q)} - \allones}}{\rm F}^2,
\end{align}
where equality $(\ref{eq: consensus derivation}.a)$ holds because all clients are initiated at the same weight.
\begin{lemma}[\cite{xiang2023towards}]%
\label{lmm: rho upper bound main text}
Under~\prettyref{ass: prob lower bound}, it holds that
$
    \rho \le 1 - \frac{\delta^4 (1 - (1-\delta)^m)^2}{8}.
$
\end{lemma}
Recall that $\rho$ bounds the expected spectral norm of the information mixing matrix $\Wt{t}$.
It is important to have $\rho < 1$
for an exponential decay of the consensus error (see~\prettyref{lmm: spectral norm}).
We now proceed to present the convergence rates.
In the sequel, we assume
it holds for
$\eta_g$ and 
$\eta_l$ that
\begin{small}
\begin{align}
\label{eq: lr condition main text}
    \eta_l \eta_g &\le 
    \frac{\pth{1 - \sqrt{\rho}}\delta}{80 s (L+1) \pth{\sqrt{\rho} + 1}  \sqrt{
    \pth{\beta^2 + 1}(1 + L^2)}}; ~
    \eta_l \le
    \frac{\delta}{200 s L \sqrt{
    \pth{\beta^2 + 1}
    (1 + L^2)}}.
\end{align}
\end{small}

The proof of the consensus error
borrows insights from the analysis of the gossip algorithm~\cite{nedic2017achieving,wang2022matcha}
but with substantial adaptation to
accommodate the novel auxiliary formulation and multi-step local updates.
Under the learning rate conidtions in~\eqref{eq: lr condition main text} and Assumptions~\ref{ass: prob lower bound}, \ref{ass: 2 smmothness}, \ref{ass: bounded variance client-wise} and \ref{ass: bounded similarity},
we can show that
\begin{small}
\begin{align}
    \frac{1}{mT} \sum_{t=0}^{T-1}\sum_{i=1}^m 
    \expect{\norm{\x_i^t - \bz_i^t}^2} 
    &\asymp
    \frac{1}{mT} \sum_{t=0}^{T-1}\sum_{i=1}^m 
    \expect{\norm{\bz_i^t - \bar{\bz}^t}^2}
    \asymp
    \frac{1}{T}\sum_{t=0}^{T-1}\expect{\norm{\nabla F(\bar{\bz}^t)}^2}.
    \label{eq: approx and consensus for z}
\end{align}
\end{small}

It remains to bound the full convergence error of $\bz_i^t$,
which is presented in~\prettyref{thm: z bar rate}.
\begin{theorem}[Convergence error of $\bz_i^t$]
\label{thm: z bar rate}
Suppose that Assumptions~\ref{ass: prob lower bound}, \ref{ass: 2 smmothness}, \ref{ass: bounded variance client-wise} and \ref{ass: bounded similarity} hold.
Choose learning rates $\eta_l$ and $\eta_g$
such that the conditions in~\eqref{eq: lr condition main text} are met
for $T\ge 1$,
it holds that
\begin{align}
    \frac{1}{T}\sum_{t=0}^{T-1}\expect{\norm{\nabla F(\bar{\bz}^t)}^2}
    &\lesssim
   \frac{\pth{F(\bar{\bz}^0) - F^\star}}{\eta_l \eta_g s T}
    +\frac{\eta_l \eta_g L \sigma^2}{ m} \frac{\delta_{\max}}{\delta^2}
    +
    \eta_l^2 \eta_g^2 s^2 L^2
    \pth{\frac{\sigma^2 +  \zeta^2}{\delta^2 (1 - \sqrt{\rho})^2}}
    . 
    \label{eq: z bar rate}
\end{align}
\end{theorem}
\vspace{-.5\baselineskip}
By addition, subtraction, and Young's inequality,
\eqref{eq: x bar consensus} and~\eqref{eq: x bar relation} hold under~\prettyref{ass: 2 smmothness}.
\begin{small}
\begin{align}
\frac{1}{T} \sum_{t=0}^{T-1}
\frac{1}{m} \sum_{i=1}^m 
\expect{\norm{\x_i^t - \bar{\x}^t}^2}
&\asymp
\frac{1}{T} \sum_{t=0}^{T-1}
\frac{1}{m} \sum_{i=1}^m 
\expect{\norm{\x_i^t - \bz_i^t}^2}
+
\frac{1}{T} \sum_{t=0}^{T-1}
\frac{1}{m} \sum_{i=1}^m 
\expect{\norm{\bz_i^t - \bar{\bz}^t}^2}
;
\label{eq: x bar consensus} \\
\frac{1}{T}\sum_{t=0}^{T-1}\expect{\norm{\nabla F(\bar{\x}^t)}^2}
&\asymp
\frac{1}{T}\sum_{t=0}^{T-1} \frac{1}{m} \sum_{i=1}^m
\expect{\norm{\x_i^t - \bz_i^t}^2
}
+
\frac{1}{T}\sum_{t=0}^{T-1} \expect{\norm{\nabla F(\bar{\bz}^t)}^2}.
\label{eq: x bar relation}
\end{align}
\end{small}
Moreover, 
from~\eqref{eq: approx and consensus for z},~\eqref{eq: x bar consensus} and~\eqref{eq: x bar relation},
it can be seen that~\eqref{eq: x consensus} holds.
\begin{small}   
\begin{align}
    \frac{1}{mT} \sum_{t=0}^{T-1}\sum_{i=1}^m 
    \expect{\norm{\x_i^t - \bar{\x}^t}^2} 
    &\asymp
    \frac{1}{T}\sum_{t=0}^{T-1}\expect{\norm{\nabla F(\bar{\bz}^t)}^2}
    \asymp
    \frac{1}{T}\sum_{t=0}^{T-1}\expect{\norm{\nabla F(\bar{\x}^t)}^2}
    .
    \label{eq: x consensus}
\end{align}
\end{small}
Combining~\eqref{eq: approx and consensus for z},
\eqref{eq: z bar rate},
\eqref{eq: x bar consensus}
and \eqref{eq: x bar relation},
we are ready for~\prettyref{cor: x bar rate}.
\begin{corollary}[Convergence rate of $\x_i^t$]
\label{cor: x bar rate}
Suppose that Assumptions \ref{ass: prob lower bound}, \ref{ass: 2 smmothness}, \ref{ass: bounded variance client-wise} and \ref{ass: bounded similarity} hold.
Choose learning rates as
$\eta_l = \frac{1}{\sqrt{T} s L}$,
$\eta_g = \sqrt{s \delta m}$
such that the conditions in~\eqref{eq: lr condition main text} are met for $T\ge 1$,
it holds that
\begin{align}
    \frac{1}{T}\sum_{t=0}^{T-1}\expect{\norm{\nabla F(\bar{\x}^t)}^2}
    &\lesssim
    \frac{L \pth{F(\bar{\x}^0) - F^\star}}{\sqrt{s \delta m T}} 
    +
    \frac{\delta_{\max} }{\delta^{\frac{3}{2}} \sqrt{s m T}} 
    \sigma^2
    +
    \frac{s m }{T}
    \pth{\frac{\sigma^2 + \zeta^2 }{\delta (1 - \sqrt{\rho})^2}}
    . 
    \label{eq: x bar rate}
\end{align}
\end{corollary}

\prettyref{cor: x bar rate} establishes the full convergence rate for~\FedAPM~algorithm.
It can be seen that
the first and second terms dominate when $T$ is sufficiently large,
which relate to
initial suboptimality gap and
stochastic gradient noise $\sigma^2$, respectively.
The non-stationary client unavailability results in
the third term, 
which relates to gradient divergence $\zeta^2$ and also to $\sigma^2$. 
The proof of~\prettyref{cor: x bar rate} follows from~\eqref{eq: x bar relation}
by plugging in~\prettyref{prop: client dis} and~\prettyref{thm: z bar rate}.
In the special case
where $k$ clients participate 
uniformly at random,
we simply have $\delta_{\max} = \delta = k / m$.
Our convergence bound 
attains the rate of
$O(1/\sqrt{s k T})$.
In other words,
we achieve the desired linear speedup property with respect to the number of local steps $s$ and the number of active clients $k$,
matching the established literature \cite{yang2021achieving,wang2022,yu2019linear,yu2019parallel}. 
The linear speedup property enables a large cross-device federated learning system to take advantage of a massive scale of parallelism.
Notice that 
the consensus error~\eqref{eq: x consensus} and the convergence rate~\eqref{eq: x bar rate} have the same asymptotic order with respect to the parameters therein.
Hence,
the consensus error also enjoys the desired linear speedup property when $T$ is sufficiently large.

\newcommand{\FedAU}{\texttt{FedAU}}
\newcommand{\FAST}{\texttt{F3AST}}
\section{Numerical Experiments}
\label{sec: numerical}
{\bf Overview.}
In this section, we 
evaluate~\FedAPM~on
real-world data sets to corroborate our analysis and compare it with the other state-of-the-art algorithms.
The missing specifications and additional results can be found in~\prettyref{app: numerical}.
Specifically,
we consider a federated learning system of one parameter server and $m = 100$ clients,
wherein clients become available intermittently.
The image classification tasks 
use CNNs and are based on
SVHN \cite{netzer2011readingdigits}, CIFAR-10 \cite{krizhevsky2009learning} and CINIC-10 \cite{darlow2018cinic} data sets.
All of them include $10$ classes of images of different categories.
To emulate a highly heterogeneous local data distribution,
the image class distribution $\nu_i \sim \mathsf{Dirichlet}(\alpha = 0.1)$
at client $i$ \cite{hsu2019measuring,wang2022,wang2023lightweight}.

\begin{table}[!t]
\centering
\caption{\footnotesize 
Results and comparisons on real-world datasets in the form of mean accuracy $\pm$ standard deviation and are obtained over 3 repetitions in different random seeds.
Results are averaged over the last $50$ rounds.
The total number of global rounds is 2000 
for SVHN, CIFAR-10 and CINIC-10.
Algorithms are categorized into two groups:
(1) ones {\bf not} aided by memory or known statistics;
(2) ones assisted by memory or known statistics.
For a fair competition,
we {\bf boldface} the best accuracy in the first group,
while the second best is \underline{underlined}.
}
\label{tab: exp main text}
\resizebox{\linewidth}{!}{
\begin{footnotesize}
\begin{tabular}{c|c|p{1.7cm} p{1.7cm}|p{1.7cm} p{1.7cm}|p{1.7cm} p{1.7cm}}
    \toprule
    \multirow{2}{*}{\begin{tabular}{@{}c@{}}{\bf Unavailable} \\ {\bf Dynamics} \end{tabular}} &
    {\bf Datasets} &
    \multicolumn{2}{c|}{\bf SVHN} & 
    \multicolumn{2}{c|}{\bf CIFAR-10} & 
    \multicolumn{2}{c}{\bf CINIC-10} \\
    \cline{2-8}
    & 
    {\bf Algorithms}&
    \multicolumn{1}{c}{\bf Train} &
    \multicolumn{1}{c|}{\bf Test} &
    \multicolumn{1}{c}{\bf Train} &
    \multicolumn{1}{c|}{\bf Test} &
    \multicolumn{1}{c}{\bf Train} &
    \multicolumn{1}{c}{\bf Test} \\
    \midrule
     \multirow{6}{*}{
\begin{tabular}{@{}c@{}} 
{Stationary} \\[1em]
\adjustbox{width=0.15\linewidth}{\includestandalone{elements/figure/table_fig/updated/flat_dynamic}}

\end{tabular}} & 
\FedAPM~({\bf ours}) & 
{\bf 86.5} $\pm$ 0.7 \%&
{\bf 86.1} $\pm$ 0.7 \%&
{\bf 68.1} $\pm$ 1.4 \%&
{\bf 66.3} $\pm$ 1.1 \%&
{\bf 47.9} $\pm$ 2.1 \%&
{\bf 47.3} $\pm$ 2.0 \%
\\
& 
\FedAvg~over {\em active} & 
82.6  $\pm$ 1.0 \%&
82.4  $\pm$ 1.1 \%&
64.1  $\pm$ 1.9 \%&
62.9  $\pm$ 1.4 \%&
43.6  $\pm$ 2.4 \%&
43.1  $\pm$ 2.4 \%
\\
& 
\FedAvg~over {\em all} & 
76.1 $\pm$ 2.1 \%&
76.1 $\pm$ 2.4 \%&
55.8 $\pm$ 2.1 \%&
55.4 $\pm$ 1.8 \%&
38.4 $\pm$ 2.1 \%&
38.0 $\pm$ 2.1 \%
\\
& 
\FedAU & 
\underline{83.4} $\pm$ 1.0 \%&
\underline{83.2} $\pm$ 1.0 \%&
\underline{65.4} $\pm$ 1.4 \%&
\underline{64.1} $\pm$ 1.0 \%&
\underline{45.6}  $\pm$ 1.5 \%&
\underline{45.2}  $\pm$ 1.5 \%
\\
& 
\FAST & 
83.2  $\pm$ 0.7 \%&
83.2  $\pm$ 0.7 \%&
64.4 $\pm$ 1.1 \%&
63.5  $\pm$ 0.9 \%&
45.3  $\pm$ 1.2 \%& 
44.8  $\pm$ 1.2 \%
\\
\noalign{\vspace{.5mm}}
\cline{2-8}
\noalign{\vspace{.5mm}}
& 
\FedAvg~with {\em known} $p_i$'s  & 
86.1  $\pm$ 0.5 \%&
85.6  $\pm$ 0.5 \%&
65.4  $\pm$ 1.0 \%&
63.1  $\pm$ 0.9 \%&
45.0  $\pm$ 1.2 \%&
44.6  $\pm$ 1.1 \%
\\
& 
\MIFA~({\em memory aided}) & 
84.2  $\pm$ 0.5 \%&
84.1  $\pm$ 0.6 \%&
66.6  $\pm$ 0.8 \%&
65.3  $\pm$ 0.6 \% &
47.5  $\pm$ 0.5 \%&
46.9  $\pm$ 0.5 \%
\\
& 
\FedVARP~({\em memory aided}) & 
84.6  $\pm$ 0.2 \%&
84.3  $\pm$ 0.1 \%&
67.5  $\pm$ 0.2 \%&
66.3  $\pm$ 0.3 \% &
47.8  $\pm$ 0.2 \%&
47.2  $\pm$ 0.2 \% \\
    
    \toprule
    \multirow{6}{*}{
\begin{tabular}{@{}c@{}} 
{\bf Non}-stationary \\ 
({\bf Staircase}) \\
[1em]
\adjustbox{width=0.15\linewidth}{\includestandalone{elements/figure/table_fig/updated/staircase_dynamic}} 
\end{tabular}}& 
\FedAPM~({\bf ours})& 
{\bf 85.9} $\pm$ 0.8 \%&
{\bf 85.6} $\pm$ 1.0 \%&
{\bf 67.7} $\pm$ 1.3 \%&
{\bf 66.0} $\pm$ 1.2 \%&
{\bf 47.5} $\pm$ 2.0 \%&
{\bf 46.9} $\pm$ 2.0 \%
\\
& 
\FedAvg~over {\em active} & 
82.5  $\pm$ 1.0 \%&
82.4  $\pm$ 0.9 \%&
64.2  $\pm$ 1.8 \%&
63.0  $\pm$ 1.4 \%&
43.7  $\pm$ 2.0 \%&
42.3  $\pm$ 2.2 \%
\\
& 
\FedAvg~over {\em all} & 
75.9 $\pm$ 2.1 \%&
75.9 $\pm$ 2.3 \%&
55.7 $\pm$ 2.1 \%&
55.4 $\pm$ 1.8 \%&
38.4 $\pm$ 2.0 \%&
37.9 $\pm$ 2.0 \%
\\
& 
\FedAU & 
\underline{83.6} $\pm$ 0.8 \%&
\underline{83.4} $\pm$ 0.8 \%&
\underline{65.2} $\pm$ 1.7 \%&
\underline{63.9} $\pm$ 1.5 \%&
\underline{45.7}  $\pm$ 1.5 \%&
\underline{45.1}  $\pm$ 1.5 \%
\\
& 
\FAST &
83.1  $\pm$ 0.6 \%&
83.1  $\pm$ 0.6 \%&
64.3 $\pm$  1.1\%&
63.3  $\pm$ 0.9 \%&

45.2  $\pm$ 1.2 \%& 
44.8  $\pm$ 1.2 \%
\\
\noalign{\vspace{.5mm}}
\cline{2-8}
\noalign{\vspace{.5mm}}
& 
\FedAvg~with {\em known} $p_i^t$'s  & 
85.8  $\pm$ 0.8 \%&
85.2  $\pm$ 0.9 \%&
68.0  $\pm$ 1.6 \%&
66.1  $\pm$ 1.8 \%&
45.0  $\pm$ 1.1 \%&
44.7  $\pm$ 1.0 \%
\\
& 
\MIFA~({\em memory aided}) & 
84.2  $\pm$ 0.5 \%&
84.0  $\pm$ 0.5 \%&
66.7  $\pm$ 0.7 \%&
65.3  $\pm$ 0.5 \% &
47.5  $\pm$ 0.5 \%&
46.9  $\pm$ 0.5 \%
\\
& 
\FedVARP~({\em memory aided}) & 
84.6  $\pm$ 0.2 \%&
84.3  $\pm$ 0.3 \%&
67.3  $\pm$ 0.3 \%&
66.1  $\pm$ 0.3 \% &
47.7  $\pm$ 0.2 \%&
47.2  $\pm$ 0.1 \% \\
    
    \toprule
    
    \multirow{6}{*}{
\begin{tabular}{@{}c@{}} 
\addlinespace[1ex]
{\bf Non}-stationary  \\ 
({\bf Sine}) \\[1em]
\adjustbox{width=0.15\linewidth}{\includestandalone{elements/figure/table_fig/updated/sine_dynamic}}
\end{tabular}} & 
\FedAPM~({\bf ours})& 
{\bf 85.7} $\pm$ 0.9 \%&
{\bf 85.6} $\pm$ 0.9 \%&
{\bf 64.9} $\pm$ 1.9 \%&
{\bf 63.5} $\pm$ 2.0 \%&
{\bf 46.4} $\pm$ 2.4 \%&
{\bf 45.8} $\pm$ 2.4 \%
\\
& 
\FedAvg~over {\em active} & 
82.1  $\pm$ 1.1 \%&
82.0  $\pm$ 1.3 \%&
63.3  $\pm$ 1.9 \%&
62.1  $\pm$ 1.8 \%&
43.1  $\pm$ 2.5 \%&
42.6  $\pm$ 2.5 \%
\\
& 
\FedAvg~over {\em all} & 
71.3 $\pm$ 2.5 \%&
71.3 $\pm$ 2.8 \%&
52.2 $\pm$ 2.4 \%&
52.1 $\pm$ 2.2 \%&
36.4 $\pm$ 2.0 \%&
36.0 $\pm$ 1.9 \%
\\
& 
\FedAU & 
\underline{82.5} $\pm$ 1.4 \%&
\underline{82.5} $\pm$ 1.3 \%&
\underline{64.2} $\pm$ 2.3 \%&
\underline{63.0} $\pm$ 1.9 \%&
\underline{44.4}  $\pm$ 2.1 \%&
\underline{43.9}  $\pm$ 2.1 \%
\\
& 
\FAST & 
82.3  $\pm$ 1.0 \%&
82.3  $\pm$ 1.0 \%&
63.1 $\pm$ 1.7 \%&
62.3  $\pm$ 1.5 \%&
44.1  $\pm$ 1.6 \%& 
43.7  $\pm$ 1.6 \%
\\
\noalign{\vspace{.5mm}}
\cline{2-8}
\noalign{\vspace{.5mm}}
& 
\FedAvg~with {\em known} $p_i^t$'s  & 
86.3  $\pm$ 1.0 \%&
86.0  $\pm$ 1.0 \%&
69.1  $\pm$ 1.2 \%&
67.3  $\pm$ 1.3 \%&
47.9  $\pm$ 1.5 \%&
47.4  $\pm$ 1.1 \%
\\
& 
\MIFA~({\em memory aided}) & 
84.2  $\pm$ 0.4 \%&
84.1  $\pm$ 0.4 \%&
66.6  $\pm$ 0.8 \%&
65.5  $\pm$ 0.6 \% &
47.4  $\pm$ 0.5 \%&
46.9  $\pm$ 0.4 \%
\\
& 
\FedVARP~({\em memory aided}) & 
84.5  $\pm$ 0.2 \%&
84.3  $\pm$ 0.1 \%&
67.4  $\pm$ 0.2 \%&
66.0  $\pm$ 0.3 \% &
47.7  $\pm$ 0.1 \%&
47.1  $\pm$ 0.2 \% \\
    
    \toprule
    \multirow{6}{*}{
\begin{tabular}{@{}c@{}}
{\bf Non}-stationary \\
({\bf Interleaved Sine}) \\[1em]
\adjustbox{width=0.15\linewidth}{
\includestandalone{elements/figure/table_fig/updated/interleaved_sine}
}
\end{tabular}}& 
\FedAPM~({\bf ours})& 
{\bf 85.2} $\pm$ 1.6 \%&
{\bf 84.6} $\pm$ 1.6 \%&
{\bf 64.8} $\pm$ 3.1 \%&
{\bf 63.3} $\pm$ 2.7 \%&
{\bf 47.1} $\pm$ 2.7 \%&
{\bf 46.6} $\pm$ 2.7 \%
\\
& 
\FedAvg~over {\em active} & 
80.9  $\pm$ 1.7 \%&
80.7  $\pm$ 1.7 \%&
61.9  $\pm$ 2.4 \%&
60.7  $\pm$ 2.0 \%&
41.9  $\pm$ 2.7 \%&
41.5  $\pm$ 2.7 \%
\\
& 
\FedAvg~over {\em all} & 
69.5 $\pm$ 3.4 \%&
69.5 $\pm$ 4.1 \%&
51.3 $\pm$ 2.7 \%&
51.3 $\pm$ 2.7 \%&
35.9 $\pm$ 2.0 \%&
35.6 $\pm$ 2.0 \%
\\
& 
\FedAU & 
\underline{82.6} $\pm$ 1.3 \%&
\underline{82.4} $\pm$ 1.1 \%&
\underline{63.9} $\pm$ 2.2 \%&
\underline{62.8} $\pm$ 1.8 \%&
\underline{44.2}  $\pm$ 2.2 \%&
\underline{43.8}  $\pm$ 2.1 \%
\\
& 
\FAST &
81.3  $\pm$ 1.2 \%&
81.3  $\pm$ 1.2 \%&
62.2  $\pm$ 2.1 \%&
61.3  $\pm$ 1.7 \%&
43.1  $\pm$ 2.2 \%& 
42.7  $\pm$ 2.2 \%
\\
\noalign{\vspace{.5mm}}
\cline{2-8}
\noalign{\vspace{.5mm}}
& 
\FedAvg~with {\em known} $p_i^t$'s & 
85.8  $\pm$ 1.2 \%&
85.2  $\pm$ 1.3 \%&
68.7  $\pm$ 2.1 \%&
66.5  $\pm$ 2.4 \%&
47.2  $\pm$ 2.3 \%&
46.8  $\pm$ 2.2 \%
\\
& 
\MIFA~({\em memory aided}) & 
83.8  $\pm$ 0.9 \%&
83.7  $\pm$ 0.8 \%&
65.8  $\pm$ 1.9 \%&
64.6  $\pm$ 1.6 \% &
46.5  $\pm$ 1.8 \%&
45.9  $\pm$ 1.7 \%
\\
& 
\FedVARP~({\em memory aided}) & 
84.5  $\pm$ 0.3 \%&
84.1  $\pm$ 0.5 \%&
67.3  $\pm$ 0.3 \%&
65.7  $\pm$ 0.2 \% &
47.7  $\pm$ 0.5 \%&
47.2  $\pm$ 0.3 \% \\
    \bottomrule
\end{tabular}
\end{footnotesize}
}
\end{table}

\textbf{Non-stationary client unavailability.}
A total of four unavailable dynamics are evaluated in~\prettyref{tab: exp main text},
including 
stationary and
{\em non}-stationary with
staircase,
sine and
interleaved sine trajectories,
with their visualizations available in the same table.
The classification tasks become more challenging as the list progresses
due to the growing complexity in the non-stationary dynamics.
Furthermore,
our choices of the non-stationary dynamics are motivated by real-world federated learning participation statistics,
for example, sine trajectory \cite{bonawitz2019towards},
and by generalizing the existing participation patterns
such as cyclic participation \cite{pmlr-v202-cho23b,wang2023lightweight}.
In particular,
the interleaved sine dynamics is
more challenging than
the vanilla cyclic availability dynamics
since clients become available during each active period with probability that is less than $1$ and non-stationary simultaneously.
Formally,
client $i$'s dynamics is defined as $p_i^t = p_i \cdot f_i(t)$,
where $f_i(t)$ is a time-dependent function under non-stationary dynamics
but $f_i(t) = 1$ when stationary,
and $p_i = \iprod{\nu_i}{\phi}$.
$\phi$ characterizes the unbalanced contribution of different image classes to the generated probabilities.
Each element of $[\phi]_c$ is drawn from $\mathsf{Uniform}(0,\bm{\Phi}_c)$,
where a smaller $\bm{\Phi}_c$ leads to a less significant contribution of that image class.

Correlating the local data distribution and the probability of client availability is a common practice in the prior literature.
For example, Gu et al. in \cite{gu2021fast} experiment with a formula for $p_i$ so that clients that hold images of smaller digits participate less frequently. Wang and Ji in \cite{wang2023lightweight} construct $p_i$ as an inner product of the clients' local data distribution $\nu_i$ and an external distribution $\bm{\Phi}^\prime$.
It is immediately clear that
the coupling of local data distribution $(\nu_i \sim \mathsf{Dirichlet(\alpha=0.1)})$ and class contribution $\phi$ leads to {\em non-independent} $p_i$'s. 
In addition,
\prettyref{ass: prob lower bound} will not hold in the case of interleaved sine non-stationary dynamics 
since $p_i^t$'s occasionally reach 0.
Although 
being agnostic to
the challenging client unavailability dynamics not covered by our analysis,
we observe that~\FedAPM~retains its outperformance.
Comparisons will be specified next.

{\bf Benchmark algorithms and discussions.}
We compare~\FedAPM~with six baseline algorithms,
including
\FedAvg~over active clients \cite{mcmahan2017communication},
\FedAvg~over all clients, \FedAU~\cite{wang2023lightweight},
\FAST~\cite{ribero2022federated},
\FedAvg~with known $p_i^t$'~\cite{perazzone2022communication},
\MIFA~\cite{gu2021fast}
and~\FedVARP~\cite{jhunjhunwala2022fedvarp}.
The details of the algorithm and the additional results
are deferred to~\prettyref{app: numerical}.
It is observed that~\FedAPM~consistently outperforms the algorithms not aided by memory or known statistics.
Surprisingly,
\FedAPM~occasionally beats~\MIFA~and~\FedVARP,
which are memory-heavy.
We attribute it to reuse of stored gradients from the unavailable clients. 
Although~\FedAPM~brings in stalenss due to implicit gossiping,
our results (\prettyref{tab: slowdown supp} in~\prettyref{app: numerical}) indicate that there is no significant slowdown for~\FedAPM~when compared to vanilla~\FedAvg,
where we study the first round to achieve a targeted accuracy by different algorithms.
In addition,
\FedAPM~attains competitive or even better performance than~\FedAvg~with known probability,
yet unknown to the underlying dynamics in client unavailability.

\section{Conclusion}
In this paper,
we have shown that 
the impacts of heterogeneous and non-stationary client unavailability can be significant through concrete examples on~\FedAvg.
To address this,
we have proposed an algorithm~\FedAPM,
which provably converges 
by adaptively echoing clients' local improvement and by evenly diffusing local updates through implicit gossiping.
Theoretically, it achieves the desired linear speedup property.
Experiments have validated the superiority of~\FedAPM~over state-of-the-art algorithms under diversified non-stationary dynamics.
Future work will investigate how to extend our analysis to broader unavailability dynamics such as non-independent and non-stationary unavailability
and 
how to incorporate our findings
into federated learning algorithms of different local optimization methods.

\newpage

\begin{ack}
We gratefully acknowledge the support
from the National Science Foundation under grants 2106891, 2107062,
the National Science Foundation CAREER award under grant 2340482,
and the Sony Faculty Innovation Award.
The research was sponsored by the Army Research Laboratory 
under Cooperative Agreement Number W911NF-23-2-0014. 
The views and conclusions contained in this document are those of the authors and should not be interpreted as representing the official policies, either expressed or implied,
of the Army Research Laboratory,
the National Science Foundation, 
or the U.S. Government. The U.S. Government is authorized to reproduce and distribute reprints for Government purposes notwithstanding any copyright notation herein.
We also thank Connor J. McLaughlin for valuable discussions and feedback on this work.
\end{ack}

\bibliographystyle{plain}
\bibliography{one}

\newpage
\appendix

\section*{Appendices}
\label{toc}

Here,
we provide an overview of the Appendix.
In particular,
the proofs of the main results
are presented and backed by supporting lemmas and propositions.

\startcontents[sections]
\printcontents[sections]{l}{1}{\setcounter{tocdepth}{2}}

\newpage
\section{Limitations}
\label{app: limitations}
The limitations of our work are two-fold:
\begin{enumerate}[leftmargin=*]
    \item 
    The client unavailability dynamics are assumed to be independent and strictly positive across clients and rounds.
    While deriving guarantees is generally challenging without assuming independence and positivity (see~\prettyref{sec: problem formulation}), it is interesting to explore how to relax the client unavailability dynamics,
    where the probabilities can potentially have arbitrary trajectories.
    
    \item
    Our study focuses on 
    heterogeneous and non-stationary client unavailability in federated learning,
    which may vary greatly due to its inherent uncontrollable nature.
    Although we have shown~\FedAPM~provably converges to a stationary point of even non-convex objectives,
    an interesting yet challenging future direction is to incorporate variance reduction techniques for a more robust update.
\end{enumerate}

\section{Broader Impacts}
\label{app: broad impact}
Federated learning 
has become the main trend for distributed learning in recent years
and has empowered commercial industries such as
autonomous vehicles, the Internet of Things, and natural language processing.
Our paper focuses on the practical implementation of federated learning systems in the real world
and has significantly advanced the theory and algorithms for federated learning by bringing together
insights from statistics, optimization, distributed computing 
and engineering practices.
In addition,
our research is important for federated learning systems to expand their outreach to more undesirable deployment environments.
We are unaware of any potential negative social impacts of our work.

\section{Nomenclatures}
\label{app: nomenclature}
In this section,
we provide the notations and nomenclatures used throughout
our proofs for a comprehensive presentation.
However,
it is worth noting that
all notations have been properly introduced before their first use.
We next articulate the missing definitions and equation formulas.

\begin{table}[!ht]
    \centering
    \caption{Notation table}
    \label{tab: notation table}
    \begin{tabular}{c  p{12cm}}
    \toprule
       $\norm{\bm{v}}$  
       &  
       The $l_2$ norm of a given vector $\bm{v}$.\\
       \midrule
       $\fnorm{A}$  &
       The Frobenius norm of a given matrix $A$. \\
       \midrule
       $\calF^t$ &
       The sigma algebra generated by randomness up to round $t$.\\
       \midrule
       $\lambda_2(A)$ &
       The second largest eigenvalue of a square matrix $A$. \\
       \midrule
       $\reals^d$ &
       A $d$-dimensional vector space, where $d$ denotes the dimension. \\
       \midrule
       $[m]$ &
       A set $\sth{k \mid k\in\naturals, k\in[1,m]}$. \\
       \midrule
       $\indc{\calE}$ &
       An indicator function of event $\calE$, \ie, $\indc{\calE}=1$ when event $\calE$ occurs, but $\indc{\calE}=0$ otherwise. \\
       \midrule
       $\lesssim$ &
       $f(n) \lesssim g(n)$,
       if there exists a constant $c_o>0$ and an integer $n_0 \in \naturals$, $f(n)\le c_o g(n)$ for all $n \ge n_0$. \\
       \midrule
       $\asymp$ &
       $f(n) \asymp g(n)$,
       if there exists a constant $c_{\Theta}>0$ and an integer $n_0 \in \naturals$, $f(n)= c_{\Theta} g(n)$ for all $n \ge n_0$. \\
    \bottomrule
    \end{tabular}
\end{table}

\begin{table}[!t]
    \centering
    \caption{Algorithmic nomenclature table}
    \label{tab: alg nomenclature table}
    \begin{tabular}{p{1cm}  p{12cm}}
    \toprule
       $\calA^t$  
       &  
       The set of active clients in round $t$.\\
       \midrule
       $W^{t}$  
       &  
       A doubly stochastic matrix to capture the information mixing error.
       Its definition can be found in~\eqref{eq: W matrix elements}.
       \\
       \midrule
       $\tau_i(t)$  &
       $\tau_i(t) \triangleq \sup \{t^\prime \mid t^\prime < t, i \in \calA^{t^\prime}\}$ 
       defines client $i$'s most recent active round. In particular, $\tau_i(0)=-1$ for all $i \in [m]$. \\
       \midrule
       $\x_i^t$ &
       The real model at client $i$ at the {\bf beginning} of round $t$ in~\prettyref{alg: fedpbc+}. \\    
       \midrule
       $\bz_i^t$ &
       The auxiliary model at client $i$ at the {\bf beginning} of round $t$. 
       Refer to~\prettyref{def: auxiliary sequence} for more details. 
       The sequence is for analysis only and is not computed by any clients.
       \\
       \midrule
       $\x^t$ &
       The aggregated real model at the {\bf end} of round $t-1$ in~\prettyref{alg: fedpbc+}. \\    
       \midrule
       $\bz^t$ &
       The auxiliary model at the {\bf end} of round $t-1$. 
       \\
       \midrule
       $\x_i^{t \dagger}$,
       $\bz_i^{t \dagger}$ &
       The real model of an active client $i$,
       and auxiliary model of an active client $i$ 
       after $s$-step local computation in round $t$, respectively. 
       Refer to~\prettyref{alg: fedpbc+} for more details. \\
       \midrule
       $\x_i^{(t,r)}$ &
       The real model at client $i$ after $r$-step local computation. \\
       \midrule
       $\bar{\x}^t$,
       $\bar{\bz}^t$
        &
       The real and auxiliary model mean over all clients in a distributed system and in round $t$, 
       respectively. \\
       \midrule
       $F_i(\x)$ &
       The local objective function at client $i$, 
       which is assumed to be non-convex. \\
       \midrule
       $F(\x)$ &
        The global objective function defined in~\eqref{eq: global obj}:
       $F(\x) \triangleq \sum_{i=1}^m F_i(\x)/m$. \\
       \midrule
       $\nabla \ell_i(\x)$ &
       The local stochastic gradient function at client $i$ taken with respect to $\x$. \\
       \midrule
       $\nabla F_i(\x)$ &
       The local true gradient function at client $i$ taken with respect to $\x$. \\
       \midrule
       $\calD_i$ &
       Client $i$'s local data distribution. \\
       \midrule
       $\xi_i$ &
       An {\bf independent} stochastic sample drawn from client $i$'s local distribution $\calD_i$. \\     
    \bottomrule
    \end{tabular}
\end{table}

\begin{table}[!t]
    \centering
    \caption{Variable table 
    }
    \label{tab: constant table}
    \begin{tabular}{c  p{12cm}}
    \toprule
       $L$ &
       Lipschitz constant in~\prettyref{ass: 2 smmothness}. \\
       \midrule
       $\sigma^2$ &
       The upper bound of the stochastic gradient variance. \\
       \midrule
       $(\beta,~\zeta)$ &
       Parameters that capture the averaged gradient dissimilarity 
       between global and local objectives. \\
       \midrule
       $\rho$ &
       The spectral norm of a stochastic matrix in expectation.  \\
       \midrule
       $s$ &
       The number of local computation steps.  \\
       \midrule
       $m$ &
       The number of clients in the federated learning system.  \\
    \bottomrule
    \end{tabular}
\end{table}

{\bf Missing definitions and equation formulas.}
\paragraph{The iterate of $\bz_i$ when $i \in \calA^{t-1}$.}
\begin{align}
    \nonumber
    \bz_i^t &= 
    \frac{1}{\abth{\calA^{t-1}}}
    \sum_{j \in \calA^{t-1}}
    \pth{
    \bz_j^{t-1}
     -
    \eta_l \eta_g
    \sum_{r=0}^{s-1}
    \nabla \ell_j(\x_j^{(t-1,r)};\xi_i^{(t,r)})} \\
    \label{eq: z iterate last active}
    &~~~+ 
    \frac{\eta_l \eta_g}{\abth{\calA^{t-1}}}
    \sum_{j \in \calA^{t-1}}
    \pth{t - 2 - \tau_j(t-1)}
    \sum_{r=0}^{s-1}
    \pth{
    \nabla F_j(\x_j^{\tau_j(t-1)+1})
    -
    \nabla \ell_j(\x_j^{(t-1,r)};\xi_i^{(t,r)})
    }.
\end{align}
\paragraph{Local parameter innovation $\tilde{\bm{G}}^t$ of the auxiliary sequence.}
\begin{align}
    \nonumber
    \tilde{\bm{G}}^t_{i}
    &\triangleq
    \indc{i \in \calA^t }
    \qth{
    \pth{t - \tau_i(t)}
    \sum_{r=0}^{s-1}
    \nabla \ell_i(\x_i^{(t,r)})
    - 
    s\pth{t- 1 - \tau_i(t)}
    \nabla F_i(\x_i^{\tau_i(t)+1}) } \\
    \nonumber
    &~~~+
    \indc{i \notin \calA^t}
    s \nabla F_i(\x_i^{\tau_i(t)+1}) \\
    &=
    \indc{i \in \calA^t}
    \pth{t - \tau_i(t)}
    \sum_{r=0}^{s-1}
    \pth{
    \nabla \ell_i(\x_i^{(t,r)})
    -
    \nabla F_i(\x_i^{t})}
    +
    s \nabla F_i(\x_i^{t}),
    \label{eq: auxiliary update detail}
\end{align}
where the last equality holds because $\x_i^t = \x_i^{\tau_i(t)+1}$ and re-grouping.
\paragraph{Decomposition in the Proof of~\prettyref{lmm: consensus z without pseudo}.}
The local parameter innovation of the auxiliary sequence 
$\tilde{\bm{G}}^t$ can be decomposed as $\tilde{\bm{G}}^t \triangleq \TDt{t} + \Dt{t} + s \DFt{t}$.
Detailed definitions can be found below.
\begin{itemize}[leftmargin=*]
    \item 
    $[\TDt{t}]_i \triangleq
        \indc{i \in \calA^t}
        (t - \tau_i(t))
        \sum_{r=0}^{s-1} \pth{ \nabla \ell_i(\x_i^{(t,r)};\xi_i^{(t,r)})  - \nabla F_i(\x_i^{(t,r)}) };$
    \item 
    $[\Dt{t}]_i \triangleq
        \indc{i \in \calA^t}
        (t - \tau_i(t)) 
        \sum_{r=0}^{s-1}
        \pth{\nabla F_i(\x_i^{(t,r)})
        -
        \nabla F_i(\x_i^t)
        };$
    \item 
    $[\DFt{t}]_i \triangleq
        \nabla F_i(\x_i^t).$
\end{itemize}

\section{Useful Inequalities}
\label{app: preliminaries}
For completeness and for ease of exposition, we present some common inequalities that will be frequently used in our proofs.

\vskip \baselineskip 

The followings hold for any $\bm{a}_i \in \reals^d$ and any $i\in[m]$.
\begin{enumerate}
\item {Jensen's inequality.}
\begin{align}
\label{eq: jensen's inequality}
\norm{\frac{1}{m} \sum_{i=1}^m \bm{a}_i}^2
\le
\frac{1}{m} \sum_{i=1}^m \norm{\bm{a}_i}^2
~~~\text{and}~~~
\norm{\sum_{i=1}^m \bm{a}_i}^2
\le
m \sum_{i=1}^m \norm{\bm{a}_i}^2.
\end{align}
\item {Young's inequality (a.k.a. Peter-Paul inequality).}
\begin{align}
\label{eq: Young's inequality}
\iprod{\bm{a}_1}{\bm{a}_2}
\le
\frac{\norm{\bm{a}_1}^2}{2 \epsilon} + \frac{\epsilon \norm{\bm{a}_2}^2}{2}, ~~~ \text{for any }\epsilon>0. 
\end{align}
Equivalently,
we have
\begin{align}
\nonumber
\norm{\bm{a}_1 + \bm{a}_2}^2
&=
\norm{\bm{a}_1}^2
+
\norm{\bm{a}_2}^2
+
2 \iprod{\bm{a}_1}{\bm{a}_2} \\
\label{eq: Young's inequality 2}
&\le
\pth{1 + \frac{1}{\epsilon}}
\norm{\bm{a}_1}^2
+
\pth{1 + \epsilon}
\norm{\bm{a}_2}^2
, ~~~ \text{for any }\epsilon>0. 
\end{align}

\item {Smoothness corollary.} 
{\it Given Assumption \ref{ass: 2 smmothness}, it holds that}
\begin{align}
\label{eq: smoothness corollary}
\nonumber
F(\bm{a}_1) - F(\bm{a}_2) 
&=
\iprod{\bm{a}_1 - \bm{a}_2}{\int_0^1 \nabla F(\bm{a}_2 + \tau (\bm{a}_1 - \bm{a}_2)) \diff \tau}
\\
\nonumber
&=
\iprod{\nabla F(\bm{a}_2)}{\bm{a}_1 - \bm{a}_2} +
\int_0^1 \iprod{\bm{a}_1 - \bm{a}_2}{\nabla F(\bm{a}_2 + \tau (\bm{a}_1 - \bm{a}_2)) - \nabla F(\bm{a}_2) }\diff \tau \\
\nonumber
&\overset{(a)}{\le}
\iprod{\nabla F(\bm{a}_2)}{\bm{a}_1 - \bm{a}_2} +
L \int_0^1 \tau \norm{\bm{a}_1 - \bm{a}_2}\norm{(\bm{a}_1 - \bm{a}_2) }\diff \tau
\\
&\le
\iprod{\nabla F(\bm{a}_2)}{\bm{a}_1 - \bm{a}_2}
+
\frac{L}{2} \norm{\bm{a}_1 - \bm{a}_2}^2
,
\end{align}
where $(a)$ follows from Cauchy-Schwartz inequality and Assumption \ref{ass: 2 smmothness}.
\end{enumerate}

\newpage 
\section{Descent Lemma~(\prettyref{lmm: descent lemma})}
\label{app: descent lemma aux sequence}
In this section,
we first present a bound on multi-step local computation.
Then, we apply the bound to the analysis of descent lemma.

\subsection{Multi-step perturbation}
\begin{lemma}
    \label{lmm: multi-local steps}
    For $s\ge1$ and under~\prettyref{ass: 2 smmothness}, \ref{ass: bounded variance client-wise}
    and $\eta_l \le 1 / (4 s L)$
    , we have
    \begin{align*}
        \expect{\norm{\sum_{r=0}^{s-1} \nabla F_i(\x_i^{(t,r)}) - \nabla F_i(\x_i^t)}^2 ~\Big|~\calF^t}
        &\le
        4 \eta_l^2 s^3 L^2 \sigma^2
        +
        16 \eta_l^2 s^4 L^2 
        \norm{\nabla F_i(\x_i^t)}^2
    \end{align*}
\end{lemma}

\begin{proof}[\bf Proof of~\prettyref{lmm: multi-local steps}]
    The proof shares a similar road map to \cite[Lemma 2]{yang2021achieving},
    but the objective is instead to show an upper bound with respect to  $\norm{\nabla F_i (\x_i^t)}^2$.

    For $s\ge1$, it holds that
    \begin{align}
        \nonumber
        \expect{\norm{\sum_{r=0}^{s-1} \nabla F_i(\x_i^{(t,r)}) - \nabla F_i(\x_i^t)}^2
        \Big| ~\calF^t
        }
        &\overset{(a)}{\le}
        s \sum_{r=0}^{s-1}
        \expect{\norm{\nabla F_i(\x_i^{(t,r)}) - \nabla F_i(\x_i^t)}^2 \Big| ~\calF^t} \\
        \label{eq: local after smoothness}
        &\overset{(b)}{\le}
        s L^2 \sum_{r=0}^{s-1}
        \expect{\norm{\x_i^{(t,r)} - \x_i^t}^2 \Big | ~\calF^t},
    \end{align}
    where inequality $(a)$ holds because of Jensen's inequality,
    inequality $(b)$ holds because of~\prettyref{ass: 2 smmothness}.
    It remains to bound $\mathbb{E}[\|\x_i^{(t,r)} - \x_i^t\|^2 \mid ~\calF^t]$.
    In what follows, we use $\nabla \ell_i^{(t,k)}$ to denote $\nabla \ell_i (\x_i^{(t,k)})$ and $\nabla F_i^{(t,k)}$ as $\nabla F_i (\x_i^{(t,k)})$, respectively, for ease of presentation.
    \begin{align*}
        &\expect{\norm{\x_i^{(t,r)} - \x_i^t}^2\Big| ~\calF^t}
        =
        \expect{\norm{\x_i^{(t,r-1)} - \x_i^t - \eta_l \nabla \ell_i^{(t,r-1)}}^2\Big| ~\calF^t} \\
        &=
        \expect{\norm{
        -\eta_l \pth{
        \nabla \ell_i^{(t,r-1)} - \nabla F_i^{(t,r-1)}}
        +\x_i^{(t,r-1)} - \x_i^t 
        -\eta_l \pth{
        \nabla F_i^{(t,r-1)} - \nabla F_i^{t}
        + \nabla F_i^{t}}
        }^2
        \Big| ~\calF^t}  \\
        &\overset{(c)}{=}
        \eta_l^2
        \expect{\norm{
        \nabla \ell_i^{(t,r-1)} - \nabla F_i^{(t,r-1)}}^2\Big| ~\calF^t}
        +
        \expect{
        \norm{
        \x_i^{(t,r-1)} - \x_i^t 
        -\eta_l \pth{
        \nabla F_i^{(t,r-1)} - \nabla F_i^{t}
        + \nabla F_i^{t}}
        }^2\Big| ~\calF^t} \\
        &\overset{(d)}{\le}
        \eta_l^2
        \expect{\norm{
        \nabla \ell_i^{(t,r-1)} - \nabla F_i^{(t,r-1)}}^2\Big| ~\calF^t} \\
        &~~~+
        \pth{1 + \frac{1}{2 s - 1}}
        \expect{
        \norm{
        \x_i^{(t,r-1)} - \x_i^t}^2\Big| ~\calF^t}
        +
        2 s \eta_l^2
        \expect{
        \norm{
        \nabla F_i^{(t,r-1)} - \nabla F_i^{t}
        + \nabla F_i^{t}}^2
        \Big| ~\calF^t} \\
        &\le
        \eta_l^2
        \expect{\norm{
        \nabla \ell_i^{(t,r-1)} - \nabla F_i^{(t,r-1)}}^2\Big| ~\calF^t} \\
        &~~~+
        \pth{1 + \frac{1}{2 s - 1}}
        \expect{
        \norm{
        \x_i^{(t,r-1)} - \x_i^t}^2\Big| ~\calF^t}
        +
        4 s \eta_l^2
        \expect{
        \norm{
        \nabla F_i^{(t,r-1)} - \nabla F_i^{t}
        }^2\Big| ~\calF^t}
        + 
        4 s \eta_l^2
        \norm{\nabla F_i^{t}}^2 \\
        &\overset{(e)}{\le}
        \eta_l^2 \sigma^2 
        + 
        4 s \eta_l^2
        \norm{\nabla F_i^{t}}^2 \\
        &~~~+
        \pth{1 + \frac{1}{2 s - 1}}
        \expect{
        \norm{
        \x_i^{(t,r-1)} - \x_i^t}^2\Big| ~\calF^t}
        +
        4 s L^2 \eta_l^2
        \expect{
        \norm{
        \x_i^{(t,r-1)} - \x_i^t
        }^2\Big| ~\calF^t}\\
        &=
        \eta_l^2 \sigma^2 
        + 
        4 s \eta_l^2
        \norm{\nabla F_i^{t}}^2 
        +
        \pth{1 + \frac{1}{2 s - 1} + 4 s L^2 \eta_l^2}
        \expect{
        \norm{
        \x_i^{(t,r-1)} - \x_i^t}^2\Big| ~\calF^t}
        ,
    \end{align*}
    where equality $(c)$ holds because $\nabla \ell_i^{(t,k)}$ is an unbiased estimator of $\nabla F_i^{(t,r)}$,
    inequality $(d)$ holds because of Young's inequality,
    inequality $(e)$ holds because of~\prettyref{ass: 2 smmothness}.

    By $\eta_l \le \frac{1}{4 s L}$, it holds that
    \[
        \frac{1}{2 s - 1} + 4 s L^2 \eta_l^2 
        \le
        \frac{1}{2 s - 1} + \frac{1}{4 s}
        \le
        \frac{2}{2 s - 1}.
    \]
    Unroll the recursion, 
    we have
    \begin{align*}
        \expect{\norm{\x_i^{(t,r)} - \x_i^t}^2\Big| ~\calF^t} &\le
        \sum_{k=0}^{r-1}
        \pth{1 + \frac{2}{2 s - 1}}^k
        \pth{
        \eta_l^2 \sigma^2
        +
        4 s \eta_l^2 \norm{\nabla F_i^t}^2
        } \\
        &\le
        \sum_{k=0}^{s-1}
        \pth{1 + \frac{2}{2 s - 1}}^k
        \pth{
        \eta_l^2 \sigma^2
        +
        4 s \eta_l^2 \norm{\nabla F_i^t}^2
        } \\
        &=
        \frac{2 s - 1}{2}
        \qth{ 
        \pth{1 + \frac{2}{2 s - 1}}^{s - \frac{1}{2} }
        \pth{1 + \frac{2}{2 s - 1}}^{\frac{1}{2} }
        - 1}
        \pth{
        \eta_l^2 \sigma^2
        +
        4 s \eta_l^2 \norm{\nabla F_i^t}^2
        } \\
        &\overset{(f)}{\le}
        \pth{s - \frac{1}{2}}
        \qth{ 
        \sqrt{3} e
        - 1}
        \pth{
        \eta_l^2 \sigma^2
        +
        4 s \eta_l^2 \norm{\nabla F_i^t}^2
        } \\
        &\overset{(g)}{\le}
        4 s \eta_l^2 \sigma^2
        +
        16 s^2 \eta_l^2 \norm{\nabla F_i^t}^2,
    \end{align*}
    where inequality $(f)$ holds because of $(1 + 1/x)^x < \exp(1)$,
    inequality $(g)$ holds because of $\sqrt{3} \exp(1) - 1 < 4$.
    Plug it back into~\eqref{eq: local after smoothness},
    we have the desired result
    \begin{align*}
        \expect{\norm{\sum_{r=0}^{s-1} \nabla F_i(\x_i^{(t,r)}) - \nabla F_i(\x_i^t)}^2\Big| ~\calF^t}
        &\le
        4 \eta_l^2 s^3 L^2 \sigma^2
        +
        16 \eta_l^2 s^4 L^2 \norm{\nabla F_i(\x_i^t)}^2
        .
    \end{align*}
\end{proof}

\subsection{Descent lemma}
\begin{proof}[\bf Proof of~\prettyref{lmm: descent lemma}]
By Assumption \ref{ass: 2 smmothness} and inequality \eqref{eq: smoothness corollary}, we have 
\begin{align*}
F(\bar{\bz}^{t+1}) - F(\bar{\bz}^t)
&\le
\underbrace{\iprod{\nabla F(\bar{\bz}^{t})}{\bar{\bz}^{t+1} - \bar{\bz}^t}}_{(\rmA)}
+
\underbrace{\frac{L}{2} \norm{\bar{\bz}^{t+1} - \bar{\bz}^t}^2}_{(\rmB)}. 
\end{align*}
The one-round innovation of $\bar{\bm{z}}$ can be rewritten as  

\begin{align}
\nonumber
& \bar{\bz}^{t+1} - \bar{\bz}^t
= \frac{1}{m}\sum_{i\in \calA^t} \pth{\bz_i^{t \dagger} - \bz_i^{t}} + \frac{1}{m}\sum_{i\notin \calA^t} \pth{\bz_i^{t+1} - \bz_i^{t}} \\
\nonumber
&=
\frac{1}{m}\sum_{i=1}^m \indc{i \in \calA^t}
\pth{\eta_l \eta_g s \sum_{k=\tau_i(t)+1}^{t-1} \nabla F_i(\x_i^k) - \eta_l \eta_g (t - \tau_i(t)) \sum_{r=0}^{s-1} \nabla \ell_i(\x_i^{(t,r)};\xi_i^{(t,r)}) }\\
&~~~
\nonumber
-\frac{\eta_l \eta_g s}{m} \sum_{i=1}^m \indc{i \notin \calA^t}
\nabla F_i(\x_i^t)\\
\nonumber
&\overset{(a)}{=}  
\frac{1}{m}\sum_{i=1}^m \indc{i \in \calA^t}\eta_l \eta_g s (t-1 - \tau_i(t)) \nabla F_i(\x_i^t) 
- \frac{1}{m}\sum_{i=1}^m \indc{i \in \calA^t} \eta_l \eta_g (t - \tau_i(t)) \sum_{r=0}^{s-1} \nabla \ell_i(\x_i^{(t,r)};\xi_i^{(t,r)}) \\
\nonumber
&~~~-\frac{\eta_l \eta_g s}{m} \sum_{i=1}^m \indc{i \notin \calA^t}
\nabla F_i(\x_i^t) \\
\nonumber
&\overset{(b)}{=}
\frac{\eta_l \eta_g}{m} \sum_{i=1}^m \indc{i \in \calA^t}
(t - \tau_i(t))
\sum_{r=0}^{s-1} \pth{\nabla F_i(\x_i^{(t,r)}) - \nabla \ell_i(\x_i^{(t,r)};\xi_i^{(t,r)})} \\
\nonumber
&~~~+
\frac{\eta_l \eta_g}{m} \sum_{i=1}^m \indc{i \in \calA^t}
(t - \tau_i(t)) 
\sum_{r=0}^{s-1}
\pth{ \nabla F_i(\x_i^t) - 
\nabla F_i(\x_i^{(t,r)})} \\
\nonumber
&~~~-\frac{\eta_l \eta_g s}{m} \sum_{i=1}^m 
\nabla F_i(\x_i^t), 
\end{align}
where equality $(a)$ using the fact that $\x_i^k = \x_i^t$ for all $k$ such that $\tau_i(t)+1 \le k \le t$, and equality $(b)$ is obtained by adding and subtracting $\nabla \ell_i(\x_i^{t};\xi_i^{(t,r)})$ 
and by the fact that $\pth{\indc{i \in \calA^t} + \indc{i \notin \calA^t}} = 1$. 

\paragraph{Bounding $(\rmA)$.}
\begin{align*}
(\rmA)
&=\iprod{\nabla F(\bar{\bz}^{t})}{\bar{\bz}^{t+1} - \bar{\bz}^t} \\
& = \underbrace{
\eta_l \eta_g\iprod{\nabla F(\bar{\bz}^t)}{
\frac{1}{m} \sum_{i=1}^m \indc{i \in \calA^t} \sum_{p=-1}^{t-1} \indc{\tau_i(t) = p}
(t - p)
\sum_{r=0}^{s-1} \pth{\nabla F_i(\x_i^{(t,r)}) - \nabla \ell_i(\x_i^{(t,r)};\xi_i^{(t,r)})}}}_{(\rmA.\rmI)} \\
&~~~+
\underbrace{\frac{\eta_l \eta_g}{m} \sum_{i=1}^m \indc{i \in \calA^t}
\sum_{p=-1}^{t-1} \indc{\tau_i(t) = p}
\iprod{\nabla F (\bar{\bz}^t)}{(t - p) 
\sum_{r=0}^{s-1}
\pth{ \nabla F_i(\x_i^t) - 
\nabla F_i(\x_i^{(t,r)})} }}_{(\rmA.\rmI\rmI)} \\
&~~~+
\underbrace{
\frac{\eta_l \eta_g s}{m} \sum_{i=1}^m \iprod{\nabla F (\bar{\bz}^t)}{
\nabla F_i(\bz_i^t) - \nabla F_i(\x_i^t)}
}_{(\rmA.\rmI\rmI\rmI)} 
-
\underbrace{\eta_l \eta_g s \iprod{\nabla F(\bar{\bz}^t)}{\frac{1}{m}\sum_{i=1}^m \nabla F_i(\bz_i^t)}}_{(\rmA.\rmI\rmV)}.
\end{align*}
\paragraph{Bounding $(\rmA.\rmI)$}
\begin{small}
\begin{align*}
    &\expect{(\rmA.\rmI)\Big|\calF^t} \\
    &\overset{(a)}{=}
    \eta_l \eta_g 
    \expect{
    \expect{
    \iprod{\nabla F(\bar{\bz}^t)}{\frac{1}{m}\sum_{i=1}^m \indc{i\in\calA^t} \sum_{p=-1}^{t-1} \indc{\tau_i(t) = p}(t-p) 
    \sum_{r=0}^{s-1} \pth{\nabla F_i(\x_i^{(t,r)}) - \nabla \ell_i(\x_i^{(t,r)};\xi_i^{(t,r)})}}
    \Big| \x_i^{(t,r)}, \calF^t}
    \Big|\calF^t} \\
    &\overset{(b)}{=}
    \eta_l \eta_g 
    \left <
    {\nabla F(\bar{\bz}^t)}, 
    \right . \\
    &\qquad \qquad \quad
    \left . 
    {\frac{1}{m}\sum_{i=1}^m 
    \expect{\indc{i\in\calA^t}\Big|\calF^t}
    \sum_{p=-1}^{t-1} \indc{\tau_i(t) = p}(t-p) 
    \sum_{r=0}^{s-1} 
    \expect{
    \expect{
    \pth{\nabla F_i(\x_i^{(t,r)}) - \nabla \ell_i(\x_i^{(t,r)};\xi_i^{(t,r)})}
    \Big|  \x_i^{(t,r)}, \calF^t}
    \Big|\calF^t}} \right > \\
    &=0,
\end{align*}
\end{small}
where equality $(a)$ holds because of the law of total expectation,
equality $(b)$ holds because $\indc{i \in \calA^t}$ is by definition independent of others and~\prettyref{ass: bounded variance client-wise}.

\paragraph{Bounding $(\rmA.\rmI\rmI)$}

\begin{align*}
(\rmA.\rmI\rmI)
&\overset{(c)}{\le}
\frac{\eta_l \eta_g}{m} \sum_{i=1}^m \indc{i \in \calA^t} \sum_{p=-1}^{t-1} \indc{\tau_i(t) = p}
\pth{\frac{s}{8}\norm{\nabla F(\bar{\bz}^t)}^2
+
\frac{2 (t - p)^2}{s} \norm{\sum_{r=0}^{s-1} \nabla F_i(\x_i^{t}) - \nabla F_i(\x_i^{(t,r)})}^2} \\
&=
\frac{\eta_l \eta_g s}{8 m} \sum_{i=1}^m \indc{i \in \calA^t} 
\norm{\nabla F(\bar{\bz}^t)}^2 \\
&~~~+
\frac{\eta_l \eta_g}{m} \sum_{i=1}^m \indc{i \in \calA^t} \sum_{p=-1}^{t-1} \indc{\tau_i(t) = p}
\frac{2 (t - p)^2}{s} \norm{\sum_{r=0}^{s-1} \nabla F_i(\x_i^{t}) - \nabla F_i(\x_i^{(t,r)})}^2 
,
\end{align*}
where inequality $(c)$ holds because of Young's inequality.
It follows that
\begin{align*}
\expect{(\rmA.\rmI\rmI)\Big|\calF^t}
&\overset{(d)}{\le}
\frac{\eta_l \eta_g s}{8} 
\norm{\nabla F(\bar{\bz}^t)}^2 
+
\frac{8 \eta_g \eta_l^3 s^2 L^2 \sigma^2}{m} \sum_{i=1}^m 
\sum_{p=-1}^{t-1} \indc{\tau_i(t) = p}
(t - p)^2
\\
&~~~
+
\frac{32 \eta_g \eta_l^3 s^3 L^2 }{m} \sum_{i=1}^m 
\sum_{p=-1}^{t-1} \indc{\tau_i(t) = p}
(t - p)^2
\norm{\nabla F_i (\x_i^t)}^2 \\
&=
\frac{\eta_l \eta_g s}{8} 
\norm{\nabla F(\bar{\bz}^t)}^2
+
\frac{8 \eta_g \eta_l^3 s^2 L^2 \sigma^2}{m} \sum_{i=1}^m 
\sum_{p=-1}^{t-1} \indc{\tau_i(t) = p}
(t - p)^2 \\
&~~~+
\frac{32 \eta_g \eta_l^3 s^3 L^2 }{m} \sum_{i=1}^m \sum_{p=-1}^{t-1} \indc{\tau_i(t) = p}
(t - p)^2 
\norm{\nabla F_i(\x_i^{p+1})}^2,
\end{align*}
where inequality $(d)$ holds because of~\prettyref{lmm: multi-local steps},
the last equality using the fact that $\x_i^k = \x_i^t$ for all $k$ such that $ \tau_i(t) + 1 \le k \le t$.

\paragraph{Bounding $(\rmA.\rmI\rmI\rmI)$.}
\begin{align*}
(\rmA.\rmI\rmI\rmI)
&=
\frac{\eta_l \eta_g s }{m} 
\sum_{i=1}^m 
\iprod{\nabla F(\bar{\bz}^t)}{\nabla F_i(\bz_i^t) - \nabla F_i(\x_i^t)} 
\overset{(e)}{\le}
\frac{\eta_l \eta_g s}{8} \norm{\nabla F(\bar{\bz}^t)}^2
+
\frac{2 \eta_l \eta_g s L^2}{m} \sum_{i=1}^m \norm{\bz_i^t - \x_i^t}^2,
\end{align*}
where inequality $(e)$ follows from Young's inequality
and Assumption \ref{ass: 2 smmothness}.
It holds that,
\begin{align*}
\expect{ (\rmA.\rmI\rmI\rmI)\Big|\calF^t}
&\le
\frac{\eta_l \eta_g s}{8} \norm{\nabla F(\bar{\bz}^t)}^2
+
\frac{2 \eta_l \eta_g s L^2}{m} \sum_{i=1}^m \norm{\bz_i^t - \x_i^t}^2.
\end{align*}
\paragraph{Bounding $(\rmA.\rmI\rmV)$}
\begin{align*}
(\rmA.\rmI\rmV) 
&=
\frac{\eta_l \eta_g s}{2} \pth{
\norm{\nabla F (\bar{\bz}^t)}^2
+
\norm{\frac{1}{m}\sum_{i=1}^m \nabla F_i(\bz_i^t)}^2
-
\norm{\nabla F (\bar{\bz}^t) - \frac{1}{m}\sum_{i=1}^m \nabla F_i(\bz_i^t)}^2
}, 
\end{align*}
where the equality follows from the identity in Appendix \ref{app: preliminaries} (3).
It holds that
\begin{align*}
\expect{(\rmA.\rmI\rmV)\Big|\calF^t}
&=
\frac{\eta_l \eta_g s}{2} \pth{
\norm{\nabla F (\bar{\bz}^t)}^2
+
\norm{\frac{1}{m}\sum_{i=1}^m \nabla F_i(\bz_i^t)}^2
-
\norm{\frac{1}{m} \sum_{i=1}^m \nabla F_i (\bar{\bz}^t) - \frac{1}{m}\sum_{i=1}^m \nabla F_i(\bz_i^t)}^2
} \\
&\ge
\frac{\eta_l \eta_g s}{2} \pth{
\norm{\nabla F (\bar{\bz}^t)}^2
+
\norm{\frac{1}{m}\sum_{i=1}^m \nabla F_i(\bz_i^t)}^2
-
\frac{L^2}{m}
\sum_{i=1}^m
\norm{\bar{\bz}^t - \bz_i^t}^2
} .
\end{align*}
Putting $(\rmA)$ together,
\begin{align*}
&\expect{
(\rmA)
\Big|\calF^t}
\le
- \frac{\eta_l \eta_g s}{4} \norm{\nabla F(\bar{\bz}^t)}^2 
+
\frac{8 \eta_g \eta_l^3 s^2 L^2 \sigma^2}{m} \sum_{i=1}^m 
\sum_{p=-1}^{t-1} \indc{\tau_i(t) = p}
(t - p)^2 \\
&~~~+
\frac{2 \eta_l \eta_g s L^2}{m}\sum_{i=1}^m \norm{\x_i^t - \bz_i^t}^2
+
\frac{\eta_l \eta_g s L^2}{2m} 
\sum_{i=1}^m
\norm{\bar{\bz}^t - \bz_i^t}^2  \\
&~~~
-\frac{\eta_l \eta_g s}{2} \norm{\frac{1}{m} \sum_{i=1}^m \nabla F_i(\bz_i^t)}^2
+
\frac{32 \eta_g \eta_l^3 s^3 L^2 }{m} \sum_{i=1}^m \sum_{p=-1}^{t-1} \indc{\tau_i(t) = p}
(t - p)^2 
\norm{\nabla F_i(\x_i^{p+1})}^2.
\end{align*}
\paragraph{Bounding $(\rmB)$.}

\begin{align*}
(\rmB)
&\le
\underbrace{
2 L
\frac{\eta_l^2 \eta_g^2}{m^2} 
\norm{
\sum_{i=1}^m 
\indc{i \in \calA^t}
(t - \tau_i(t))
\sum_{r=0}^{s-1} \pth{\nabla F_i(\x_i^{(t,r)}) - \nabla \ell_i(\x_i^{(t,r)};\xi_i^{(t,r)})}}^2}_{(\rmB.\rmI)} \\
&~~~+
\underbrace{
2 L
\frac{\eta_l^2 \eta_g^2}{m^2} 
m \sum_{i=1}^m \indc{i \in \calA^t}
(t - \tau_i(t))^2
\norm{
\sum_{r=0}^{s-1} \pth{\nabla F_i(\x_i^t) - \nabla F_i(\x_i^{(t,r)})}
}^2}_{(\rmB.\rmI\rmI)} \\
&~~~+
\underbrace{2 L
\frac{\eta_l^2 \eta_g^2 s^2}{m^2} 
m \sum_{i=1}^m
\norm{
\nabla F_i(\x_i^t) - \nabla F_i(\bz_i^t)
}^2}_{(\rmB.\rmI\rmI\rmI)}
+
\underbrace{
2 L \eta_l^2 \eta_g^2 s^2
\norm{
\frac{1}{m}
\sum_{i=1}^m 
\nabla F_i(\bz_i^t)
}^2}_{(\rmB.\rmI\rmV)}
\end{align*}

\paragraph{Bounding $(\rmB.\rmI)$}
Recall that $\delta_{\max} \triangleq \sup_{i\in [m],t \in [T]} p_i^t.$
It holds that,
\begin{small}
\begin{align*}
    \expect{(\rmB.\rmI)\Big|\calF^t} 
    &\overset{(f)}{=}
    2 L
    \frac{\eta_l^2 \eta_g^2}{m^2} 
    \sum_{i=1}^m 
    \expect{\indc{i \in \calA^t}\Big|\calF^t}
    (t - \tau_i(t))^2
    \sum_{r=0}^{s-1} 
    \expect{
    \expect{
    \norm{\nabla F_i(\x_i^{(t,r)}) - \nabla \ell_i(\x_i^{(t,r)};\xi_i^{(t,r)})}^2 \Big|\x_i^{(t,r)}, \calF^t
    }
    \Big|\calF^t} \\
    &\overset{(g)}{\le}
    \frac{2 \eta_l^2 \eta_g^2 s L \delta_{\max}\sigma^2}{m^2} 
    \sum_{i=1}^m
    \sum_{p=-1}^{t-1}
    \indc{\tau_i(t) = p}
    (t - p)^2,
\end{align*}
\end{small}
where equality $(f)$ holds by the law of total expectation and by the independence of event $\{ i \in \calA^t \}$,
inequality $(g)$ holds because of~\prettyref{ass: bounded variance client-wise} 
and by definition $p_i^t \le \delta_{\max}$.

\paragraph{Bounding $(\rmB.\rmI\rmI)$}
We have,
\begin{align*}
    &\expect{(\rmB.\rmI\rmI)\Big|\calF^t} 
    \le
    2 L
    \frac{\eta_l^2 \eta_g^2}{m} 
    \sum_{i=1}^m 
    \sum_{p=-1}^{t-1}
    \indc{\tau_i(t) = p}
    (t - p)^2
    4 \eta_l^2 s^3 L^2 
    \sigma^2 \\
    &\qquad \qquad \qquad \quad
    +
    2 L
    \frac{\eta_l^2 \eta_g^2}{m} 
    \sum_{i=1}^m 
    \indc{\tau_i(t) = p}
    \sum_{p=-1}^{t-1}
    (t - p)^2
    16 \eta_l^2 s^4 L^2 
    \norm{\nabla F_i(\x_i^{t})}^2 \\
    &=
    \frac{8 \eta_g^2 \eta_l^4 s^3 L^3 \sigma^2 }{m} 
    \sum_{i=1}^m 
    \sum_{p=-1}^{t-1}
    \indc{\tau_i(t) = p}
    (t - p)^2
    +
    \frac{
    32 \eta_g^2 \eta_l^4
    s^4 L^3
    }{m} 
    \sum_{i=1}^m 
    \sum_{p=-1}^{t-1}
    \indc{\tau_i(t) = p}
    (t - p)^2
    \norm{\nabla F_i(\x_i^{p+1})}^2,
\end{align*}
where the last equality using the fact that $\x_i^k = \x_i^t$ for all $k$ such that $ \tau_i(t) + 1 \le k \le t$.

{\bf Bounding $(\rmB.\rmI\rmI\rmI)$.}
\quad 
$\expect{(\rmB.\rmI\rmI\rmI)\Big|\calF^t} 
\le \frac{2 \eta_l^2 \eta_g^2 s^2 L^3}{m} 
\sum_{i=1}^m \norm{\x_i^t - \bz_i^t}^2.$

Putting $(\rmB)$ together, we get
\begin{align*}
\expect{(\rmB)\Big|\calF^t}
&\le
\frac{2 \eta_l^2 \eta_g^2 s L \delta_{\max} \sigma^2}{m^2} 
\sum_{p=-1}^{t-1}
\indc{\tau_i(t) = p}
(t - p)^2
+
\frac{8 \eta_g^2 \eta_l^4 s^3 L^3 \sigma^2 }{m} 
\sum_{i=1}^m 
\sum_{p=-1}^{t-1}
\indc{\tau_i(t) = p}
(t - p)^2 \\
&~~~+
\frac{
32 \eta_g^2 \eta_l^4
s^4 L^3
}{m} 
\sum_{i=1}^m 
\sum_{p=-1}^{t-1}
\indc{\tau_i(t) = p}
(t - p)^2
\norm{\nabla F_i(\x_i^{p+1})}^2 \\
&~~~+
\frac{2 \eta_l^2 \eta_g^2 s^2 L^3}{m}
\sum_{i=1}^m
\norm{\x_i^t - \bz_i^t}^2
+
2 L \eta_l^2 \eta_g^2 s^2
\norm{
\frac{1}{m}
\sum_{i=1}^m 
\nabla F_i(\bz_i^t)
}^2
.
\end{align*}
Now, everything:
\begin{align*}
\expect{F(\bar{\bz}^{t+1}) - F(\bar{\bz}^t)\Big|\calF^t}
&\le
- \frac{\eta_l \eta_g s}{4} \norm{\nabla F(\bar{\bz}^t)}^2 \\
&~~~
- \frac{\eta_l \eta_g s}{2} \pth{1 - 4 L \eta_l \eta_g s}
\norm{\frac{1}{m} \sum_{i=1}^m \nabla F_i(\bz_i^t)}^2 \\
&~~~
+\frac{2 \eta_l^2 \eta_g^2 s L \delta_{\max} \sigma^2}{m^2} 
\sum_{i=1}^m
\sum_{p=-1}^{t-1}
\indc{\tau_i(t) = p}
(t - p)^2
\\
&~~~+
\frac{8 \eta_g \eta_l^3 s^2 L^2 
\pth{1 + \eta_g \eta_l s L}
\sigma^2
}{m} \sum_{i=1}^m 
\sum_{p=-1}^{t-1} \indc{\tau_i(t) = p}
(t - p)^2  \\
&~~~
+ 2 \eta_l \eta_g s L^2 \pth{1 + \eta_l \eta_g s L}
\frac{1}{m}
\sum_{i=1}^m
\norm{\x_i^t - \bz_i^t}^2 %
+
\frac{\eta_l \eta_g s L^2}{2m}
\sum_{i=1}^m
\norm{\bz_i^t - \bar{\bz}^t}^2 \\
&~~~
+
32 \eta_g \eta_l^3 s^3 L^2
\pth{
1+
\eta_g \eta_l s L
}
\frac{1}{m} 
\sum_{i=1}^m 
\sum_{p=-1}^{t-1}
\indc{\tau_i(t) = p}
(t - p)^2
\norm{\nabla F_i(\x_i^{p+1})}^2 \\
&\le
- \frac{\eta_l \eta_g s}{4} \norm{\nabla F(\bar{\bz}^t)}^2 
+\frac{2 \eta_l^2 \eta_g^2 s L \delta_{\max} \sigma^2}{m^2} 
\sum_{i=1}^m
\sum_{p=-1}^{t-1}
\indc{\tau_i(t) = p}
(t - p)^2
\\
&~~~+
\frac{9
\eta_g \eta_l^3 s^2 L^2 
\sigma^2
}{m} \sum_{i=1}^m 
\sum_{p=-1}^{t-1} \indc{\tau_i(t) = p}
(t - p)^2 \\
&~~~
+ 
2.2 \eta_l \eta_g s L^2 
\frac{1}{m}
\sum_{i=1}^m
\norm{\x_i^t - \bz_i^t}^2
+
\frac{\eta_l \eta_g s L^2}{2m}
\sum_{i=1}^m
\norm{\bz_i^t - \bar{\bz}^t}^2 \\
&~~~
+
35 \eta_g \eta_l^3 s^3 L^2
\frac{1}{m} 
\sum_{i=1}^m 
\sum_{p=-1}^{t-1}
\indc{\tau_i(t) = p}
(t - p)^2
\norm{\nabla F_i(\x_i^{p+1})}^2 
,
\end{align*}
where the last inequality holds because 
$\eta_l \eta_g \le \frac{9}{100 s L}$ 
and that $\norm{\frac{1}{m} \sum_{i=1}^m \nabla F_i(\bz_i^t)}^2 \ge 0$.
\end{proof}

\newpage
\section{Intermediate Results}
\label{app: intermediate without pseudo}
In this section,
we present the intermediate results that serve as handy tools in building up our proofs afterwards.

\subsection{Bounding local and global dissimilarity}

\begin{proposition}%
\label{prop: average gradient to global gradient}
For any $t$, it holds that 
\begin{align*}
\frac{1}{m}\sum_{i=1}^m\norm{\nabla F_i(\bz_i^t)}^2 \le \frac{3L^2}{m} \sum_{i=1}^m \norm{\bz_i^t - \bar{\bz}^t}^2 + 3\pth{\beta^2 + 1}\norm{\nabla F(\bar{\bz}^t)}^2 + 3\zeta^2. 
\end{align*}
\end{proposition}
\begin{proof}[\bf Proof of Proposition \ref{prop: average gradient to global gradient}]
\begin{align*}
\frac{1}{m}\sum_{i=1}^m\norm{\nabla F_i(\bz_i^t)}^2  &=  \frac{1}{m}\sum_{i=1}^m\norm{\nabla F_i(\bz_i^t) - \nabla F_i(\bar{\bz}^t) + \nabla F_i(\bar{\bz}^t) - \nabla F(\bar{\bz}^t) + \nabla F(\bar{\bz}^t)}^2\\
& \le  \frac{3}{m}\sum_{i=1}^m \norm{\nabla F_i(\bz_i^t) - \nabla F_i(\bar{\bz}^t)}^2 +  \frac{3}{m}\sum_{i=1}^m \norm{\nabla F_i(\bar{\bz}^t) - \nabla F(\bar{\bz}^t)}^2 + 3 \norm{\nabla F(\bar{\bz}^t)}^2\\
& \overset{(a)}{\le} \frac{3L^2}{m} \sum_{i=1}^m \norm{\bz_i^t - \bar{\bz}^t}^2  + 3 \beta^2 \norm{\nabla F(\bar{\bz}^t)}^2 + 3\zeta^2 + 3 \norm{\nabla F(\bar{\bz}^t)}^2\\
& = \frac{3L^2}{m} \sum_{i=1}^m \norm{\bz_i^t - \bar{\bz}^t}^2 + 3\pth{\beta^2 + 1}\norm{\nabla F(\bar{\bz}^t)}^2 + 3\zeta^2, 
\end{align*} 
where inequality (a) follows from Assumptions \ref{ass: 2 smmothness} and \ref{ass: bounded similarity}.
\end{proof}

\subsection{Weight re-equalization (\prettyref{prop: similar speed})}
\begin{proof}[\bf Proof of~\prettyref{prop: similar speed}]
We show~\prettyref{prop: similar speed} by induction.

When $T=1$ and $i \in \calA^0$, we have
$
\sum_{t=0}^{0} \indc{i \in \calA^t} \pth{t - \tau_i(t)} = \indc{i \in \calA^0} \pth{0 - \tau_i(0)} = 1.
$
Therefore, the base case holds.

The induction hypothesis is that $\sum_{t=0}^{K-1} \indc{i \in \calA^t} \pth{t - \tau_i(t)} = K$ holds for $i \in \calA^{K-1}$. Next, we focus on $K+1$:
\begin{align}
\sum_{t=0}^{K} \indc{i \in \calA^t} \pth{t - \tau_i(t)} &= 
\sum_{t=0}^{K-1} \indc{i \in \calA^t} \pth{t - \tau_i(t)} + \indc{i \in \calA^K} \pth{K - \tau_i(K)}.
\label{eq: induction without pseudo}
\end{align}
Now, we have two cases:
\begin{itemize}[leftmargin=*]
\item Suppose $i \in \calA^{K-1}$, then we simply have $\tau_i(K) = K-1$. It follows that
$\text{Eq}.~\eqref{eq: induction without pseudo} \overset{(a)}{=} K + 1$,
where $(a)$ follows from induction hypothesis.
\item Suppose $i \notin \calA^{K-1}$, 
\begin{align*}
\sum_{t=0}^{K} \indc{i \in \calA^t} \pth{t - \tau_i(t)} &\overset{(b)}{=} 
\sum_{t=0}^{\tau_i(K)} \indc{i \in \calA^t} \pth{t - \tau_i(t)} + \indc{i \in \calA^K} \pth{K - \tau_i(K)} \\
& = \tau_i(K)+1 +  \pth{K - \tau_i(K)} = K + 1,
\end{align*}
where $(b)$ follows because $\indc{i \in  \calA^t} = 0$ for $ \tau_i(K) \le t \le K-1$ and induction hypothesis that $\sum_{t=0}^{\tau_i(K)} \indc{i \in \calA^t} \pth{t - \tau_i(t)} = \tau_i(K)+1$ for $i \in \calA^{\tau_i(K)}$.
\end{itemize}
\end{proof}

\subsection{Unavailable statistics (\prettyref{lmm: geo second moment main text})}
\begin{proof}
[\bf Proof of~\prettyref{lmm: geo second moment main text}]
\begin{align*}
\expect{t - \tau_i(t)}
&=
\sum_{r = 0}^{t}
\prob{t - \tau_i(t) > r}
=
\sum_{r = 0}^{t}
\prod_{r_1 = t-r}^{t-1} \pth{1 - p_i^{r_1}} 
\le 
\sum_{r = 0}^{t}
(1 - \delta)^r
\le
\frac{1}{\delta}.
\end{align*}

From \cite[Section 12, Theorem 12.3 (i)]{gut2006probability},
we know that 
\[
    \expect{g(X)} 
    =
    g(0) + \int_0^\infty g^\prime (x) \prob{X > x} \mathrm{d} x
    ,
\]
where $X$ is a non-negative random variable, and $g$ a non-negative strictly increasing differentiable function.
It follows that,
\begin{align*}
    \expect{ X^2 } &\le
    0 + 2 \int_0^\infty x \prob{X > x} \mathrm{d} x
    =
    2 \sum_{n=1}^\infty
    \int_{n-1}^n x \prob{X > x} \mathrm{d} x \\
    &\overset{(a)}{\le}
    2 \sum_{n=1}^\infty
    n
    \int_{n-1}^n \prob{X > x} \mathrm{d} x \\
    &\overset{(b)}{\le}
    2 \sum_{n=1}^\infty
    n \prob{X > n-1}
    \int_{n-1}^n  \mathrm{d} x
    =
    2 \sum_{n=1}^\infty
    n \prob{X > n-1}
    ,
\end{align*}
where 
inequality $(a)$ holds because $x\le n,~\forall x \in (n-1,n]$,
inequality $(b)$ holds because CCDF $\prob{X>x}$ is non-increasing.
In particular, for a discrete random variable,
we have $ \prob{X > n-1} = \prob{X \ge n}$.

Therefore,
\begin{align*}
    \expect{\pth{t - \tau_i(t)}^2}
    &\le
    2 \sum_{n=1}^\infty
    n \prob{t - \tau_i(t) \ge n}
    \le
    2 \sum_{n=1}^\infty
    n (1 - \delta)^{n-1}
    \le
    \frac{2}{\delta^2}.
\end{align*}
\end{proof}

\subsection{Auxiliary sequence construction and properties (\prettyref{prop: client dis})}
\begin{proposition}
\label{prop: x z: inactive and active}     
For any $t\ge 0$, when $i\notin\calA^t$, it holds that 
$\x_i^{t+1} - \bz_i^{t+1} = 
\eta_l \eta_g s (t-\tau_i(t+1) )
\nabla F_i(\x_i^{\tau_{i}(t+1)+1})$;
when $i\in\calA^t$, it holds that 
$\bz_i^{t\dagger} = \x_i^{t\dagger},~ \bz^{t+1} = \x^{t+1},$ and $\bz_i^{t+1} = \x_i^{t+1}$. 
\end{proposition}
\begin{proof}[\bf Proof of Proposition \ref{prop: x z: inactive and active}]
The proof is divided into two parts: 
$i \notin \calA^t$ and $i \in \calA^t$,

\paragraph{When $ i \notin \calA^t$.}

It holds that
\begin{align*}
\x_i^{t+1} - \bz_i^{t+1} 
&=
\x_i^{\tau_i(t+1)+1} - \qth{ \bz_i^{\tau_i(t+1)+1} - \eta_l \eta_g s \sum_{k=\tau_i(t+1)+1}^{t} \nabla F_i(\x_i^{k})} \\
&\overset{(a)}{=}
\x_i^{\tau_i(t+1)+1} - \qth{ \x_i^{\tau_i(t+1)+1} - \eta_l \eta_g s \sum_{k=\tau_i(t+1)+1}^{t} \nabla F_i(\x_i^{\tau_i(t+1)+1})} \\ 
&=
\eta_l \eta_g s (t - \tau_i(t+1)) \nabla F_i(\x_i^{\tau_i({t+1)+1}}), 
\end{align*}
where equality (a) follows from~\prettyref{def: auxiliary sequence} for inactive clients.

\paragraph{When $i \in \calA^t$.}

Note that if $\bz_i^{t++} = \x_i^{t++}$ for each $i\in \calA^t$, then by the aggregation rules, we know $\x^{t+1} = \pth{1/{\abth{\calA^t}}} \sum_{i\in\calA^t} \x_i^{t++} = \pth{1/{\abth{\calA^t}}} \sum_{i\in\calA^t} \bz_i^{t++} = \bz^{t+1}.$ 
Then, we know that $\x_i^{t+1} = \bz_i^{t+1}, ~ \forall ~ i\in \calA^t.$
Hence, to show the Proposition, it is sufficient to show $\bz_i^{t++} = \x_i^{t++}$ holds for $i\in \calA^t$, which can be shown by induction. 

When $t=0$, 
\begin{align*}
\bz_i^{0++} = \bz_i^0 + 0 - \pth{\x_i^{(0,0)} - \x_i^{(0,s)}}   
= \x_i^0  - \pth{\x_i^{(0,0)} - \x_i^{(0,s)}} = \x_i^{0++}. 
\end{align*}
Thus, the base case holds. The induction hypothesis is that $\bz_i^{t++} = \x_i^{t++}, ~ \forall~ i\in \calA^t$ is true for all $t\ge 0$.  
Now, we focus on $t+1$. 
\begin{align*}
\bz_i^{(t+1)++} & = \bz_i^{t+1} + \eta_l \eta_g s \sum_{k=\tau_i(t+1) +1}^{t} \nabla F_i(\x_i^k) -(t+1- \tau_i(t+1)) \pth{\x_i^{(t+1,0)} - \x_i^{(t+1,s)}}     \\
& = \bz_i^{t+1} + \eta_l \eta_g s (t-\tau_i(t+1)) \nabla F_i(\x_i^{\tau_i(t+1)+1}) -(t+1- \tau_i(t+1)) \pth{\x_i^{(t+1,0)} - \x_i^{(t+1,s)}}  \\
& \overset{(a)}{=} \bz_i^{\tau_i(t+1)+1} 
- \eta_l \eta_g s (t -\tau_i(t+1)-1+1)  \nabla F_i(\x_i^{\tau_i(t+1)+1}) \\
& \qquad
+ \eta_l \eta_g s (t-\tau_i(t+1)) \nabla F_i(\x_i^{\tau_i(t+1)+1})
-(t+1- \tau_i(t+1)) \pth{\x_i^{(t+1,0)} - \x_i^{(t+1,s)}}  \\
& = \bz_i^{\tau_i(t+1)+1}  -(t+1- \tau_i(t+1)) \pth{\x_i^{(t+1,0)} - \x_i^{(t+1,s)}}\\
& \overset{(b)}{=} \x_i^{\tau_i(t+1)+1}  -(t+1- \tau_i(t+1)) \pth{\x_i^{(t+1,0)} - \x_i^{(t+1,s)}}\\
& = \x_i^{(t+1)++},  
\end{align*}
where equality (a) follows from the auxiliary updates $\bz_i$, 
and equality (b) holds because of the induction hypothesis and the fact that $\tau_i(t+1)<t+1$ and $i\in \calA^{\tau_i(t+1)}$.  
\end{proof}

\begin{proof}[\bf Proof of Proposition \ref{prop: client dis}]
From Propositions \ref{prop: x z: inactive and active}, we have 
\begin{align*}
\norm{\x_i^t - \bz_i^t}^2
&\le 
\norm{\eta_l \eta_g s \pth{t - \tau_i(t) - 1} \nabla F_i(\x_i^t)}^2 \\
&=
\eta_l^2
\eta_g^2
s^2 \sum_{p=-1}^{t-1} 
\indc{\tau_i(t)=p}
\pth{t - p - 1}^2 \norm{\nabla F_i(\x_i^{p+1})}^2 
.
\end{align*}
Take expectation over all the randomness

\begin{align*}
\expect{\norm{\x_i^t - \bz_i^t}^2}
&\overset{(a)}{\le}
\eta_l^2 \eta_g^2 s^2 \sum_{p=-1}^{t-1} \expect{\indc{\tau_i(t)=p}} \pth{t - p - 1}^2 \expect{\norm{\nabla F_i(\x_i^{p+1})}^2} \\
&\overset{(b)}{\le}
\eta_l^2 \eta_g^2 s^2 
\sum_{p=-1}^{t-1} 
\pth{t - p - 1}^2 
\prob{\tau_i(t)=p}
\cdot
\expect{\norm{\nabla F_i(\bz_i^{p+1})}^2} 
,
\end{align*}
where inequality $(a)$ follows because by definition $\indc{\tau_i(t)=p}$ is independent of $\norm{\nabla F_i(\x_i^{p+1})}^2$,
inequality $(b)$ follows because $\x_i^{p+1} = \bz_i^{p+1}$ from Proposition \ref{prop: x z: inactive and active}.  

\begin{align*}
\frac{1}{T}\sum_{t=0}^{T-1}
\frac{1}{m}&\sum_{i=1}^m
\expect{\norm{\x_i^t - \bz_i^t}^2}
= 
\eta_l^2 \eta_g^2 s^2 
\frac{1}{T}\sum_{t=0}^{T-1}\frac{1}{m}\sum_{i=1}^m \sum_{p=-1}^{t-1} 
\prob{\tau_i(t) = p}
\pth{t - p - 1}^2 \expect{\norm{\nabla F_i(\bz_i^{p+1})}^2} \\
&\overset{(c)}{\le}
\eta_l^2 \eta_g^2 s^2 \frac{1}{m} \sum_{i=1}^m \frac{1}{T}\sum_{t=0}^{T-1} \expect{\norm{\nabla F_i(\bz_i^{t})}^2} 
\pth{\expect{\pth{t - \tau_i(t)}^2}} \\
&\overset{(d)}{\le}
\eta_l^2 \eta_g^2 s^2
\pth{\variance}
\frac{1}{m} \sum_{i=1}^m \frac{1}{T}\sum_{t=0}^{T-1} 
\expect{\norm{\nabla F_i(\bz_i^{t})}^2} \\
&\le 
3 \eta_l^2 \eta_g^2 s^2 
\pth{\variance}
\pth{\beta^2 + 1}
\frac{1}{T}\sum_{t=0}^{T-1} 
\expect{\norm{\nabla F(\bar{\bz}^{t})}^2} %
+%
3 \eta_l^2 \eta_g^2 s^2 
\pth{\variance}
\zeta^2\\
&~~~+
3 \eta_l^2 \eta_g^2 s^2 
L^2 \pth{\variance}
\frac{1}{m}
\sum_{i=1}^m \frac{1}{T}\sum_{t=0}^{T-1} 
\expect{\norm{\bz_i^{t} - \bar{\bz}^t}^2} 
,
\end{align*}
where inequality $(c)$ follows from re-indexing, 
inequality $(d)$ from \prettyref{lmm: geo second moment main text}.
\end{proof}

\subsection{Consensus error of the auxiliary sequence}
\begin{lemma}[Consensus error of $\bz_i^t$]
    \label{lmm: consensus z without pseudo}
    Assuming that 
    $\eta_l \le \delta /(20 s L )$,
    and $\eta_l \eta_g \le \delta(1 - \sqrt{\rho}) / ( 10 s L (\sqrt{\rho}+1))$,
    under~\prettyref{ass: 2 smmothness},
    \ref{ass: bounded variance client-wise}
    and \ref{ass: bounded similarity},
    it holds that
    \begin{align*}
    \frac{1}{m}
    \sum_{i=1}^m
    \frac{1}{T}
    \sum_{t=0}^{T-1}
    \sum_{i=1}^m
    \expect{\norm{\bz_i^t - \bar{\bz}^t}^2}
    &\le
    \frac{3 \rho s \eta_l^2 \eta_g^2}{(1 - \sqrt{\rho})^2 \delta^2} 
    \sigma^2 \\
    &~~~+
    \frac{40 \rho s^2 \eta_l^2 \eta_g^2}{(1 - \sqrt{\rho})^2}
    \zeta^2 \\
    &~~~+
    \frac{40 \rho s^2 \eta_l^2 \eta_g^2 
    \pth{\beta^2 + 1}}{(1 - \sqrt{\rho})^2}
    \frac{1}{T}\sum_{t=0}^{T-1}
    \expect{\norm{\nabla F(\bar{\bz}^t)}^2}.
\end{align*}  
\end{lemma}
\begin{proof}[\bf Proof of \prettyref{lmm: consensus z without pseudo}]
When $t=0$,  $ \bm{Z}^0 = [\bz^0, \cdots, \bz^0]$, which immediately leads to 
\begin{align*}
 \bm{Z}^0 \pth{\identity - \allones}  = [\bz^0, \cdots, \bz^0] - [\bz^0, \cdots, \bz^0] = \bm{0}.   
\end{align*}

For $t\ge 1$, 
recall that $W^{(t)}$ is a doubly stochastic matrix to characterize the information mixture,
and $\tilde{\bm{G}}^t$, 
defined in~\eqref{eq: auxiliary update detail},
captures the local parameter changes in each round. 
It can be seen that  
\begin{align*}
\bm{Z}^{(t)} = 
\pth{\bm{Z}^{(t-1)} - \eta_l \eta_g \tilde{\bm{G}}^{t-1}}
W^{(t-1)}.  
\end{align*}
Expanding $\bm{Z}$, we get 
\begin{align*}
\bm{Z}^{(t)} \pth{\identity - \allones} 
&= (\bm{Z}^{(t-1)} - \eta_l \eta_g \tilde{\bm{G}}^{t-1}) W^{(t-1)} \pth{\identity - \allones}\\
& = \bm{Z}^0 \prod_{\ell=0}^{t-1} W^{\ell} \pth{\identity - \allones} - \eta_l \eta_g \sum_{q=0}^{t-1} \tilde{\bm{G}}^{q} 
\prod_{\ell=q}^{t-1} W^{(\ell)} \pth{\identity- \allones}. 
\end{align*}
where the last follows from the fact that all clients are initiated at the same weights.
Note that $\prod_{\ell=q}^{t-1} W^{(\ell)} \identity = \prod_{\ell=q}^{t-1} W^{(\ell)}$ 
and $\prod_{\ell=q}^{t-1} W^{(\ell)} \allones =  \allones$. Thus, 
\begin{align*}
\bm{Z}^{(t)} \pth{\identity - \allones}  
= \bm{Z}^0 \pth{\prod_{\ell=0}^{t-1} W^{\ell} - \allones} - \eta_l \eta_g \sum_{q=0}^{t-1} \tilde{\bm{G}}^{q} \pth{\prod_{\ell=q}^{t-1} W^{(\ell)}  - \allones} 
= - \eta_l \eta_g \sum_{q=0}^{t-1} \tilde{\bm{G}}^{q} \pth{\prod_{\ell=q}^{t-1} W^{(\ell)} - \allones}, 
\end{align*}
where the last equality holds because that $ \bm{Z}^0 = [\bz^0, \cdots, \bz^0]$, which immediately leads to 
\begin{align*}
 \bm{Z}^0 \pth{\prod_{\ell=0}^{t-1} W^{\ell} - \allones}  = [\bz^0, \cdots, \bz^0] - [\bz^0, \cdots, \bz^0] = \bm{0}.   
\end{align*}

Let matrix notations 
${\tilde \Delta}^t$,
$\Delta^t$ and 
$\nabla \bm{F}_{\x}^t$ 
define as follows:
\begin{align*}
    \bm{G}^{q}_i
    &=
    \underbrace{
    \indc{i \in \calA^t}
    (t - \tau_i(t))
    \sum_{r=0}^{s-1} \pth{\nabla \ell_i(\x_i^{(t,r)};\xi_i^{(t,r)}) - \nabla F_i(\x_i^{(t,r)})}}_{
        [\TDt{t}]_i
    }%
    +%
    \underbrace{
    \indc{i \in \calA^t}
    (t - \tau_i(t)) 
    \sum_{r=0}^{s-1}
    \pth{
    \nabla F_i(\x_i^{(t,r)})
    -
    \nabla F_i(\x_i^t) 
    }}_{
        [\Dt{t}]_i
    }  \\
    &~~~+
    s
    \underbrace{\nabla F_i(\x_i^t)}_{
        [\DFt{t}]_i
    }.
\end{align*}

It follows that
\begin{align*}
\fnorm{\bm{Z}^{(t)} \pth{\identity - \allones}}^2   
& =
\fnorm{
\sum_{q=0}^{t-1} 
\pth{\TDt{q} + \Dt{q} + \DFt{q}}
\pth{\prod_{\ell=q}^{t-1} W^{\pth{\ell}} - \allones}}^2 \\
&=
\fnorm{
\sum_{q=0}^{t-1} 
\TDt{q}
\pth{\prod_{\ell=q}^{t-1} W^{\pth{\ell}} - \allones}}^2 
+
\fnorm{
\sum_{q=0}^{t-1} 
\pth{\Dt{q} + \DFt{q}}
\pth{\prod_{\ell=q}^{t-1} W^{\pth{\ell}} - \allones}
}^2 \\
&~~~+
2 \left \langle
\sum_{q=0}^{t-1} 
\TDt{q}
\pth{\prod_{\ell=q}^{t-1} W^{\pth{\ell}} - \allones},
\sum_{q=0}^{t-1} 
\pth{\Dt{q} + \DFt{q}}
\pth{\prod_{\ell=q}^{t-1} W^{\pth{\ell}} - \allones}
\right \rangle_{\rm F}.
\end{align*}
Take expectation with respect to randomness in stochastic gradients, 
denote by $\expects{\cdot}{\xi}$:
\begin{align*}
&\expects{\fnorm{\bm{Z}^{(t)} \pth{\identity - \allones}}^2}{\xi}
=
\expects{\fnorm{
\sum_{q=0}^{t-1} 
\TDt{q}
\pth{\prod_{\ell=q}^{t-1} W^{\pth{\ell}} - \allones}}^2 }{\xi}
+
\expects{
\fnorm{
\sum_{q=0}^{t-1} 
\pth{\Dt{q} + \DFt{q}}
\pth{\prod_{\ell=q}^{t-1} W^{\pth{\ell}} - \allones}
}^2}{\xi} \\
&\qquad \qquad +
2 \expects{
\left \langle
\sum_{q=0}^{t-1} 
\TDt{q}
\pth{\prod_{\ell=q}^{t-1} W^{\pth{\ell}} - \allones},
\sum_{q=0}^{t-1} 
\pth{\Dt{q} + \DFt{q}}
\pth{\prod_{\ell=q}^{t-1} W^{\pth{\ell}} - \allones}
\right \rangle_{\rm F}}{\xi} \\
&=
\expects{\fnorm{
\sum_{q=0}^{t-1} 
\TDt{q}
\pth{\prod_{\ell=q}^{t-1} W^{\pth{\ell}} - \allones}}^2 }{\xi}
+
\expects{
\fnorm{
\sum_{q=0}^{t-1} 
\pth{\Dt{q} + \DFt{q}}
\pth{\prod_{\ell=q}^{t-1} W^{\pth{\ell}} - \allones}
}^2}{\xi} \\
&~~~
\qquad \qquad
+2
\left \langle
\sum_{q=0}^{t-1} 
 \expects{\TDt{q}}{\xi}
\pth{\prod_{\ell=q}^{t-1} W^{\pth{\ell}} - \allones},
\sum_{q=0}^{t-1} 
\pth{\Dt{q} + \DFt{q}}
\pth{\prod_{\ell=q}^{t-1} W^{\pth{\ell}} - \allones}
\right \rangle_{\rm F} \\
&\le
\expects{\fnorm{
\sum_{q=0}^{t-1} 
\TDt{q}
\pth{\prod_{\ell=q}^{t-1} W^{\pth{\ell}} - \allones}}^2 }{\xi}
+
\expects{
\fnorm{
\sum_{q=0}^{t-1} 
\pth{\Dt{q} + \DFt{q}}
\pth{\prod_{\ell=q}^{t-1} W^{\pth{\ell}} - \allones}
}^2}{\xi},
\end{align*}
where the last inequality holds because $\expects{\TDt{q}}{\xi} = 0$.
Next, we take expectation over the remaining randomness.
\begin{align}
\notag
\expect{\fnorm{\bm{Z}^{(t)} \pth{\identity - \allones}}^2}
&\le
\expect{\fnorm{
\sum_{q=0}^{t-1} 
\TDt{q}
\pth{\prod_{\ell=q}^{t-1} W^{\pth{\ell}} - \allones}}^2}
+
\expect{
\fnorm{
\sum_{q=0}^{t-1} 
\pth{\Dt{q} + \DFt{q}}
\pth{\prod_{\ell=q}^{t-1} W^{\pth{\ell}} - \allones}
}^2} \\\notag
&
\le \eta_l^2 \eta_g^2 \underbrace{\fnorm{
\sum_{q=0}^{t-1} 
    \TDt{q}
\pth{\prod_{\ell=q}^{t-1} W^{\pth{\ell}} - \allones}
}^2}_{(\rmI)}\\
&~~
\nonumber
+ 2\eta_l^2 \eta_g^2 \underbrace{\fnorm{
    \sum_{q=0}^{t-1}
    \Dt{q}
\pth{\prod_{\ell=q}^{t-1} W^{\pth{\ell}} - \allones}}^2}_{(\rmI\rmI)}\\
&~~~
+ 2\eta_l^2 \eta_g^2 s^2 \underbrace{\fnorm{\sum_{q=0}^{t-1}
    \DFt{q}
\pth{\prod_{\ell=q}^{t-1} W^{\pth{\ell}} - \allones}}^2}_{(\rmI\rmI\rmI)}. 
\label{eq: conseneus iterative error without pseudo}
\end{align}

\paragraph{Bounding $\expect{(\rmI)}$}
\begin{align}
\label{eq: conseneus iterative error 1 without pseudo}
    \nonumber
    \expect{(\rmI) } &= \sum_{q=0}^{t-1} \expect{\fnorm{ 
        \TDt{q}
    \pth{\prod_{\ell=q}^{t-1} W^{\pth{\ell}} - \allones}
    }^2 } \\
     \notag
     &\qquad \qquad + \sum_{q=0}^{t-1} \sum_{p=0, p\neq q}^{t-1} \expect{\iprod{
        \TDt{p}
     \pth{\prod_{\ell=p}^{t-1}W^{\pth{\ell}} - \allones}}{
        \TDt{q}
    \pth{\prod_{\ell=q}^{t-1} W^{\pth{\ell}} - \allones}
    } }\\
    &\overset{(a)}{\le} \sum_{q=0}^{t-1} \rho^{t-q}
    \expect{\fnorm{
        \TDt{q}
    }^2 }
    , 
\end{align}
where inequality $(a)$ holds because of~\prettyref{ass: bounded variance client-wise}.
It remains to bound 
$\expect{\fnorm{
    \TDt{q}
}^2}$.
\begin{align*}
\fnorm{
    \TDt{q}
}^2
&=
\sum_{i=1}^m
\indc{i \in \calA^q}
\norm{
\sum_{p=-1}^{q-1}
\indc{\tau_i(t) = p}
(q - p)
\sum_{r=0}^{s-1} \pth{
 \nabla \ell_i(\x_i^{(q,r)};\xi_i^{(q,r)}) - \nabla F_i(\x_i^{(q,r)})}
 }^2.
\end{align*}

\begin{align*}
    \expects{\fnorm{\TDt{q}}^2}{\xi}
    &=
    \sum_{i=1}^m
    \indc{i \in \calA^q}
    \sum_{p=-1}^{q-1}
    \indc{\tau_i(t) = p}
    (q - p)^2
    \sum_{r=0}^{s-1} 
    \expects{\norm{
    \nabla \ell_i(\x_i^{(q,r)};\xi_i^{(p,r)}) - \nabla F_i(\x_i^{(q,r)})
    }^2}{\xi} \\
    &\le    
    s \sigma^2
    \sum_{i=1}^m
    \indc{i \in \calA^q}
    \sum_{p=-1}^{q-1}
    \indc{\tau_i(t) = p}
    (q - p)^2.
\end{align*}

Take expectation over the remaining randomness:
\begin{align*}
    \expect{\fnorm{\TDt{q}}^2}
    &=
    \expect{\expects{\fnorm{\TDt{q}}^2}{\xi}}
    \le    
    s \sigma^2
    \sum_{i=1}^m
    \expect{\indc{i \in \calA^q}}
    \sum_{p=-1}^{q-1}
    \expect{\indc{\tau_i(t) = p}}
    (q - p)^2 
    \le
    \frac{2 m s \sigma^2}{\delta^2} 
\end{align*}

Therefore,
\begin{align*}
    \frac{1}{m T}
    \sum_{i=1}^m
    \sum_{t=0}^{T-1}
    \expect{(\rmI)}
    &\le
    \frac{s \rho}{\pth{1 - \rho}}
    \pth{\variance}
    \sigma^2.
\end{align*}

\paragraph{Bounding $\expect{(\rmI\rmI)}$}
\begin{align*}
    &\expect{(\rmI\rmI)}
    =
    \expect{\fnorm{\sum_{q=0}^{t-1}
    \Dt{q}
    \pth{\prod_{\ell=q}^{t-1} W^{\pth{\ell}} - \allones}}^2} \\
    &= \sum_{q=0}^{t-1} \expect{\fnorm{ 
    \Dt{q}
    \pth{\prod_{\ell=q}^{t-1} W^{\pth{\ell}} - \allones}
    }^2 } 
     + \sum_{q=0}^{t-1} \sum_{p=0, p\neq q}^{t-1} 
     \expect{
     \iprod{
     \Dt{p} 
     \pth{\prod_{\ell=p}^{t-1}W^{\pth{\ell}} - \allones}}{
     \Dt{q}
    \pth{\prod_{\ell=q}^{t-1} W^{\pth{\ell}} - \allones}
    } }\\
    &\le 
    \sum_{q=0}^{t-1} \rho^{t-q}\expect{\fnorm{\Dt{q}}^2} 
    +\sum_{q=0}^{t-1} \sum_{p=0, p\neq q}^{t-1}\expect{ 
    \fnorm{\Dt{p}
    \pth{\prod_{\ell=p}^{t-1} W^{\pth{\ell}} - \allones}}
    \fnorm{\Dt{q}
    \pth{\prod_{\ell=q}^{t-1} W^{\pth{\ell}} - \allones}
    } }
    \\
    &\le 
    \sum_{q=0}^{t-1} \rho^{t-q}
    \expect{\fnorm{\Dt{q}}^2 } %
    +\sum_{q=0}^{t-1} \sum_{p=0, p\neq q}^{t-1} \expect{
    {
    \frac{\rho^{t-p}}{2\epsilon}
    \fnorm{\Dt{p}}^2 
    +
    \frac{\epsilon \rho^{t-q}}{2}
    \fnorm{\Dt{q}}^2 } }, 
\end{align*}

Next, we bound the second term, choose $\epsilon = \rho^{\frac{q-p}{2}},$
\begin{align}
\label{eq: conseneus  error 1: auxiliary without pseudo}
\nonumber&
\sum_{q=0}^{t-1} \sum_{p=0, p\neq q}^{t-1}  \frac{\sqrt{\rho}^{2t-p-q}}{2}\expect{
{
\fnorm{\Dt{p}}^2 
+
\fnorm{\Dt{q}}^2 } } %
\le \sum_{q=0}^{t-1} \sum_{p=0}^{t-1} \frac{\sqrt{\rho}^{2t-p-q}}{2}\expect{
{
\fnorm{\Dt{p}}^2 
+
\fnorm{\Dt{q}}^2 }}\\
\nonumber
& = \sum_{p=0}^{t-1}  
\frac{\sqrt{\rho}^{t-p}}{2}\expect{
\fnorm{\Dt{p}}^2 }
\sum_{q=0}^{t-1}\sqrt{\rho}^{t-q}
+
\sum_{q=0}^{t-1}  \frac{\sqrt{\rho}^{t-q}}{2}\expect{
\fnorm{\Dt{q}}^2 }
\sum_{p=0}^{t-1}
\sqrt{\rho}^{t-p}\\ 
& = \frac{\sqrt{\rho} - \sqrt{\rho}^{t+1}}{1-\sqrt{\rho}}\sum_{q=0}^{t-1}  \sqrt{\rho}^{t-q}\expect{
\fnorm{\Dt{q}}^2 }.
\end{align}

Plugging the upper bound in~\eqref{eq: conseneus  error 1: auxiliary without pseudo} into~\eqref{eq: conseneus iterative error 1 without pseudo}, 
we get
\begin{align}
\nonumber
\expect{(\rmI\rmI)} 
& \le\sum_{q=0}^{t-1} 
\qth{\sqrt{\rho}^{t-q} + \frac{\sqrt{\rho} - \sqrt{\rho}^{t+1}}{1-\sqrt{\rho}}}
\sqrt{\rho}^{t-q}\expect{
\fnorm{\Dt{q}}^2 }
\overset{(b)}{\le}\sum_{q=0}^{t-1} 
\qth{\frac{\sqrt{\rho} + \sqrt{\rho}}{1-\sqrt{\rho}}}
\sqrt{\rho}^{t-q}\expect{\fnorm{
\Dt{q}
}^2 }\\
\label{eq: consensus II without pseudo}
& \le \frac{2\sqrt{\rho}}{1-\sqrt{\rho}}\sum_{q=0}^{t-1} 
\sqrt{\rho}^{t-q}\expect{\fnorm{\Dt{q}}^2 },
\end{align}
where inequality $(b)$ follows because that $\sqrt{\rho}^{t-q} \le \sqrt{\rho}$ for any $q\le t-1$, and that $\sqrt{\rho}^{t+1}\ge 0$. 
It remains to bound $\expect{\fnorm{\Dt{q}}^2}$.
Take expectation with respect to randomness in stochastic gradients:
\begin{align*}
\expects{\fnorm{\Dt{q}}^2}{\xi}
&\le
4 \eta_l^2 s^3 L^2 
\sum_{i=1}^m
\sum_{p=-1}^{q-1}
\indc{\tau_i(q) = p}
(q - p)^2
\sigma^2
\\
&~~~+
16 \eta_l^2 s^4 L^2
\sum_{i=1}^m
\sum_{p=-1}^{q-1}
\indc{\tau_i(q) = p}
(q - p)^2
\norm{\nabla F_i(\x_i^q)}^2,
\end{align*}
where the inequality holds due to~\prettyref{lmm: multi-local steps}.
Next, we take expectation over the remaining randomness and plug back into~\eqref{eq: consensus II without pseudo}:
\begin{align*}
    &\expect{(\rmI\rmI)}\le \frac{2\sqrt{\rho}}{1-\sqrt{\rho}}\sum_{q=0}^{t-1} 
    \sqrt{\rho}^{t-q}\expect{\fnorm{\Dt{q}}^2 } \\
    &\le 
    \frac{8 \rho}{\pth{1 - \sqrt{\rho}}^2}
    \pth{\variance}
    \eta_l^2 s^3 L^2
    m \sigma^2 \\
    &~~~+
    \frac{32 \sqrt{\rho}}{1 - \sqrt{\rho}}
    \pth{\variance}
    \eta_l^2 s^4 L^2
    \sum_{i=1}^m  \sum_{q=0}^{t-1} \expect{\norm{\nabla F_i(\x_i^{q})}^2} 
    \sum_{k=1}^{T-1-t}
    \sqrt{\rho}^k \\
    &\le
    \frac{8 \rho}{\pth{1 - \sqrt{\rho}}^2}
    \pth{\variance}
    \eta_l^2 s^3 L^2
    m \sigma^2 %
    +
    \frac{32 \rho}{\pth{1 - \sqrt{\rho}}^2}
    \pth{\variance}
    \eta_l^2 s^4 L^2
    \sum_{i=1}^m  \sum_{q=0}^{t-1} \expect{\norm{\nabla F_i(\x_i^{q})}^2} 
    ,
\end{align*}
where the last inequality holds because of re-index and grouping.
Therefore,
\begin{align*}
    &\frac{1}{m T}
    \sum_{t=1}^{T-1}
    \expect{(\rmI \rmI)} 
    \le
    \frac{8 \rho}{\pth{1 - \sqrt{\rho}}^2}
    \pth{\variance}
    \eta_l^2 s^3 L^2
    \sigma^2 \\
    &~~~+
    \frac{32 \rho}{\pth{1 - \sqrt{\rho}}^2}
    \pth{\variance}
    \eta_l^2 s^4 L^2
    \frac{1}{T}
    \sum_{t=1}^{T-1}
    \frac{1}{m}
    \sum_{i=1}^m  
    \expect{\norm{\nabla F_i(\x_i^{t})}^2}  \\
    &\le
    \frac{8 \rho}{\pth{1 - \sqrt{\rho}}^2}
    \pth{\variance}
    \eta_l^2 s^3 L^2
    \sigma^2 %
    +
    \frac{64 \rho}{\pth{1 - \sqrt{\rho}}^2}
    \pth{\variance}
    \eta_l^2 s^4 L^4
    \frac{1}{T}
    \sum_{t=1}^{T-1}
    \frac{1}{m}
    \sum_{i=1}^m  
    \expect{\norm{\x_i^t - \bz_i^t}^2} \\
    &~~~+
    \frac{64 \rho}{\pth{1 - \sqrt{\rho}}^2}
    \pth{\variance}
    \eta_l^2 s^4 L^2
    \frac{1}{T}
    \sum_{t=1}^{T-1}
    \frac{1}{m}
    \sum_{i=1}^m 
    \expect{\norm{\nabla F_i(\bz_i^{t})}^2}
\end{align*}

\paragraph{Bounding $\expect{(\rmI\rmI\rmI)}$}
Use a similar trick as in bounding $\expect{(\rmI \rmI)},$ and we get
\begin{align*}
\expect{(\rmI\rmI\rmI)} &= \expect{\fnorm{\sum_{q=0}^{t-1}
    \DFt{q}
\pth{\prod_{\ell=q}^{t-1} W^{\pth{\ell}} - \allones}}^2} 
\le 
\frac{2\sqrt{\rho}}{1-\sqrt{\rho}}\sum_{q=0}^{t-1} \sqrt{\rho}^{t-q}\expect{\fnorm{
    \DFt{q}
}^2 },
\end{align*}
so that
\begin{align*}
\frac{1}{mT}\sum_{t=0}^{T-1}\expect{(\rmI\rmI\rmI)} 
&\le
\frac{2\sqrt{\rho}}{mT \pth{1-\sqrt{\rho}}}\sum_{t=0}^{T-1}\expect{\fnorm{\DFt{t}}^2} \sum_{q=1}^{T-1-t}\sqrt{\rho}^{q}\\
&\le \frac{2 \rho}{\pth{1-\sqrt{\rho}}^2} \frac{1}{mT}\sum_{t=0}^{T-1} \sum_{i=1}^m \expect{\norm{\nabla F_i(\x_i^t)}^2} \\
&\le \frac{4 \rho L^2}{\pth{1-\sqrt{\rho}}^2} \frac{1}{mT}\sum_{t=0}^{T-1} \sum_{i=1}^m \expect{\norm{\x_i^t - \bz_i^t}^2} 
+
\frac{4 \rho}{\pth{1-\sqrt{\rho}}^2} \frac{1}{mT}\sum_{t=0}^{T-1} \sum_{i=1}^m \expect{\norm{\nabla F_i(\bz_i^t)}^2}.
\end{align*}

\paragraph{Putting them together
}
\begin{align*}
&\frac{1}{mT}\sum_{t=0}^{T-1}\expect{\fnorm{\bm{Z}^{\pth{t}} \pth{\identity - \allones}}^2}
\le
\frac{s \rho \eta_l^2 \eta_g^2}{(1 - \sqrt{\rho})^2} 
\pth{\variance}
\pth{1 
+ 16 \eta_l^2 s^2 L^2}
\sigma^2 \\
&~~~+
\frac{8 \rho s^2 L^2 \eta_l^2 \eta_g^2}{(1 - \sqrt{\rho})^2}
\pth{1 + 16 \eta_l^2 s^2 L^2 \pth{\variance}}
\frac{1}{T}
\sum_{t=1}^{T-1}
\frac{1}{m}
\sum_{i=1}^m
\expect{\norm{\x_i^t - \bz_i^t}^2} \\
&~~~+
\frac{8 \rho s^2 \eta_l^2 \eta_g^2}{(1 - \sqrt{\rho})^2}
\pth{1 + 16 \eta_l^2 s^2 L^2 \pth{\variance}}
\frac{1}{T}
\sum_{t=1}^{T-1}
\frac{1}{m}
\sum_{i=1}^m
\expect{\norm{\nabla F_i (\bz_i^t)}^2}
.
\end{align*}
Plug in~\prettyref{prop: client dis}.
\begin{align*}
&\frac{1}{mT}\sum_{t=0}^{T-1}\expect{\fnorm{\bm{Z}^{\pth{t}} \pth{\identity - \allones}}^2}
\le
\frac{s \rho \eta_l^2 \eta_g^2}{(1 - \sqrt{\rho})^2} 
\pth{\variance}
\pth{1 + 20 \eta_l^2 s^2 L^2}
\sigma^2 \\
&~~~+
\frac{8 \rho s^2 \eta_l^2 \eta_g^2}{(1 - \sqrt{\rho})^2}
\pth{1 + 16 \eta_l^2 s^2 L^2 \pth{\variance}}
\pth{1 + \eta_l^2 \eta_g^2 s^2 L^2 \pth{\variance}}
\frac{1}{T}
\sum_{t=1}^{T-1}
\frac{1}{m}
\sum_{i=1}^m
\expect{\norm{\nabla F_i (\bz_i^t)}^2} \\
&\le
\frac{1.05 \rho s \eta_l^2 \eta_g^2}{(1 - \sqrt{\rho})^2} 
\pth{\variance}
\sigma^2 
+
\frac{9 \rho s^2 \eta_l^2 \eta_g^2}{(1 - \sqrt{\rho})^2}
\frac{1}{T}
\sum_{t=1}^{T-1}
\frac{1}{m}
\sum_{i=1}^m
\expect{\norm{\nabla F_i (\bz_i^t)}^2} 
,
\end{align*}
where the last inequality holds because 
$\eta_l \le \delta /(20 s L )$
and $\eta_l \eta_g \le \delta / ( 10 s L )$.
Next, plug in Proposition \ref{prop: average gradient to global gradient}.
\begin{align*}
    &\frac{1}{mT}\sum_{t=0}^{T-1}\expect{\fnorm{\bm{Z}^{\pth{t}} \pth{\identity - \allones}}^2}
    \le
    \frac{1.05 \rho s \eta_l^2 \eta_g^2}{(1 - \sqrt{\rho})^2} 
    \pth{\variance}
    \sigma^2 
    +
    \frac{27 \rho s^2 \eta_l^2 \eta_g^2}{(1 - \sqrt{\rho})^2}
    \zeta^2 \\
    &~~~+
    \frac{27 \rho s^2 \eta_l^2 \eta_g^2 \pth{\beta^2 + 1}}{(1 - \sqrt{\rho})^2}
    \frac{1}{T}\sum_{t=0}^{T-1}
    \expect{\norm{\nabla F(\bar{\bz}^t)}^2} 
    +
    \frac{27 \rho s^2 L^2 \eta_l^2 \eta_g^2}{(1 - \sqrt{\rho})^2}
    \frac{1}{T}\sum_{t=0}^{T-1}
    \frac{1}{m}\sum_{i=1}^m
    \expect{\norm{\bz_i^t - \bar{\bz}^t}^2}
    .
\end{align*}
It follows that
\begin{align*}
    \frac{1}{mT}\sum_{t=0}^{T-1}\expect{\fnorm{\bm{Z}^{\pth{t}} \pth{\identity - \allones}}^2}
    &\le
    \frac{3 \rho s \eta_l^2 \eta_g^2}{(1 - \sqrt{\rho})^2 \delta^2} 
    \sigma^2 \\
    &~~~+
    \frac{40 \rho s^2 \eta_l^2 \eta_g^2}{(1 - \sqrt{\rho})^2}
    \zeta^2 \\
    &~~~+
    \frac{40 \rho s^2 \eta_l^2 \eta_g^2 
    \pth{\beta^2 + 1}}{(1 - \sqrt{\rho})^2}
    \frac{1}{T}\sum_{t=0}^{T-1}
    \expect{\norm{\nabla F(\bar{\bz}^t)}^2}.
\end{align*}  
which is due to the fact that $\eta_l \eta_g \le \frac{1 - \sqrt{\rho}}{10 s L (\sqrt{\rho} + 1)}$.
\end{proof}

\subsection{Spectral norm upper bound (\prettyref{lmm: rho upper bound main text})}
\prettyref{lmm: rho upper bound main text} adapts from \cite{xiang2023towards},
we present its proof here for completeness.
\begin{proof}[\bf Proof of~\prettyref{lmm: rho upper bound main text}]
For ease of exposition, in this proof we drop time index $t$. 
We first get the explicit expression for $\expect{W^2_{jj^{\prime}}\mid \calA \neq \emptyset}$. 
Suppose that $\calA\not=\emptyset$. 
We have 
\begin{align*}
W^2_{jj^{\prime}} & = \sum_{k=1}^m W_{jk}W_{j^{\prime}k} 
= W_{jj}W_{j^{\prime}j} + W_{jj^{\prime}}W_{j^{\prime}j^{\prime}} + \sum_{k\in [m]\setminus \{j, j^{\prime}\}} W_{jk}W_{j^{\prime}k}.   
\end{align*}
When $k\not=j$ and $k\not=j^{\prime}$,
we have 
\begin{align*}
W_{jk}W_{j^{\prime}k} & = \frac{1}{|\calA|^2} \indc{j\in \calA} \indc{j^{\prime}\in \calA}\indc{k\in \calA}.     
\end{align*}
In addition, we have $W_{jj}W_{j^{\prime}j}  = \frac{1}{|\calA|^2} \indc{j\in \calA}\indc{j^{\prime}\in \calA},$ and $W_{j^{\prime}j^{\prime}}W_{jj^{\prime}} = \frac{1}{|\calA|^2} \indc{j\in \calA}\indc{j^{\prime}\in \calA}$.  
Thus, 
\begin{itemize}[leftmargin=*]
\item 
For  $j \neq j^\prime$, we have
\begin{align*}
& W^2_{jj^{\prime}} = \sum_{k=1}^m W_{jk}W_{j^{\prime}k} 
= \frac{1}{|\calA|}\indc{j\in \calA} \indc{j^{\prime}\in \calA};
\end{align*}
\item
For $j = j^\prime$, we have
\begin{align*}
& W^2_{jj} = \frac{1}{|\calA|}\indc{j\in \calA} + \pth{1-\indc{j\in \calA}}.
\end{align*}
\end{itemize}
In the special case where $\calA = \emptyset$, we simply have $W = \identity$ by the algorithmic clauses.
Therefore,
$\expect{W_{j j^\prime} \mid \calA = \emptyset} \ge 0$ holds for any pair of $j,j^\prime  \in [m]$.
It follows, by the law of total expectation and for all $j,j^\prime \in [m]$, that
\begin{align*}
\expect{W_{j j^\prime}}
&=
\expect{W_{j j^\prime} \mid \calA = \emptyset}
\prob{\calA = \emptyset} 
+
\expect{W_{j j^\prime} \mid \calA \neq \emptyset}
\prob{\calA \neq \emptyset} 
\\
&\ge
\expect{W_{j j^\prime} \mid \calA \neq \emptyset}
\prob{\calA \neq \emptyset}.
\end{align*}
\begin{itemize}[leftmargin=*]
\item 
For $j \neq j^\prime$, it holds that
\begin{align*}
&\expect{W^2_{jj^{\prime}}\mid \calA \neq \emptyset} 
=
\expect{\frac{1}{|\calA|}\indc{j\in \calA} \indc{j^{\prime}\in \calA} \Big| \calA \neq \emptyset} 
\overset{(a)}{\ge} 
\expect{\frac{1}{m}\indc{j\in \calA} \indc{j^{\prime}\in \calA} \Big| \calA \neq \emptyset}
= 
\frac{p_j p_{j^{\prime}}}{m} 
\ge \frac{\delta^2}{m},
\end{align*}
where inequality $(a)$ holds because $\abth{\calA} \le m$
;
\item
For $j = j^\prime$, it holds that
\begin{align*}
\expect{W^2_{jj}\mid \calA\neq\emptyset} 
&= 
\expect{\frac{1}{|\calA|}\indc{j\in \calA} + \pth{1-\indc{j\in \calA}} \Big| \calA \neq \emptyset} \\
&\ge
\expect{\frac{1}{m}\qth{\indc{j\in \calA} + \pth{1-\indc{j\in \calA}}} \Big| \calA \neq \emptyset} 
=
\frac{1}{m}
\ge
\frac{\delta^2}{m}.
\end{align*}
\end{itemize}
Recall that $M = \expect{W^2}$.
Next, we show that each element of $M$ is lower bounded.
\begin{align*}
M_{jj^{\prime}}
\ge 
\expect{W_{j j^\prime}^2\mid \calA \neq \emptyset}
\prob{\calA \neq \emptyset}
\ge \frac{\delta^2}{m} \qth{1-\pth{1-\delta}^m}.
\end{align*} 
We note that $\rho (t) = \lambda_2 (M)$,
where $\lambda_2$ is the second largest eigenvalue of %
matrix $M$. 
A Markov chain with $M$ as the transition matrix is ergodic as the chain is (1) {\it irreducible}: $M_{j j^\prime}\ge  \frac{\delta^2}{m}\qth{1-\pth{1-c}^m}>0$ for $j,j^\prime \in[m]$ and (2) {\it aperiodic} (it has self-loops).  
In addition, $W$ matrix is by definition doubly-stochastic. 
Hence, $M$ has a uniform stationary distribution
$\pi = \Indc^\top/ m.$
Furthermore, the irreducible Markov chain is reversible since 
it holds
for all the states that
$
\pi_i M_{ij} = \pi_j M_{ji}.
$
The conductance $\Phi$ of a reversible Markov chain \cite{jerrum1988conductance} with  a transition matrix $M$ can be bounded by
\begin{align*}
&\Phi(M) 
= \min_{\sum_{i\in\calS} \pi_i \le \frac{1}{2}} \frac{\sum_{i\in\calS, j\notin \calS} \pi_i M_{ij}}{\sum_{i\in \calS} \pi_i}
\ge \frac{\pth{\frac{\delta}{m}}^2\qth{1-\pth{1-\delta}^m}\abth{\calS}\abth{\bar{\calS}}}{\frac{\abth{\calS}}{m}} = \frac{\delta^2\qth{1-\pth{1-\delta}^m}}{m} \abth{\bar{\calS}},
\end{align*}
where
$
\abth{\bar{\calS}} = m - \abth{\calS} \ge \frac{m}{2}.
$
From Cheeger's inequality, we know that
$
\frac{1 - \lambda_2}{2} \le \Phi(M) \le \sqrt{2 \pth{1-\lambda_2}}.
$
Finally, we have
\begin{align*}
\Phi(M) 
&\ge \frac{\delta^2\qth{1-\pth{1-\delta}^m}}{m} \abth{\bar{\calS}} \ge \frac{\delta^2\qth{1-\pth{1-\delta}^m}}{2}.
\end{align*}
Thus,
$
\rho(t) = \lambda_2 \le 1 - \frac{\Phi^2\pth{M}}{2} \le 1 - \frac{\delta^4\qth{1-\pth{1-\delta}^m}^2}{8}.
$
\end{proof}

\newpage
\section{Convergence Error of $\bar{\bz}^t$ (\prettyref{thm: z bar rate})}

In the sequel, 
we recall and assume the following learning rate conditions in~\eqref{eq: lr condition main text}:
\begin{align*}
    \eta_l \eta_g &\le 
    \frac{\pth{1 - \sqrt{\rho}}\delta}{80 s (L+1) \pth{\sqrt{\rho} + 1}  \sqrt{
    \pth{\beta^2 + 1}(1 + L^2)}}; ~
    \eta_l \le
    \frac{\delta}{200 s L \sqrt{
    \pth{\beta^2 + 1}
    (1 + L^2)}}.
\end{align*}

Recall that $\delta_{\max} \triangleq \max_{i\in[m],t\in[T]} p_i^t$ and $F^\star \triangleq \min_{\x} F(\x)$.

\begin{proof}[\bf Proof of~\prettyref{thm: z bar rate}]
Take expectation over all the randomness, plug in \prettyref{lmm: consensus z without pseudo} and Proposition \ref{prop: client dis}.
By telescoping sum, it holds that
\begin{align}
\nonumber
&\frac{\expect{ F^{\star} - F(\bar{\bz}^0) }}{T}
\le
- \frac{\eta_l \eta_g s}{4} 
\frac{1}{T}
\sum_{t=0}^{T-1}
\expect{\norm{\nabla F(\bar{\bz}^t)}^2 }
+\frac{2 \eta_l^2 \eta_g^2 s L {\delta_{\max}} \sigma^2}{m^2 T} 
\sum_{t=0}^{T-1}
\sum_{i=1}^m
\sum_{p=-1}^{t-1}
\expect{\indc{\tau_i(t) = p}}
(t - p)^2
\\
\nonumber
&~~~+
\frac{9
\eta_g \eta_l^3 s^2 L^2 
\sigma^2
}{m T} 
\sum_{t=0}^{T-1}
\sum_{i=1}^m 
\sum_{p=-1}^{t-1} 
\expect{\indc{\tau_i(t) = p}}
(t - p)^2 \\
\label{eq: z telescope third to last without pseudo}
&~~~
+ 
2.2 \eta_l \eta_g s L^2 
\frac{1}{m T}
\sum_{t=0}^{T-1}
\sum_{i=1}^m
\expect{\norm{\x_i^t - \bz_i^t}^2} \\
\label{eq: z telescope second to last without pseudo}
&~~~
+
\frac{\eta_l \eta_g s L^2}{2m T}
\sum_{t=0}^{T-1}
\sum_{i=1}^m
\expect{\norm{\bz_i^t - \bar{\bz}^t}^2 }\\
&~~~
\label{eq: z telescope last without pseudo}
+
\frac{35 \eta_g \eta_l^3 s^3 L^2}{m T}
\sum_{t=0}^{T-1}
\sum_{i=1}^m 
\sum_{p=-1}^{t-1}
\expect{\indc{\tau_i(t) = p}}
(t - p)^2
\expect{\norm{\nabla F_i(\x_i^{p+1})}^2}
.
\end{align}

Next, we bound~\eqref{eq: z telescope third to last without pseudo},~\eqref{eq: z telescope second to last without pseudo} and~\eqref{eq: z telescope last without pseudo}, respectively.
First, we show that
\begin{align}
\notag
    &\frac{1}{mT}
    \sum_{t=0}^{T-1}
    \sum_{i=1}^m
    \expect{\norm{\nabla F_i (\bz_i^t)}^2}\\
    \notag
    &\le
    3 \zeta^2
    +
    3\pth{\beta^2 + 1}
    \frac{1}{T}
    \sum_{t=0}^{T-1}
    \expect{\norm{\nabla F (\bar{\bz}^t)}^2} 
    +
    \frac{3L^2}{m T}
    \sum_{t=0}^{T-1}
    \sum_{i=1}^m
    \expect{\norm{\bz_i^t - \bar{\bz}^t}^2} \\
\label{eq: hetero consensus}
    &\le 
    3
    \qth{1 + \frac{40 \rho s^2 \eta_l^2 \eta_g^2 L^2}{(1 - \sqrt{\rho})^2}}
    \zeta^2
    +
    3\pth{\beta^2 + 1}
    \qth{1 + \frac{40 \rho s^2 \eta_l^2 \eta_g^2 L^2}{(1 - \sqrt{\rho})^2}}
    \frac{1}{T}
    \sum_{t=0}^{T-1}
    \expect{\norm{\nabla F (\bar{\bz}^t)}^2} 
    +
    \frac{9 \rho s \eta_l^2 \eta_g^2 L^2}{(1 - \sqrt{\rho})^2 \delta^2} \sigma^2
    ,
\end{align}
where the last inequality follows from~\prettyref{lmm: consensus z without pseudo}.

For~\eqref{eq: z telescope third to last without pseudo},
we have
\begin{align*}
    2.2 \eta_l \eta_g s L^2 
    \frac{1}{T}
    \sum_{t=0}^{T-1}
    \frac{1}{m}
    \sum_{i=1}^m
    \expect{\norm{\x_i^t - \bz_i^t}^2} 
    & \le
    \frac{4.4
    \eta_l^3 \eta_g^3 s^3 L^2 }{\delta^2}
    \frac{1}{T}
    \sum_{t=0}^{T-1}
    \frac{1}{m}
    \sum_{i=1}^m
    \expect{\norm{\nabla F_i (\bz_i^t)}^2} \\
    &\le
    \frac{s^2 \eta_l^3 \eta_g^3 L^2}{2 \delta^2} 
    \sigma^2
    +
    \frac{14
    \eta_l^3 \eta_g^3 s^3 L^2
    }{\delta^2}
    \pth{1 + \frac{40 \eta_l^2 \eta_g^2 \rho s^2 L^2}{(1 - \sqrt{\rho})^2}}
    \zeta^2 \\
    &~~~+
    \frac{14
    \eta_l^3 \eta_g^3 s^3 L^2
    }{\delta^2}
    \qth{
    \pth{\beta^2 + 1} + \frac{40 \eta_l^2 \eta_g^2 \rho s^2 L^2}{(1 - \sqrt{\rho})^2}}
    \frac{1}{T}
    \sum_{t=0}^{T-1}
    \expect{\norm{\nabla F (\bar{\bz}^t)}^2} 
    ,
\end{align*}
where the last inequality holds due to~\eqref{eq: hetero consensus}.
For~\eqref{eq: z telescope second to last without pseudo},
we similarly have
\begin{align*}
    \frac{\eta_l \eta_g s L^2}{2mT}
    \sum_{t=0}^{T-1}
    \sum_{i=1}^m
    \expect{\norm{\bz_i^t - \bar{\bz}^t}^2}
    &\le
    \frac{1.5 \rho s^2 \eta_l^3 \eta_g^3 L^2}{(1 - \sqrt{\rho})^2 \delta^2}
    \sigma^2 
    +
    \frac{20 \rho s^3 \eta_l^3 \eta_g^3 L^2 }{(1 - \sqrt{\rho})^2}
    \zeta^2 \\
    &~~~+
    \frac{20 \rho s^3 \eta_l^3 \eta_g^3 L^2 \pth{\beta^2 + 1}}{(1 - \sqrt{\rho})^2}
    \frac{1}{T}\sum_{t=0}^{T-1}
    \expect{\norm{\nabla F(\bar{\bz}^t)}^2}.
\end{align*}
For~\eqref{eq: z telescope last without pseudo},
we have
\begin{align*}
&
35 \eta_g \eta_l^3 s^3 L^2
\frac{1}{T}
\sum_{t=0}^{T-1}
\frac{1}{m} 
\sum_{i=1}^m 
\sum_{p=-1}^{t-1}
\expect{\indc{\tau_i(t) = p}}
(t - p)^2
\expect{
\norm{\nabla F_i(\x_i^{p+1})}^2}\\
&\le
\frac{
70 \eta_g \eta_l^3 s^3 L^2
}{mT \delta^2} 
\sum_{t=0}^{T-1}
\sum_{i=1}^m 
\expect{\norm{\nabla F_i(\x_i^{t})}^2} \\
&\le
\frac{140 \eta_g \eta_l^3 s^3 L^4}{mT \delta^2} 
\sum_{t=0}^{T-1}
\sum_{i=1}^m 
\expect{\norm{\x_i^{t} - \bz_i^t}^2}
+
\frac{140 \eta_g \eta_l^3 s^3 L^2}{mT \delta^2} 
\sum_{t=0}^{T-1}
\sum_{i=1}^m 
\expect{\norm{\nabla F_i(\bz_i^{t})}^2} \\
&\le
\pth{1 + \frac{2\eta_l^2 \eta_g^2 s^2 L^2}{\delta^2} 
}
\pth{\variance}
\frac{70 \eta_g \eta_l^3 s^3 L^2}{mT} 
\sum_{t=0}^{T-1}
\sum_{i=1}^m 
\expect{\norm{\nabla F_i(\bz_i^{t})}^2} \\
&\overset{(a)}{\le}
\pth{\variance}
\frac{71 \eta_g \eta_l^3 s^3 L^2}{mT} 
\sum_{t=0}^{T-1}
\sum_{i=1}^m 
\expect{\norm{\nabla F_i(\bz_i^{t})}^2} \\
&\overset{(b)}{\le}
\frac{426 \eta_g \eta_l^3 s^3 L^2}{\delta^2}
\qth{1 + \frac{40 \rho s^2 \eta_l^2 \eta_g^2 L^2}{(1 - \sqrt{\rho})^2}}
\zeta^2
+
\frac{426 \eta_g \eta_l^3 s^3 L^2}{\delta^2}
\pth{\beta^2 + 1}
\qth{1 + \frac{40 \rho s^2 \eta_l^2 \eta_g^2 L^2}{(1 - \sqrt{\rho})^2}}
\frac{1}{T}
\sum_{t=0}^{T-1}
\expect{\norm{\nabla F (\bar{\bz}^t)}^2} \\
&~~~+
\frac{\eta_g \eta_l^3 s^3 L^2 \sigma^2}{2 \delta^2}
,
\end{align*}
where inequality $(a)$ holds because of~\eqref{eq: lr condition main text},
inequality $(b)$ holds because of~\eqref{eq: hetero consensus}.

Putting~\eqref{eq: z telescope third to last without pseudo},~\eqref{eq: z telescope second to last without pseudo} and~\eqref{eq: z telescope last without pseudo} together and plugging them back into the telescoping sum,
it holds that
\begin{align*}
&\frac{\expect{ F^{\star} - F(\bar{\bz}^0) }}{T} \\
&\le
- \pth{
\frac{\eta_l \eta_g s}{4}
-
\frac{
14
\pth{\beta^2 + 1}
\eta_l^3 \eta_g^3 s^3 L^2
\pth{1 + L^2}}{\delta^2}
-
\frac{20 \rho s^3 \eta_l^3 \eta_g^3 L^2 \pth{\beta^2 + 1}}{(1 - \sqrt{\rho})^2}}
\frac{1}{T} \sum_{t=0}^{T-1}
\expect{\norm{\nabla F(\bar{\bz}^t)}^2} \\
&~~~- \pth{
-
\frac{426 \eta_g \eta_l^3 s^3 L^2
\pth{\beta^2 + 1}
\pth{1 + L^2}}{\delta^2}
}
\frac{1}{T} \sum_{t=0}^{T-1}
\expect{\norm{\nabla F(\bar{\bz}^t)}^2}
\\
&~~~+\frac{4 \eta_l^2 \eta_g^2 s L {\delta_{\max}} \sigma^2}{m \delta^2} 
+
\pth{
\frac{\eta_l^3 \eta_g^3 s^2 L^2
\sigma^2}{2\delta^2}
+
\frac{1.5 \rho s^2 \eta_l^3 \eta_g^3 L^2}{(1 - \sqrt{\rho})^2 \delta^2}
\sigma^2 
+
\frac{\eta_g \eta_l^3 s^3 L^2
\sigma^2}{2 \delta^2}}
\\
&~~~+ 
\frac{
15
\eta_l^3 \eta_g^3 s^3 L^2
\zeta^2}{\delta^2}
+
\frac{20 \rho s^3 \eta_l^3 \eta_g^3 L^2 }{(1 - \sqrt{\rho})^2}
\zeta^2 
+
\frac{
430 \eta_g \eta_l^3 s^3 L^2
\zeta^2}{\delta^2} \\
&\le
- \frac{\eta_l \eta_g s }{6}
\frac{1}{T} \sum_{t=0}^{T-1}
\expect{\norm{\nabla F(\bar{\bz}^t)}^2} \\
&~~~+\frac{4 \eta_l^2 \eta_g^2 s L {\delta_{\max}} \sigma^2}{m \delta^2} 
+
\pth{
\frac{\eta_l^3 \eta_g^3 s^2 L^2
\sigma^2}{2\delta^2}
+
\frac{1.5 \rho s^2 \eta_l^3 \eta_g^3 L^2}{(1 - \sqrt{\rho})^2 \delta^2}
\sigma^2 
+
\frac{\eta_g \eta_l^3 s^3 L^2
\sigma^2}{2 \delta^2}}
\\
&~~~+ 
\frac{
15
\eta_l^3 \eta_g^3 s^3 L^2
\zeta^2}{\delta^2}
+
\frac{20 \rho s^3 \eta_l^3 \eta_g^3 L^2 }{(1 - \sqrt{\rho})^2}
\zeta^2 
+
\frac{
430 \eta_g \eta_l^3 s^3 L^2
\zeta^2}{\delta^2}
,
\end{align*}
where the last inequality holds because of \eqref{eq: lr condition main text}.

Combining the above and rearranging the terms,
we get
\begin{align*}    
    \frac{1}{T} \sum_{t=0}^{T-1}
    \expect{\norm{\nabla F(\bar{\bz}^t)}^2} 
    &\le
    \frac{6\pth{F(\bar{\bz}^0) - F^\star}}{\eta_l \eta_g s T} \\
    &~~~+\frac{24 \eta_l \eta_g L {\delta_{\max}} \sigma^2}{m \delta^2} 
    +
    \pth{
    \frac{3 \eta_l^2 \eta_g^2 s L^2
    \sigma^2}{\delta^2}
    +
    \frac{9 \rho s \eta_l^2 \eta_g^2 L^2}{(1 - \sqrt{\rho})^2 \delta^2}
    \sigma^2 
    +
    \frac{3 \eta_l^2 s^2 L^2
    \sigma^2}{\delta^2}}
    \\
    &~~~+ 
    \frac{
    90
    \eta_l^2 \eta_g^2 s^2 L^2
    \zeta^2}{\delta^2}
    +
    \frac{120 \rho s^2 \eta_l^2 \eta_g^2 L^2 }{(1 - \sqrt{\rho})^2}
    \zeta^2 
    +
    \frac{
    2580 \eta_l^2 s^2 L^2
    \zeta^2}{\delta^2} \\
    &\le
    \frac{6\pth{F(\bar{\bz}^0) - F^\star}}{\eta_l \eta_g s T} 
    +\frac{24 \eta_l \eta_g L {\delta_{\max}} \sigma^2}{m \delta^2} 
    +
    \frac{15 \eta_l^2 \eta_g^2 s^2 L^2
    \sigma^2}{(1 - \sqrt{\rho})^2 \delta^2}
    +
    \frac{
    2800 \eta_l^2 \eta_g^2 s^2 L^2
    \zeta^2}{\delta^2 (1 - \sqrt{\rho})^2}
    ,
\end{align*}  
where the last inequality holds because $\rho < 1$.
In terms of asymptotics,
we have
\begin{align*}
    \frac{1}{T}\sum_{t=0}^{T-1}\expect{\norm{\nabla F(\bar{\bz}^t)}^2}
    &\lesssim
    \frac{\pth{F(\bar{\bz}^0) - F^\star}}{\eta_l \eta_g s T}
    +\frac{\eta_l \eta_g L \sigma^2}{ m} \frac{\delta_{\max}}{\delta^2}
    +
    \eta_l^2 \eta_g^2 s^2 L^2
    \pth{\frac{\sigma^2 +  \zeta^2}{\delta^2 \pth{1 - \sqrt{\rho}}^2}}
    ,
\end{align*}
where we use the convention that $\eta_g \ge 1$ for ease of presentation.
\end{proof}

\newpage
\section{Convergence Rate of $\bar{\x}^t$ (\prettyref{cor: x bar rate})}

\subsection{Convergence error of~\prettyref{alg: fedpbc+}}
\begin{corollary}[Convergence error of $\x_i^t$]
\label{cor: x bar without pseudo}
Suppose learning rates conditions in~\eqref{eq: lr condition main text} are met for $\eta_l$ and $\eta_g$,
and Assumptions \ref{ass: prob lower bound}, \ref{ass: 2 smmothness}, \ref{ass: bounded variance client-wise} and \ref{ass: bounded similarity} hold
for $T \ge 1$,
it holds that
\begin{align*}
    \frac{1}{T}\sum_{t=0}^{T-1}\expect{\norm{\nabla F(\bar{\x}^t)}^2}
    &\lesssim
    \frac{\pth{F(\bar{\x}^0) - F^\star}}{\eta_l \eta_g s T}
    +\frac{\eta_l \eta_g L \sigma^2}{ m} \frac{\delta_{\max}}{\delta^2} 
    +
    \eta_l^2 \eta_g^2 s^2 L^2
    \pth{\frac{\sigma^2 + \zeta^2}{\delta^2 \pth{1 - \sqrt{\rho}}^2}}
    ,
\end{align*}
\end{corollary}
\begin{proof}[\bf Proof of \prettyref{cor: x bar without pseudo}]
\begin{align*}
&\frac{1}{T}\sum_{t=0}^{T-1}\expect{\norm{\nabla F(\bar{\x}^t)}^2}\le
\frac{3}{T}\sum_{t=0}^{T-1}\expect{\norm{\nabla F(\bar{\x}^t) - \nabla F(\bar{\bz}^t)}^2}
+
\frac{3}{2T}\sum_{t=0}^{T-1} \expect{\norm{\nabla F(\bar{\bz}^t)}^2} \\
&\overset{(a)}{\le}
\frac{3 L^2}{T}\sum_{t=0}^{T-1} \expect{\norm{\bar{\x}^t - \bar{\bz}^t}^2}
+
\frac{3}{2T}\sum_{t=0}^{T-1} \expect{\norm{\nabla F(\bar{\bz}^t)}^2} \\
&\overset{(b)}{\le}
\frac{3 L^2}{T}\sum_{t=0}^{T-1} \frac{1}{m} \sum_{i=1}^m
\expect{\norm{\x_i^t - \bz_i^t}^2}
+
\frac{3}{2T}\sum_{t=0}^{T-1} \expect{\norm{\nabla F(\bar{\bz}^t)}^2} \\
&\le
3
\pth{\variance}
\frac{\eta_l^2 \eta_g^2 s^2 L^2}{T}\sum_{t=0}^{T-1} \frac{1}{m} \sum_{i=1}^m
\expect{\norm{\nabla F_i(\bz_i^t)}^2}
+
\frac{3}{2T}\sum_{t=0}^{T-1} \expect{\norm{\nabla F(\bar{\bz}^t)}^2} 
,
\end{align*}
where inequality $(a)$ follows from~\prettyref{app: preliminaries} 2, 
inequality $(b)$ follows from Assumption \ref{ass: 2 smmothness}.

Further plug in~\prettyref{prop: average gradient to global gradient},
\begin{align*}
    \frac{1}{T}\sum_{t=0}^{T-1}\expect{\norm{\nabla F(\bar{\x}^t)}^2}
    &
    \le
    \frac{3}{2T}\sum_{t=0}^{T-1} \expect{\norm{\nabla F(\bar{\bz}^t)}^2} 
    +
    9 \eta_l^2 \eta_g^2 s^2 L^2
    \pth{\variance}
    \pth{\beta^2 + 1}
    \frac{1}{T}
    \sum_{t=0}^{T-1} 
    \expect{\norm{\nabla F(\bar{\bz}^t)}^2} \\
    &~~~+
    9 \eta_l^2 \eta_g^2 s^2 L^4
    \pth{\variance}
    \frac{1}{T}
    \sum_{t=0}^{T-1} 
    \frac{1}{m} \sum_{i=1}^m
    \expect{\norm{\bz_i^t - \bar{\bz}^t}^2}
    +
    9 \eta_l^2 \eta_g^2 s^2 L^2
    \pth{\variance}
    \zeta^2 
    . 
\end{align*}

Finally, plug in~\prettyref{lmm: consensus z without pseudo}.
\begin{align*}
    \frac{1}{T}\sum_{t=0}^{T-1}\expect{\norm{\nabla F(\bar{\x}^t)}^2}
    &\le
    \pth{
    \frac{3}{2} + 
    9 \eta_l^2 \eta_g^2 s^2 L^2
    \pth{\variance}
    \pth{\beta^2 + 1} 
    \frac{90}{80^2}
    }
    \frac{1}{T}\sum_{t=0}^{T-1} \expect{\norm{\nabla F(\bar{\bz}^t)}^2} 
     \\
    &~~~+
    \frac{9 \times 8}{80^2} \eta_l^2 \eta_g^2 s L^2
    \pth{\variance}
    \sigma^2 
    +
    9 \eta_l^2 \eta_g^2 s^2 L^2
    \pth{\variance}
    \zeta^2
    +
    \frac{9 \times 90}{200^2} \eta_l^2 \eta_g^2 s^2 L^2
    \zeta^2 \\
    &\le
    \frac{2}{T}\sum_{t=0}^{T-1} \expect{\norm{\nabla F(\bar{\bz}^t)}^2} 
    +
    \frac{s L^2 \eta_l^2 \eta_g^2}{\delta^2}
    \sigma^2 
    +
    \frac{9 \eta_l^2 \eta_g^2 s^2 L^2}{\delta^2}
    \zeta^2
    +
    s^2 L^2 \eta_l^2 \eta_g^2 
    \zeta^2  \\
    & \le
    \frac{12\pth{F(\bar{\bz}^0) - F^\star}}{\eta_l \eta_g s T}
    +\frac{48 \eta_l \eta_g L {\delta_{\max}} \sigma^2}{m \delta^2} 
    +
    \frac{31 \eta_l^2 \eta_g^2 s^2 L^2 \sigma^2}{(1 - \sqrt{\rho})^2 \delta^2}
    +
    \frac{
    5600 \eta_l^2 \eta_g^2 s^2 L^2
    \zeta^2}{(1 - \sqrt{\rho})^2 \delta^2},
\end{align*}
where the last inequality holds because $\rho < 1$.
In terms of asymptotics,
we have
\begin{align*}
    \frac{1}{T}\sum_{t=0}^{T-1}\expect{\norm{\nabla F(\bar{\x}^t)}^2}
    &\lesssim
    \frac{\pth{F(\bar{\x}^0) - F^\star}}{\eta_l \eta_g s T}
    +\frac{\eta_l \eta_g L \sigma^2}{ m} \frac{\delta_{\max}}{\delta^2} 
    +
    \eta_l^2 \eta_g^2 s^2 L^2
    \pth{\frac{\sigma^2 + \zeta^2}{\delta^2 (1 - \sqrt{\rho})^2}}
    ,
\end{align*}
where we use the convention that $\eta_g \ge 1$ for ease of presentation.
\end{proof}

\subsection{Convergence rate of~\prettyref{alg: fedpbc+}}
\begin{proof}[\bf Proof of~\prettyref{cor: x bar rate}]
Choose step-size as
$\eta_l = \frac{1}{\sqrt{T} s L}$,
$\eta_g = \sqrt{s \delta m}$
such that learning rate conditions in~\eqref{eq: lr condition main text} are met,
it holds that
\begin{align*}
    \frac{1}{T}\sum_{t=0}^{T-1}\expect{\norm{\nabla F(\bar{\x}^t)}^2}
    &\lesssim
    \frac{L \pth{F(\bar{\x}^0) - F^\star}}{\sqrt{s \delta mT}}
    +
    \frac{\delta_{\max}}{\delta^{\frac{3}{2}} \sqrt{s m T}} \sigma^2
    +
    \frac{s m}{T}
    \pth{\frac{\sigma^2 +\zeta^2}{\delta (1 - \sqrt{\rho})^2}}
    .
\end{align*}

\end{proof}

\newpage
\section{Additional Results and Interpretations}

\subsection{Consensus error of~\prettyref{alg: fedpbc+}}
\begin{corollary}[Consensus error of $\x_i^t$]
\label{cor: x consensus without pseudo}
Suppose learning rates conditions are met in~\eqref{eq: lr condition main text} for $\eta_l$ and $\eta_g$,
and Assumptions \ref{ass: prob lower bound}, \ref{ass: 2 smmothness}, \ref{ass: bounded variance client-wise} and \ref{ass: bounded similarity} hold
for $T \ge 1$,
it holds that
\begin{align*}
 \frac{1}{T} \sum_{t=0}^{T-1}
    \frac{1}{m} \sum_{i=1}^m 
    \expect{\norm{\x_i^t - \bar{\x}^t}^2}
    &\lesssim
    \frac{\pth{F(\bar{\x}^0) - F^\star}}{\eta_l \eta_g s T}
    +\frac{\eta_l \eta_g L \sigma^2}{ m} \frac{\delta_{\max}}{\delta^2} \\
    &~~~\qquad +
    \eta_l^2 \eta_g^2 s^2 L^2
    \pth{\frac{\sigma^2 + \zeta^2}{\delta^2}}
    \qth{
    1
    +
    \frac{\rho}{\pth{1 - \sqrt{\rho}}^2}
    }
,
\end{align*}
\end{corollary}
\begin{proof}[\bf Proof of \prettyref{cor: x consensus without pseudo}]
\begin{align*}
&\frac{1}{T} \sum_{t=0}^{T-1}
\frac{1}{m} \sum_{i=1}^m 
\norm{\x_i^t - \bar{\x}^t}^2
=
\frac{1}{T} \sum_{t=0}^{T-1}
\frac{1}{m} \sum_{i=1}^m 
\norm{\x_i^t - \bz_i^t + \bz_i^t - \bar{\bz}^t + \bar{\bz}^t - \bar{\x}^t}^2 \\
&\overset{(a)}{\le}
\frac{1}{T} \sum_{t=0}^{T-1}
\frac{3}{m} \sum_{i=1}^m 
\norm{\x_i^t - \bz_i^t}^2
+
\frac{1}{T} \sum_{t=0}^{T-1}
\frac{3}{m} \sum_{i=1}^m 
\norm{\bz_i^t - \bar{\bz}^t}^2
+
\frac{1}{T} \sum_{t=0}^{T-1}
3
\norm{\bar{\bz}^t - \bar{\x}^t}^2 \\
&\overset{(b)}{\le}
\frac{1}{T} \sum_{t=0}^{T-1}
\frac{3}{m} \sum_{i=1}^m 
\norm{\x_i^t - \bz_i^t}^2
+
\frac{1}{T} \sum_{t=0}^{T-1}
\frac{3}{m} \sum_{i=1}^m 
\norm{\bz_i^t - \bar{\bz}^t}^2
+
\frac{1}{T} \sum_{t=0}^{T-1}
\frac{3}{m} \sum_{i=1}^m 
\norm{\bz_i^t - \x_i^t}^2 \\
&=
\frac{1}{T} \sum_{t=0}^{T-1}
\frac{6}{m} \sum_{i=1}^m 
\norm{\x_i^t - \bz_i^t}^2
+
\frac{1}{T} \sum_{t=0}^{T-1}
\frac{3}{m} \sum_{i=1}^m 
\norm{\bz_i^t - \bar{\bz}^t}^2
,
\end{align*}
where inequalities $(a)$ and $(b)$ follow from Jensen's inequality.
Plug in Proposition \ref{prop: client dis} and take expectation over all the randomness, we get
\begin{align*}
&\frac{1}{T} \sum_{t=0}^{T-1}
\frac{1}{m} \sum_{i=1}^m 
\expect{\norm{\x_i^t - \bar{\x}^t}^2}
\le
\frac{36 \eta_l^2 \eta_g^2 s^2 }{\delta^2}
\pth{\beta^2 + 1}
\frac{1}{T}\sum_{t=0}^{T-1} 
\expect{\norm{\nabla F(\bar{\bz}^{t})}^2} \\
&~~~+
\frac{36 \eta_l^2 \eta_g^2 s^2 }{\delta^2}
\zeta^2
+
\pth{3 + \frac{36 \eta_l^2 \eta_g^2 s^2 
L^2}{\delta^2}} %
\frac{1}{m}
\sum_{i=1}^m \frac{1}{T}\sum_{t=0}^{T-1} 
\expect{\norm{\bz_i^{t} - \bar{\bz}^t}^2} \\
&\le
\frac{36 \eta_l^2 \eta_g^2 s^2 }{\delta^2}
\pth{\beta^2 + 1}
\frac{1}{T}\sum_{t=0}^{T-1} 
\expect{\norm{\nabla F(\bar{\bz}^{t})}^2}
+
\frac{36 \eta_l^2 \eta_g^2 s^2 }{\delta^2}
\zeta^2 \\
&~~~+
\frac{4}{m}
\sum_{i=1}^m \frac{1}{T}\sum_{t=0}^{T-1} 
\expect{\norm{\bz_i^{t} - \bar{\bz}^t}^2},
\end{align*}
where the last inequality holds because of learning rate condition in~\eqref{eq: lr condition main text}.
Next, plug in~\prettyref{lmm: consensus z without pseudo}:
\begin{align*}
&\frac{1}{T} \sum_{t=0}^{T-1}
\frac{1}{m} \sum_{i=1}^m 
\expect{\norm{\x_i^t - \bar{\x}^t}^2}
\le
\frac{36 \eta_l^2 \eta_g^2 s^2 }{\delta^2}
\pth{\beta^2 + 1}
\frac{1}{T}\sum_{t=0}^{T-1} 
\expect{\norm{\nabla F(\bar{\bz}^{t})}^2} \\
&~~~+
\frac{36 \eta_l^2 \eta_g^2 s^2 }{\delta^2}
\zeta^2
+
\frac{4}{m}
\sum_{i=1}^m \frac{1}{T}\sum_{t=0}^{T-1} 
\expect{\norm{\bz_i^{t} - \bar{\bz}^t}^2} \\
&\le
\frac{36 \eta_l^2 \eta_g^2 s^2 }{\delta^2}
\pth{\beta^2 + 1}
\frac{1}{T}\sum_{t=0}^{T-1} 
\expect{\norm{\nabla F(\bar{\bz}^{t})}^2}
+
\frac{1}{4 T}\sum_{t=0}^{T-1}\expect{\norm{\nabla F(\bar{\bz}^t)}^2} \\
&~~~
+
\frac{36 \eta_l^2 \eta_g^2 s^2 }{\delta^2}
\zeta^2
+
\frac{12 \rho s \eta_l^2 \eta_g^2}{(1 - \sqrt{\rho})^2 \delta^2} 
\sigma^2 
+
\frac{160 \rho s^2 \eta_l^2 \eta_g^2}{(1 - \sqrt{\rho})^2}
\zeta^2 \\
&\le
\frac{1}{2 T}\sum_{t=0}^{T-1} 
\expect{\norm{\nabla F(\bar{\bz}^{t})}^2}
+
\frac{12 \rho s \eta_l^2 \eta_g^2}{(1 - \sqrt{\rho})^2 \delta^2} 
\sigma^2 
+
\frac{36 \eta_l^2 \eta_g^2 s^2 }{\delta^2}
\zeta^2
+
\frac{160 \rho s^2 \eta_l^2 \eta_g^2}{(1 - \sqrt{\rho})^2}
\zeta^2
.
\end{align*}
Finally,
we plug in~\prettyref{thm: z bar rate}
\begin{align*}
    &\frac{1}{T} \sum_{t=0}^{T-1}
    \frac{1}{m} \sum_{i=1}^m 
    \expect{\norm{\x_i^t - \bar{\x}^t}^2}
    \le
    \frac{3\pth{F(\bar{\x}^0) - F^\star}}{\eta_l \eta_g s T} 
    +\frac{12 \eta_l \eta_g L {\delta_{\max}} \sigma^2}{m \delta^2} 
    +
    \frac{28 s^2 \eta_l^2 \eta_g^2 L^2}{\delta^2 (1 - \sqrt{\rho})^2} 
    \sigma^2
    + 
    \frac{
    1600
    \eta_l^2 \eta_g^2 s^2 L^2}{\delta^2 (1 - \sqrt{\rho})^2}
    \zeta^2
    ,
\end{align*}
where we use the fact that $\bar{\bz}^0 = \bar{\x}^0$ and $\rho < 1$, 
and the convention that $\eta_g \ge 1$ and $L \ge 1$ for ease of presentation.

In terms of asymptotics,
we have
\begin{align*}
    \frac{1}{T} \sum_{t=0}^{T-1}
    \frac{1}{m} \sum_{i=1}^m 
    \expect{\norm{\x_i^t - \bar{\x}^t}^2}
    &\lesssim
    \frac{\pth{F(\bar{\x}^0) - F^\star}}{\eta_l \eta_g s T}
    +\frac{\eta_l \eta_g L \sigma^2}{ m} \frac{\delta_{\max}}{\delta^2}
    +
    \eta_l^2 \eta_g^2 s^2 L^2
    \pth{\frac{\sigma^2 + \zeta^2}{\delta^2 (1 - \sqrt{\rho})^2}}
    .
\end{align*}
\end{proof}

\subsection{Orders of the asymptotic rates}
From~\prettyref{thm: z bar rate}, 
\prettyref{cor: x bar without pseudo},
\prettyref{cor: x consensus without pseudo},
it is easy to see from the theorem statements that they are all of the same asymptotic order, \ie,
\[
    \frac{1}{T}\sum_{t=0}^{T-1}\mathbb{E}[\| \nabla F (\bar{\x}^t)\|_2^2]
    \asymp
    \frac{1}{T}\sum_{t=0}^{T-1}
    \frac{1}{m} \sum_{i=1}^m
    \mathbb{E}[\| \x_i^t - \bar{\x}^t\|_2^2]
    \asymp
    \frac{1}{T}\sum_{t=0}^{T-1}\mathbb{E}[\| \nabla F (\bar{\bz}^t)\|_2^2].
\]
In addition,
by applying learning rate conditions in~\eqref{eq: lr condition main text} to~\prettyref{lmm: consensus z without pseudo} and~\prettyref{prop: client dis},
we can also see that
\[
    \frac{1}{T}\sum_{t=0}^{T-1}
    \frac{1}{m} \sum_{i=1}^m
    \mathbb{E}[\| \x_i^t - \bz_i^t\|_2^2]
    \asymp
    \frac{1}{T}\sum_{t=0}^{T-1}
    \frac{1}{m} \sum_{i=1}^m
    \mathbb{E}[\| \bz_i^t - \bar{\bz}^t\|_2^2]
    \asymp
    \frac{1}{T}\sum_{t=0}^{T-1}\mathbb{E}[\| \nabla F (\bar{\bz}^t)\|_2^2].
\]
Therefore,
we conclude that~\eqref{eq: approx and consensus for z},~\eqref{eq: x bar consensus} and~\eqref{eq: x bar relation} hold.

\newpage

\section{Numerical Experiments}
\label{app: numerical}
\subsection{Code}
The code for reproducing our experiments is available at \url{https://github.com/mingxiang12/FedAWE}.

\subsection{Experimental setups}
\label{app: numerical setup}
\begin{wrapfigure}[18]{r}{0.35\textwidth} 
\centering
\resizebox{\linewidth}{!}{
\includegraphics[width=\linewidth, trim=0.1cm 0.1cm 0 0.1cm, clip]{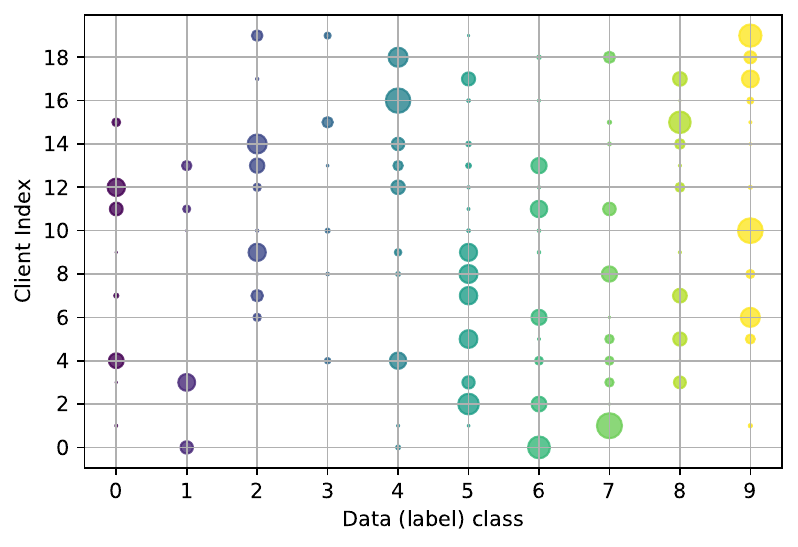}}
\caption{
\footnotesize
An example of data heterogeneity using $\mathsf{Dirichlet}(\alpha=0.1)$ distribution with $20$ clients.
$x$-axis denotes the categories of images, 
while $y$-axis denotes the client index.
The size of a circle refers to the proportion of pictures in a given class.
The color of a circle distinguishes images with different categories.
}
\centering
\label{fig: noniid 20 cifar10}
\end{wrapfigure}
\textbf{Hardware and Software Setups.}
\begin{itemize}[leftmargin=*]
\item{\bf Hardware.}
The simulations are performed on a private cluster with 64 CPUs, 500 GB RAM and 8 NVIDIA A5000 GPU cards.

\item {\bf Software.}
We code the experiments based on PyTorch 1.13.1 \cite{paszke2019pytorch} and Python 3.7.16.
\end{itemize}

\textbf{Neural Network and Hyper-parameter Specifications.}
\label{app: hyperparameter}
Table~\ref{tbl: cnn structures} specifies details of the structures of the convolutional neural network and training.
We initialize CNNs using the Kaiming initialization.
The initial local learning rate $\eta_0$ and the global learning rate $\eta_g$ are searched, 
based on the best performance after $500$ global rounds, 
over two grids $\sth{0.1, 0.05, 0.01, 0.005, 0.001, 0.0005}$ and $\sth{0.5,1,1.5,5,10,50}$, respectively.
The results are presented in Table~\ref{tbl: learning rate}.

The difference between~\FedAvg~over active clients and~\FedAvg~over all clients
is that the latter counts the contributions of unavailable clients as $\bm{0}$'s.
We set $\beta = 0.001$ for~\FAST~\cite{ribero2022federated}, 
which is tuned over a grid of $\sth{0.1, 0.05, 0.01, 0.005, 0.001, 0.0005}$.
In addition,
as recommended by \cite{wang2023lightweight}, 
we choose $K=50$ in~\FedAU~without further specification.
We train CNNs on all datasets for $2000$ rounds.
\prettyref{fig: motivating example non-stationary} adopts the same hyperparameter setups, yet with only $1000$ training rounds.
\begin{table}[!t]
\caption{
\footnotesize
Neural network architecture, loss function, learning rate scheduling, training steps and batch size specifications}
\label{tbl: cnn structures}
\resizebox{\linewidth}{!}{
\begin{tabular}{cccc}
\toprule
{\bf Datasets}& 
{\bf SVHN} & 
{\bf CIFAR-10} & 
{\bf CINIC-10} \\
\toprule
Neural network &
CNN &
CNN &
CNN \\
Model architecture$^*$ & 
\begin{tabular}{p{.18\textwidth}}
\centering
{\bf C}(3,32)
-- {\bf R}
-- {\bf M}
-- {\bf C}(32,32)
-- {\bf R}
-- {\bf M}
-- {\bf L}(128)
-- {\bf R}
-- {\bf L}(10)
\end{tabular}
&
\begin{tabular}{p{.18\textwidth}}
\centering
{\bf C}(3,32)
-- {\bf R}
-- {\bf M}
-- {\bf C}(32,32)
-- {\bf R}
-- {\bf M}
-- {\bf L}(256)
-- {\bf R}
-- {\bf L}(64)
-- {\bf R}
-- {\bf L}(10)
\end{tabular}
&
\begin{tabular}{p{.18\textwidth}}
\centering
{\bf C}(3,32)
-- {\bf R}
-- {\bf M}
-- {\bf C}(32,32)
-- {\bf R}
-- {\bf M}
-- {\bf D}
-- {\bf L}(512)
-- {\bf R}
-- {\bf D}
-- {\bf L}(256)
-- {\bf R}
-- {\bf D}
-- {\bf L}(10)
\end{tabular} \\
\midrule
Loss function &
\multicolumn{3}{c}{Cross-entropy loss} \\
\addlinespace[1ex]
\begin{tabular}{p{.25\textwidth}}
\centering
Local learning rate $\eta_l$ \\ scheduling 
\end{tabular}&
\multicolumn{3}{c}{$\eta_l = \frac{\eta_0}{\sqrt{t/10 + 1}}$, where $t$ denotes the global round.} \\
\addlinespace[1ex]
Number of local steps $s$ &
\multicolumn{3}{c}{10} \\
\addlinespace[1ex]
Number of global rounds $T$ &
\multicolumn{3}{c}{2000} 
\\
\midrule
Batch size &
\multicolumn{3}{c}{128} \\
\bottomrule
\end{tabular}}
\vskip.2\baselineskip
\begin{tabular}{p{.95\textwidth}}
$^*$
\begin{footnotesize}    
{\bf C}(\# in-channel, \# out-channel): a 2D convolution layer (kernel size 3, stride 1, padding 1);
{\bf R}: ReLU activation function;
{\bf M}: a 2D max-pool layer (kernel size 2, stride 2);
{\bf L}: (\# outputs): a fully-connected linear layer;
{\bf D}: a dropout layer (probability 0.2).
\end{footnotesize}
\end{tabular}
\end{table}
\begin{table}[!t]
\caption{Initial learning rate $\eta_0$ and global learning rate $\eta_g$}
\centering
\label{tbl: learning rate}
\resizebox{\linewidth}{!}{
\begin{tabular}{ccccccccccccccccc}
\toprule
\multicolumn{1}{c}{\bf Algorithms} &
\multicolumn{2}{c}{\begin{tabular}{c} \FedAvg \\ {\em active} \end{tabular}}&
\multicolumn{2}{c}{\begin{tabular}{c} \FedAvg \\ {\em known} \end{tabular}}&
\multicolumn{2}{c}{\begin{tabular}{c} \FedAvg \\ {\em all} \end{tabular}}&
\multicolumn{2}{c}{\FedAU}&
\multicolumn{2}{c}{\FAST}&
\multicolumn{2}{c}{\FedAPM}&
\multicolumn{2}{c}{\MIFA} &
\multicolumn{2}{c}{\FedVARP} 
\\
\midrule
\multirow{2}{*}{SVHN}&
$\eta_0$&
$\eta_g$&
$\eta_0$&
$\eta_g$&
$\eta_0$&
$\eta_g$&
$\eta_0$&
$\eta_g$&
$\eta_0$&
$\eta_g$&
$\eta_0$&
$\eta_g$&
$\eta_0$&
$\eta_g$&
$\eta_0$&
$\eta_g$ \\[2pt]
&
0.05 &
1.0 &
0.1 &
1.0 &
0.05 &
1.0 &
0.05 &
1.0 &
0.05 &
1.0 &
0.1 &
1.0 &
0.05 &
1.0 &
0.05 &
1.0 \\
\midrule
\multirow{2}{*}{CIFAR-10}&
$\eta_0$&
$\eta_g$&
$\eta_0$&
$\eta_g$&
$\eta_0$&
$\eta_g$&
$\eta_0$&
$\eta_g$&
$\eta_0$&
$\eta_g$&
$\eta_0$&
$\eta_g$&
$\eta_0$&
$\eta_g$&
$\eta_0$&
$\eta_g$ \\[2pt]
&
0.05 &
1.0 &
0.1 &
1.0 &
0.05 &
1.0 &
0.05 &
1.0 &
0.05 &
1.0 &
0.1 &
1.0 &
0.05 &
1.0 &
0.05 &
1.0 \\
\midrule
\multirow{2}{*}{CINIC-10}&
$\eta_0$&
$\eta_g$&
$\eta_0$&
$\eta_g$&
$\eta_0$&
$\eta_g$&
$\eta_0$&
$\eta_g$&
$\eta_0$&
$\eta_g$&
$\eta_0$&
$\eta_g$&
$\eta_0$&
$\eta_g$&
$\eta_0$&
$\eta_g$ \\[2pt]
&
0.05 &
1.0 &
0.1 &
1.0 &
0.05 &
1.0 &
0.05 &
1.0 &
0.05 &
1.0 &
0.1 &
1.0 &
0.05 &
1.0 &
0.05 &
1.0 \\
\bottomrule
\end{tabular}}
\end{table}

\textbf{Datasets and Data Heterogeneity.}

\textit{Datasets.}
All the datasets we evaluate contain 10 classes of images.
Some data enhancement tricks that are standard in training image classifiers are applied during training.
Specifically,
we apply random cropping and gradient clipping with a max norm of 0.5 to all dataset trainings.
Furthermore, random horizontal flipping is applied to CIFAR-10 and CINIC-10.

One full set of experiments 
takes about 6 hours on SVHN and CIFAR-10 datasets, 
while about 10 hours on CINIC-10 dataset.
\begin{itemize}[leftmargin=*]
\item 
{\bf SVHN \cite{netzer2011readingdigits}.}
The dataset contains 32$\times$32 colored images of 10 different number digits.
In total,
there are 73257 train images and 26032 test images.
\item
{\bf CIFAR-10 \cite{krizhevsky2009learning}.}
The dataset contains 32$\times$32 colored images of 10 different objects.
In total,
there are 50000 train images and 10000 test images.
\item
{\bf CINIC-10\cite{darlow2018cinic}.}
The dataset contains 32$\times$32 colored images of 10 different objects.
In total,
there are 90000 train images and 90000 test images.
\end{itemize}

\textit{Data heterogeneity.}
\prettyref{fig: noniid 20 cifar10} visualizes an example of 20 clients, 
the size of each circle corresponds to the relative proportion of images from a specific class. 
The larger the circle, the greater the share of images associated with that particular class.
Moreover,
$\alpha$ controls the heterogeneity of the data such that a greater $\alpha$ entails a more non-\iid local data distribution and vice versa.

\subsection{Non-stationary client unavailability dynamics}
\label{app: unavail dynamics}
\begin{wrapfigure}[11]{r}{0.4\textwidth} 
\vspace*{-\baselineskip}
\centering
\includegraphics[width=.9\linewidth, trim = .1cm .2cm 0 .1cm, clip]{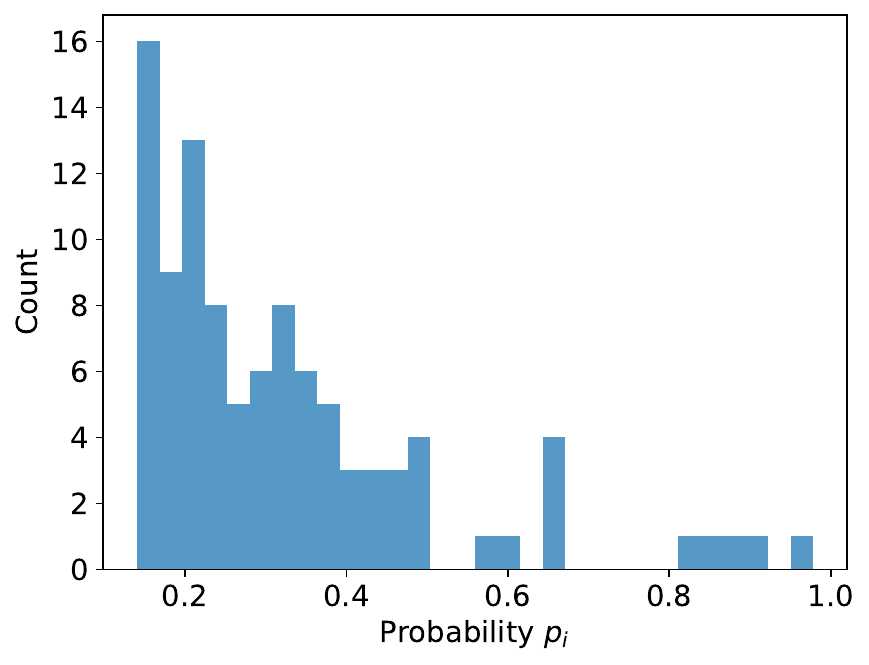}
\caption{
\footnotesize
A histogram of one generated $p_i$'s example
with a total of $m = 100$ clients.
It can be seen that the majority of $p_i$'s are below $0.5$.
}
\label{fig: heterogeneous pit example}%
\end{wrapfigure}
\textbf{Client unavailability dynamics and visualizations.}
As specified in~\prettyref{sec: numerical},
we consider a total of four client unavailable dynamics in the form of $p_i^t = p_i \cdot f_i(t)$,
where $p_i = \iprod{\nu_i}{\phi}$,
$\nu_i \sim \mathsf{Dirichlet} (\alpha)$
and $\phi$ is the distribution to characterize the uneven contributions of each image class.
In detail,
each element $[\phi]_c$ is drawn from a uniform distribution $\mathsf{Uniform} (0, \bm{\Phi}_c)$.
We set $\bm{\Phi}_c = 1$ for the first five image classes and $\bm{\Phi}_{c^\prime} = 0.5$ for the remaining five image classes.
\prettyref{fig: heterogeneous pit example} plots one resulting $p_i$'s example,
wherein $p_i$'s are heterogeneous across clients.

Next,
we formally introduce $f_i(t)$'s under each dynamic.
\begin{itemize}[leftmargin=*]
    \item Stationary: 
    $f_i(t) \triangleq 1$;
    \item Non-stationary with staircase trajectory:
    \[
        f_i(t) \triangleq 
        \indc{t \in [t_0, t_0+P/2)} 
        + 0.4 \cdot \indc{t \in [t_0+P/2, t_0+P)},
    \]
    where $P$ defines a period, $t_0 \in \{0, P, 2P, 3P, \ldots \}$.
    \item Non-stationary with sine trajectory: 
    \[
        f_i(t) \triangleq 
        \gamma \sin (2 \pi/ P \cdot t) + (1 - \gamma),
    \]
    where $\gamma$ signifies the degree of non-stationary.
    \item Non-stationary with interleaved sine trajectory:
    \[
        f_i(t) \triangleq 
        g_i(t) \cdot \indc{p_i \cdot g_i(t) \ge \delta_0},
    \]
    where $g_i(t) \triangleq \gamma \sin (2 \pi/ P \cdot t) + (1 - \gamma)$
    and $\delta_0 = 0.1$ defines a cutting-off lower bound.
    Specifically,
    $\delta_0$ cuts off the sine curve and brings in a period of zero-valued probabilities.
    As different clients have different $p_i$'s,
    the cut-off points are not synchronized among clients,
    leading to additional availability heterogeneity.
\end{itemize}
We choose $\gamma = 0.3$ and $P = 20$ for all non-stationary dynamics.
Next,
we visualize the probability trajectories along with sampled client availability in~\prettyref{fig: dynamics visualization}.
The plots confirm the intuition that
interleaved dynamics is the most difficult one,
\eg, 
no clients are available in the case of $0.1$ therein.

\begin{figure}
    \centering
    \begin{subfigure}[b]{0.48\textwidth}
    \includegraphics[width=\linewidth]{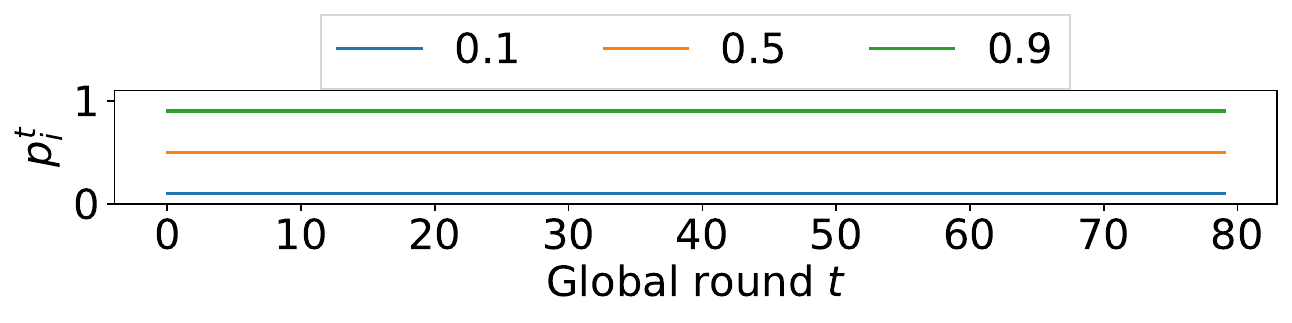}
    \includegraphics[width=\linewidth]{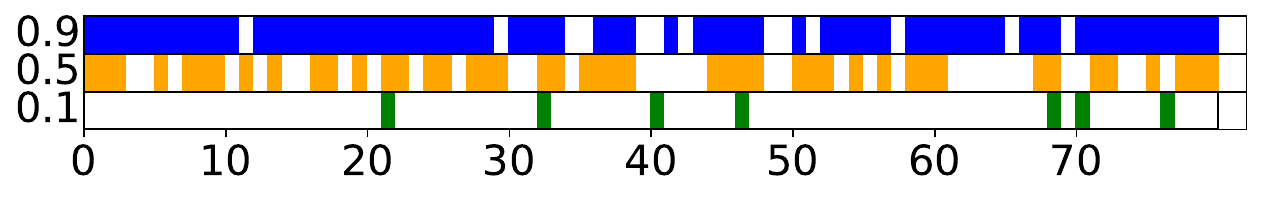}
    \caption{Stationary}
    \label{fig: constant dynamics}
    \end{subfigure}
    \begin{subfigure}[b]{0.48\textwidth}
    \includegraphics[width=\linewidth]{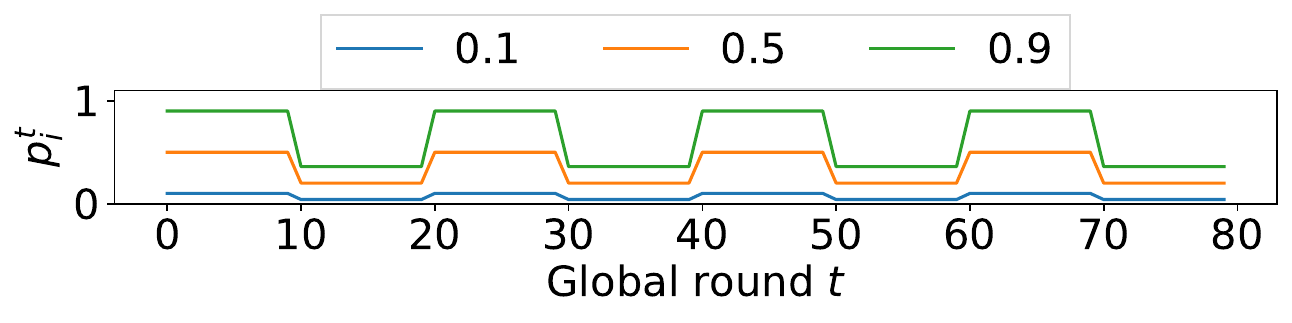}
    \includegraphics[width=\linewidth]{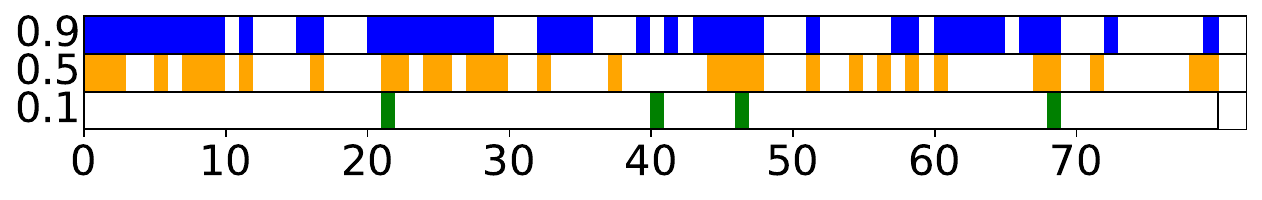}
    \caption{{\bf Non}-stationary with staircase trajectory}
    \label{fig: staircase dynamics}
    \end{subfigure}
    \begin{subfigure}[b]{0.48\textwidth}
    \includegraphics[width=\linewidth]{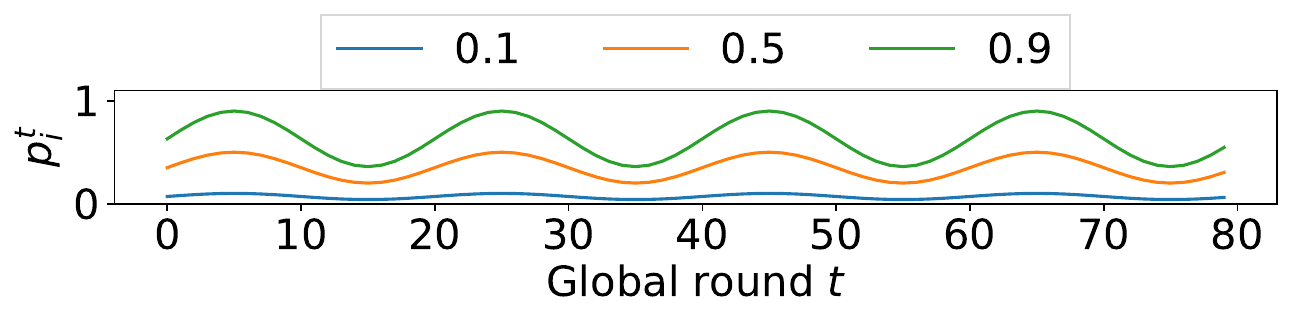}
    \includegraphics[width=\linewidth]{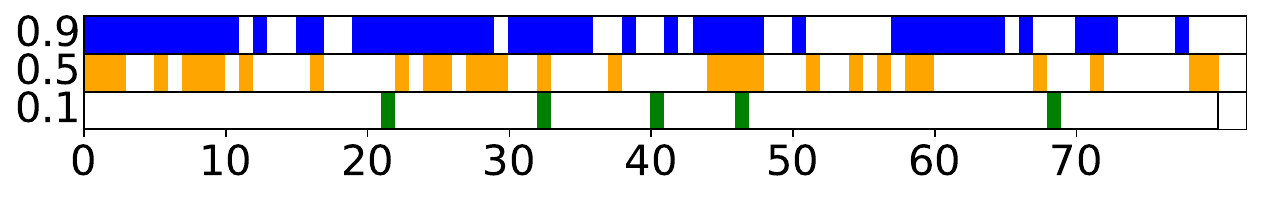}
    \caption{{\bf Non}-stationary with sine trajectory}
    \label{fig: sine dynamics}
    \end{subfigure}
    \begin{subfigure}[b]{0.48\textwidth}
    \includegraphics[width=\linewidth]{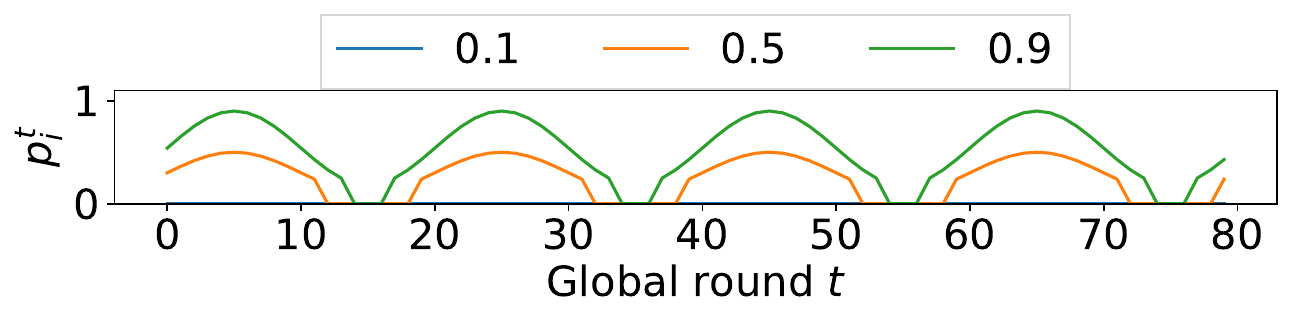}
    \includegraphics[width=\linewidth]{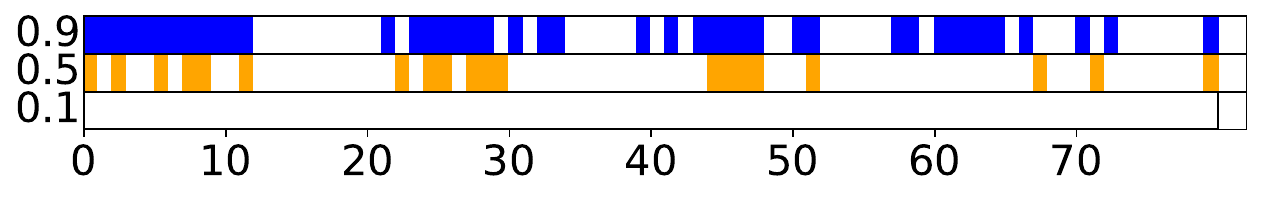}
    \caption{{\bf Non}-stationary with interleaved sine trajectory}
    \label{fig: int sine dynamics}
    \end{subfigure}
    \caption{
    Examples of client unavailability with probabilistic trajectories.
    The first row in each sub-figure plots the probabilistic trajectory of each dynamics.
    The second row visualizes the simulated client availability 
    by using a colored box to denote a client is available in that round.
    The y-axis is the base probability $p_i$ to construct $p_i^t$.
    In other words,
    more blank space means that a client is more scarcely available.
    We simulate the cases where $p_i \in \sth{0.1, 0.5, 0.9}$.
    The detailed construction of $p_i^t$ can be found in~\prettyref{app: unavail dynamics}}
    \label{fig: dynamics visualization}
\end{figure}

\subsection{Additional results}
\begin{table}[!t]
\centering
\caption{\footnotesize 
The first round to reach a targeted test accuracy under 
non-stationary of sine trajectory
over 3 random seeds.
We study the first round to reach $1/4$, $1/2$, $3/4$ and $1$ of the best test accuracy of each dataset in Table~\ref{tab: exp main text},
which is rounded up to the nearest $10\%$ below for ease of presentation.
In addition,
we sample the mean of test accuracy every 20 global rounds to mitigate %
noisy progress. 
Some algorithms may never attain the targeted accuracy
due to their inferior performance,
where we use ``--'' as a placeholder.
}
\label{tab: slowdown supp}
\resizebox{\linewidth}{!}{
\begin{tabular}{c|cccc|cccc|cccc}
\toprule
{\bf Datasets}&
\multicolumn{4}{|c|}{\bf
SVHN
}&
\multicolumn{4}{|c|}{\bf
CIFAR10
}
&
\multicolumn{4}{|c}{\bf
CINIC10
}
\\
\midrule
{\bf Quarters}&
{\bf 1/4}&
{\bf 1/2}&
{\bf 3/4}&
{\bf 1}&
{\bf 1/4}&
{\bf 1/2}&
{\bf 3/4}&
{\bf 1}&
{\bf 1/4}&
{\bf 1/2}&
{\bf 3/4}&
{\bf 1}
\\
\midrule
{\bf Test accuracy}&
{\bf $20\%$}&
{\bf $40\%$}&
{\bf $60\%$}&
{\bf $80\%$}&
{\bf $15\%$}&
{\bf $30\%$}&
{\bf $45\%$}&
{\bf $60\%$}&
{\bf $10\%$}&
{\bf $20\%$}&
{\bf $30\%$}&
{\bf $40\%$} \\
\midrule
\FedAPM~({\bf ours})& 
40 &
120 &
200 &
820&
20 &
60 &
200 &
1360 &
0 &
20 &
120 &
540 
\\
\FedAvg~over {\em active} clients& 
20 &
80 &
160 &
900 &
10 &
20 &
120 &
1060&
0 &
20 &
40 &
800
\\
\FedAvg~over {\em all} clients& 
100 &
420 &
960 &
-- &
20 &
60 &
520 &
--&
0 &
20 &
200 &
--
\\ 
\FedAU & 
60 &
100 &
160 &
840 &
10 &
20 &
100 &
960&
0 &
20 &
80 &
460
\\
\FAST & 
40 &
120 &
200 &
1080&
20 &
40 &
160 &
1300&
0 &
20 &
60 &
540
\\
\noalign{\vspace{.5mm}}
\midrule
\noalign{\vspace{.5mm}}
\FedAvg~with {\em known} $p_i^t$'s 
& 
20 &
40 &
100 &
320 &
10 &
20 &
140 &
620 &
0 &
20 &
40 &
400
\\
\MIFA~({\em memory aided}) & 
20 &
80 &
140&
600 &
10 &
20 &
80&
700 &
0 &
20 &
40&
240 
\\
\bottomrule
\end{tabular}}
\end{table}

\textbf{Staleness studies.}
\prettyref{tab: slowdown supp} 
illustrates the first round to reach a targeted test accuracy under non-stationary client availability with sine trajectory.
Specifications can be found in the caption.
It can be easily checked that,
during the initial stage (the first three quarters),
\FedAPM~slightly lags behind~\FedAvg~over active
clients.
However,
when reaching the final stage (the last quarter),
\FedAPM~attains the target accuracy in a comparable or lower number of rounds to~\FedAvg~over active clients in the evaluations on SVHN and CINIC-10 datasets.
The slowdown of~\FedAPM~on CIFAR-10 dataset is worth further investigation.
In general, we arrive numerically at the conclusion that
the staleness incurred by implicit gossiping in~\FedAPM~is mild.

\textbf{Training curves.}
In this part,
we show the training curves of~\FedAvg~over active clients,~\FedAPM~and~\MIFA.
In particular,
the presented results of~\FedAPM~are after
exponential moving average
\cite{busbridge2024scale} under a parameter $0.99$.
Note that this is to ease down the noisy progress,
and for a neat presentation only,
the reported results in the main text and ablation studies are all from raw data.
\prettyref{fig: svhn train details coarse} plots the train loss and test accuracy from raw data.
For example,
when compared with~\prettyref{fig: svhn train details refined},
EMA eases down the fluctuations but does not change either the trend or the order of algorithm performance results.
All train losses are plotted on a logarithmic scale.
The results are consistent with~\prettyref{tab: exp main text}.
\begin{figure}[!t]
    \centering
    \begin{subfigure}[b]{\linewidth}
        \includegraphics[width=.5\linewidth]{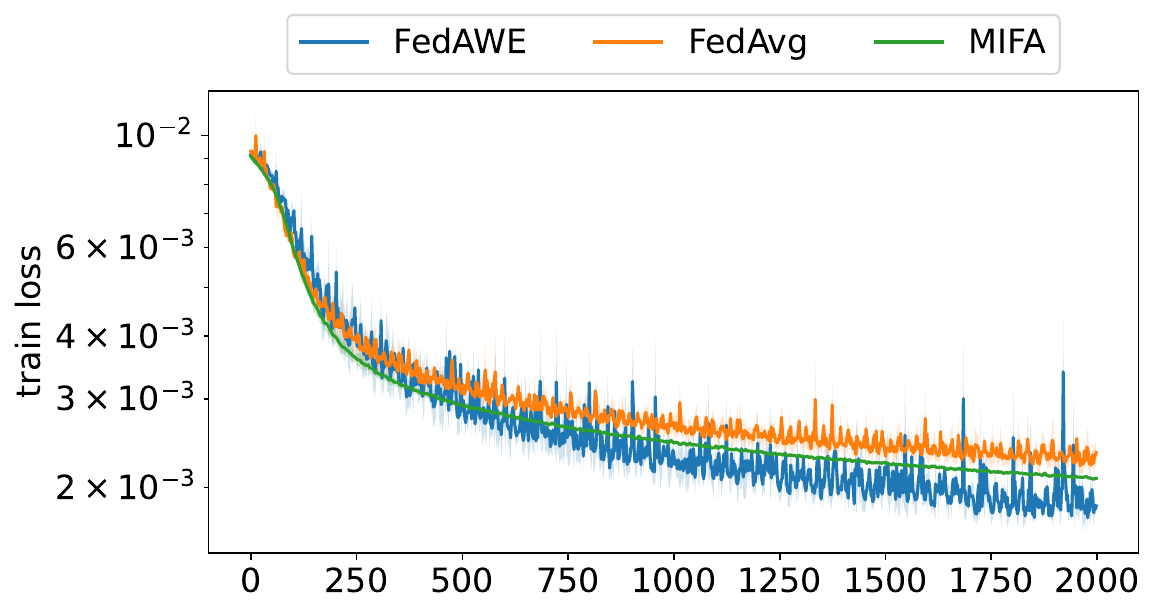}
        \includegraphics[width=.48\linewidth]{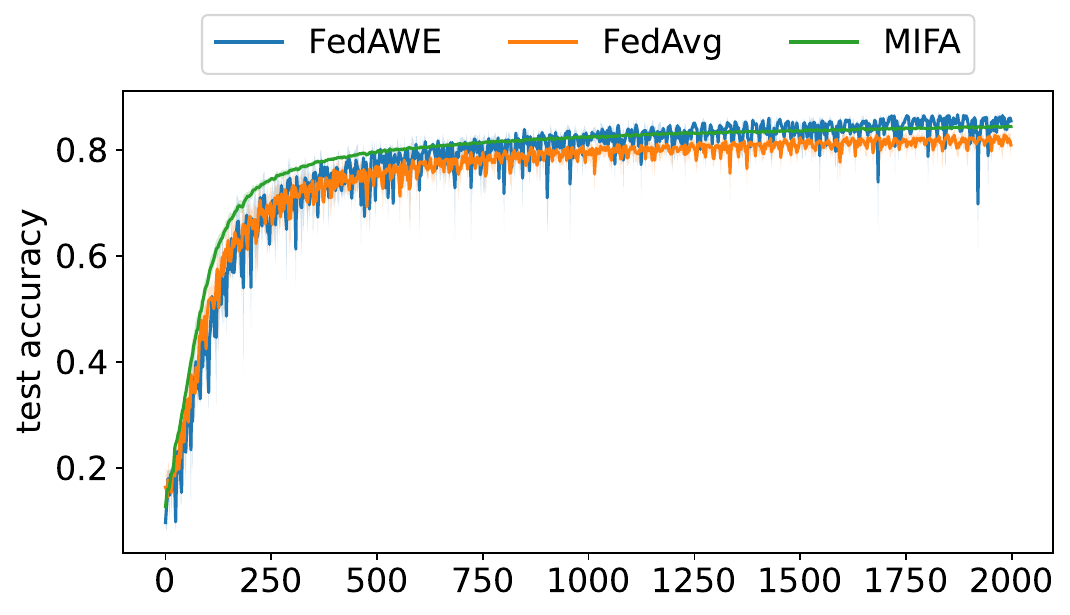}
        \caption{Evaluation results on SVHN dataset without exponential moving average}
        \label{fig: svhn train details coarse}
    \end{subfigure}
    \begin{subfigure}[b]{\linewidth}
        \includegraphics[width=.5\linewidth]{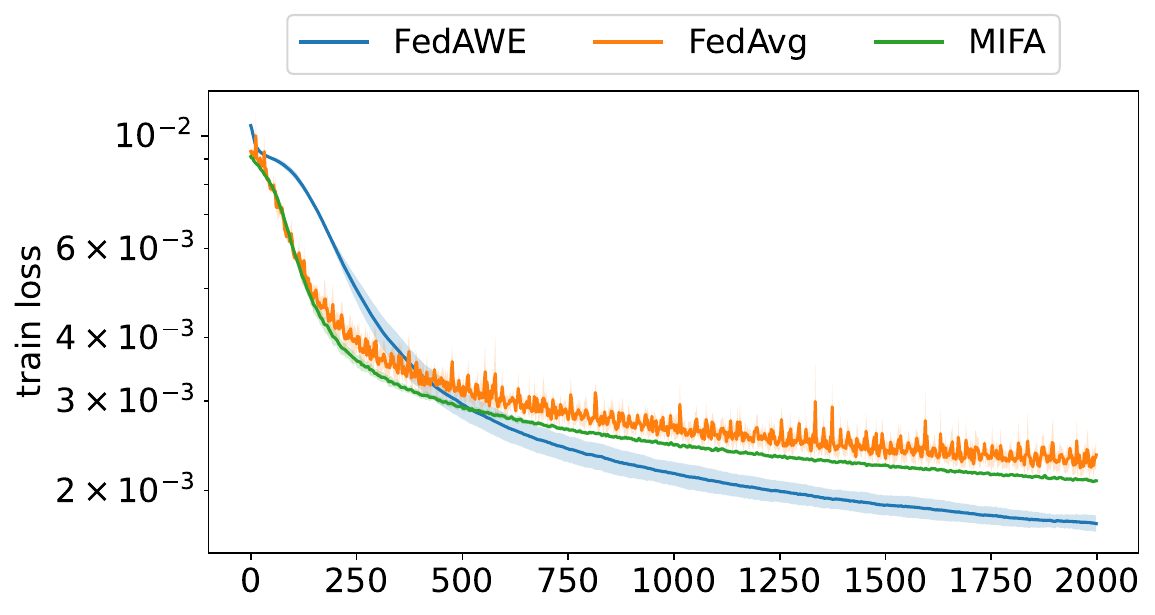}
        \includegraphics[width=.48\linewidth]{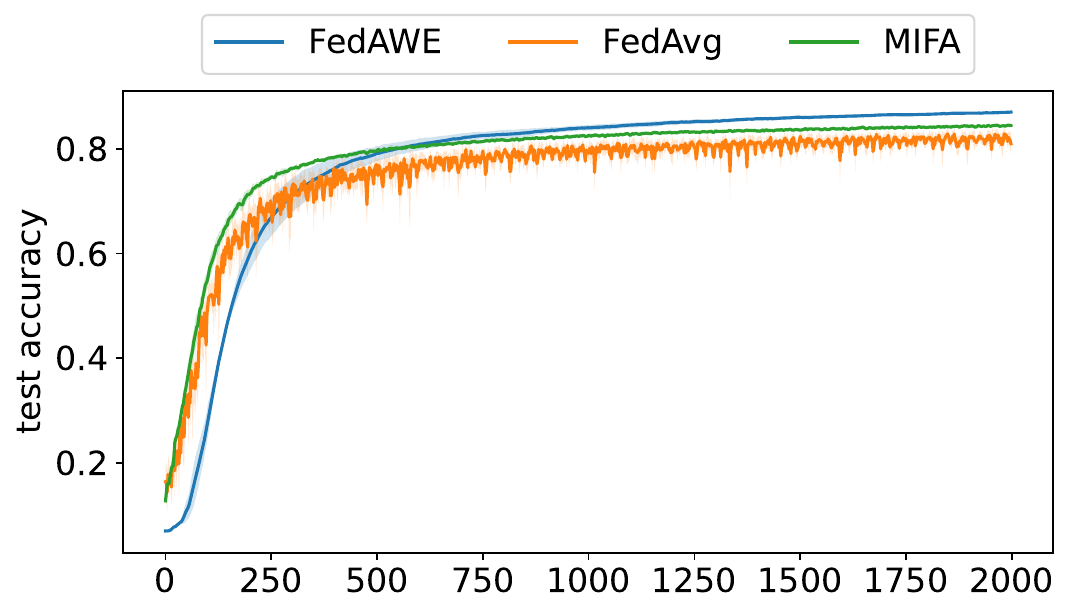}
        \caption{Evaluation results on SVHN dataset}
        \label{fig: svhn train details refined}
    \end{subfigure}
    \begin{subfigure}[b]{\linewidth}
        \includegraphics[width=.5\linewidth]{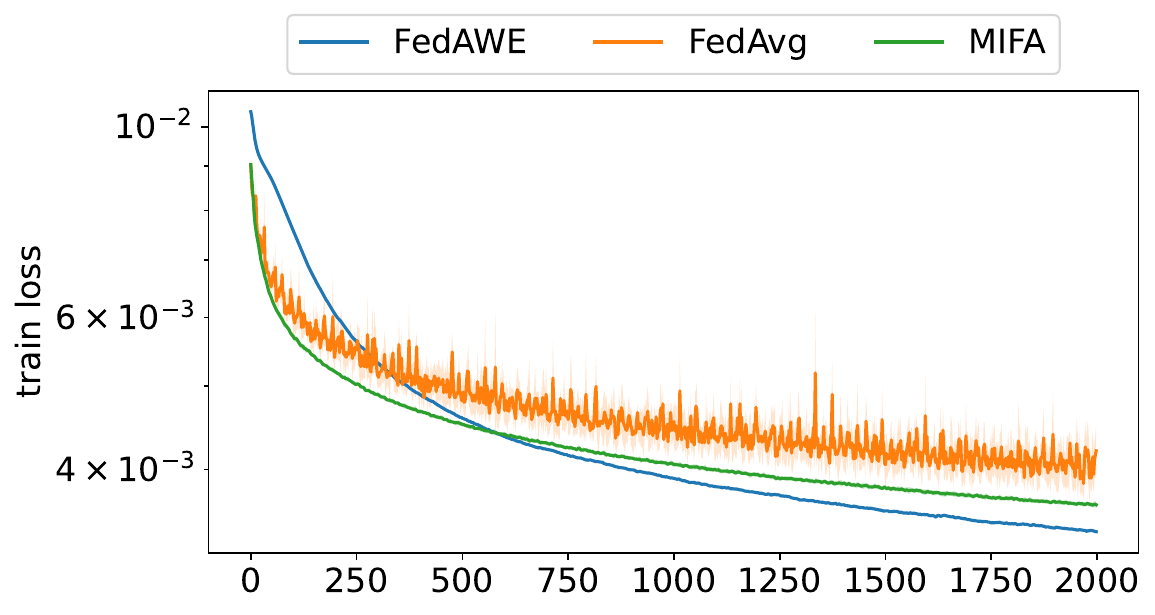}
        \includegraphics[width=.48\linewidth]{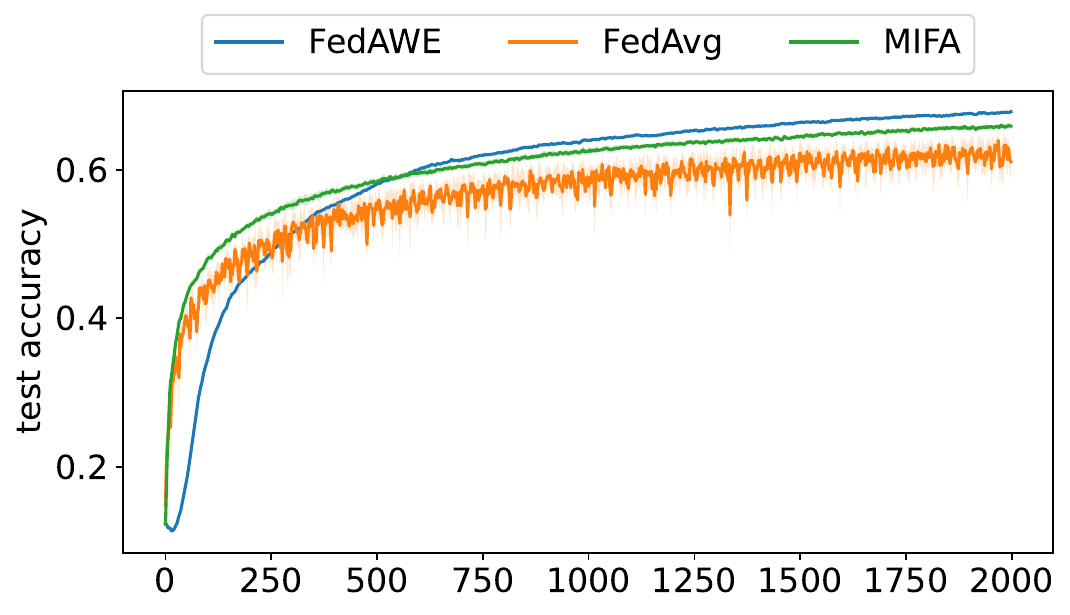}
        \caption{Evaluation results on CIFAR10 dataset}
        \label{fig: cifar10 train details refined}
    \end{subfigure}
    \begin{subfigure}[b]{\linewidth}
        \includegraphics[width=.5\linewidth]{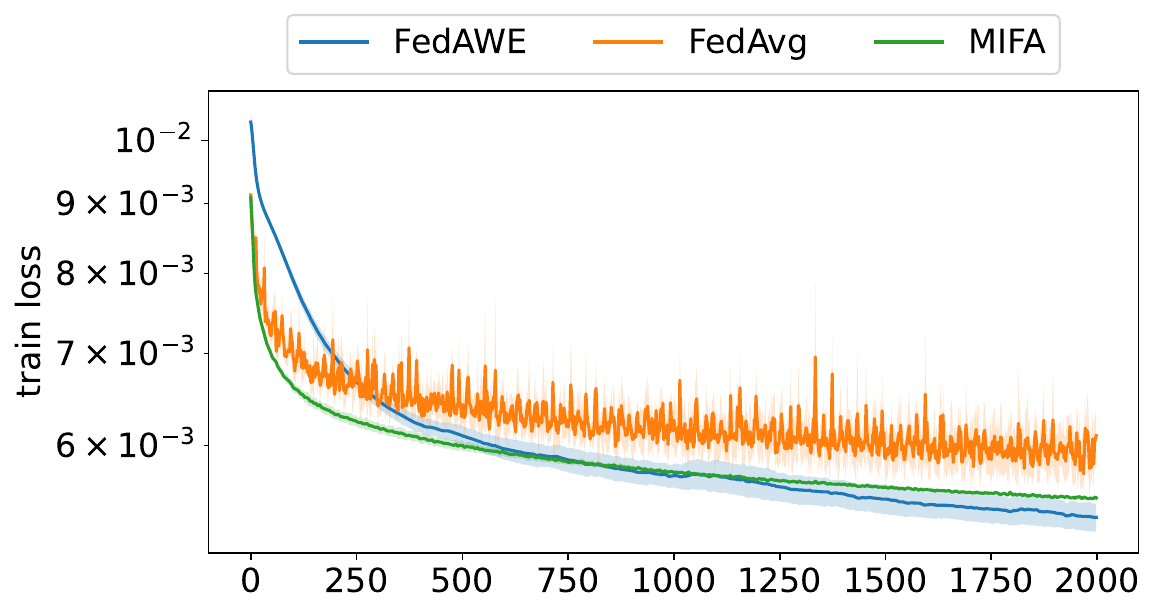}
        \includegraphics[width=.48\linewidth]{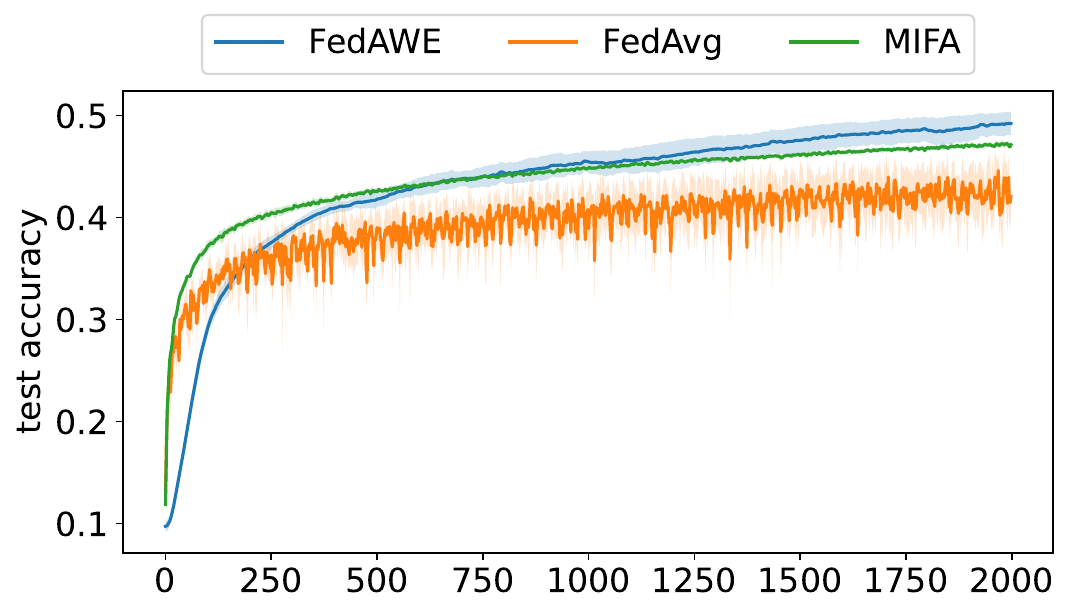}
        \caption{Evaluation results on CINIC10 dataset}
        \label{fig: cinic10 train details refined}
    \end{subfigure}
    \caption{Missing training curves under non-stationary client unavailability dynamics with sine curve}
    \label{fig: app train curves}
\end{figure}

\begin{table}[!t]
\centering
\caption{\footnotesize 
Results after different parameter $\gamma$.
$ p_i^t = p_i \cdot ( \gamma \sin(2 \pi / P  \cdot t) + (1 - \gamma))$.}
\label{tab: exp supp flu gamma}
\resizebox{\linewidth}{!}{
\begin{footnotesize}
\begin{tabular}{c|c|p{1.7cm} p{1.7cm}|p{1.7cm} p{1.7cm}|p{1.7cm} p{1.7cm}}
    \toprule
    \multirow{2}{*}{\begin{tabular}{@{}c@{}}{\bf Unavailable} \\ {\bf Dynamics} \end{tabular}} &
    {\bf Datasets} &
    \multicolumn{2}{c|}{\bf $\gamma = 0.3$} & 
    \multicolumn{2}{c|}{\bf $\gamma = 0.2$} & 
    \multicolumn{2}{c}{\bf $\gamma = 0.1$} \\
    \cline{2-8}
    & 
    {\bf Algorithms}&
    \multicolumn{1}{c}{\bf Train} &
    \multicolumn{1}{c|}{\bf Test} &
    \multicolumn{1}{c}{\bf Train} &
    \multicolumn{1}{c|}{\bf Test} &
    \multicolumn{1}{c}{\bf Train} &
    \multicolumn{1}{c}{\bf Test} \\
    \hline
    \multirow{6}{*}{
\begin{tabular}{@{}c@{}} 
\addlinespace[1ex]
{\bf Non}-stationary  \\ 
({\bf Sine})\\
\adjustbox{width=0.12\linewidth}{\includestandalone{elements/figure/table_fig/sine_curve}}
\end{tabular}} 
& 
\FedAPM~({\bf ours})& 
{\bf 85.7} $\pm$ 0.9 \%&
{\bf 85.6} $\pm$ 0.9 \%&
{\bf 85.7} $\pm$ 0.5 \%&
{\bf 85.7} $\pm$ 0.5 \%&
{\bf 85.8} $\pm$ 0.6 \%&
{\bf 85.7} $\pm$ 0.7 \%
\\
& 
\FedAvg~over {\em active} & 
82.1  $\pm$ 1.1 \%&
82.0  $\pm$ 1.3 \%&
82.0  $\pm$ 1.2 \%&
81.9  $\pm$ 1.2 \%&

82.3  $\pm$ 0.9 \%&
82.2  $\pm$ 1.0 \%
\\
& 
\FedAvg~over {\em all} & 
71.3 $\pm$ 2.5 \%&
71.3 $\pm$ 2.8 \%&
73.2  $\pm$ 2.5 \%&
73.2  $\pm$ 2.8 \%&
74.0  $\pm$ 2.1 \%&
74.9  $\pm$ 2.4 \%
\\
& 
\FedAU & 
\underline{82.5} $\pm$ 1.4 \%&
\underline{82.5} $\pm$ 1.3 \%&
\underline{83.5}  $\pm$ 0.3 \%&
\underline{83.4}  $\pm$ 0.4 \%&
\underline{83.7}  $\pm$ 0.3 \%&
\underline{83.6}  $\pm$ 0.3 \%
\\
& 
\FAST & 
82.3  $\pm$ 1.0 \%&
82.3  $\pm$ 1.0 \%&
82.3 $\pm$ 0.9 \%&
82.6  $\pm$ 0.8 \%&
82.9  $\pm$ 0.7 \%& 
82.9 $\pm$ 0.6 \%
\\
\noalign{\vspace{.5mm}}
\cline{2-8}
\noalign{\vspace{.5mm}}
& 
\FedAvg~with {\em known} $p_i^t$'s  & 
86.3  $\pm$ 1.0 \%&
86.0  $\pm$ 1.0 \%&
86.2  $\pm$ 1.2 \%&
86.0  $\pm$ 1.4 \%&
86.4  $\pm$ 0.9 \%&
86.0  $\pm$ 0.8 \%
\\
& 
\MIFA~({\em memory aided}) & 
84.2  $\pm$ 0.4 \%&
84.1  $\pm$ 0.4 \%&
84.6  $\pm$ 0.1 \%&
84.5  $\pm$ 0.1 \%&
84.6  $\pm$ 0.1 \%&
84.4  $\pm$ 0.1 \%
\\
\bottomrule
\end{tabular}
\end{footnotesize}
}
\end{table}
\begin{table}[!t]
\centering
\caption{\footnotesize 
Results after different $\mathsf{Dirichlet}$ parameter $\alpha$.
$ p_i^t = p_i (\gamma \sin(2 \pi / P  \cdot t) + (1 - \gamma))$.
}
\label{tab: exp supp dir alpha}
\resizebox{\linewidth}{!}{
\begin{footnotesize}
\begin{tabular}{c|c|p{1.7cm} p{1.7cm}|p{1.7cm} p{1.7cm}|p{1.7cm} p{1.7cm}}
    \toprule
    \multirow{2}{*}{\begin{tabular}{@{}c@{}}{\bf Unavailable} \\ {\bf Dynamics} \end{tabular}} &
    {\bf Datasets} &
    \multicolumn{2}{c|}{\bf $\alpha = 0.05$} & 
    \multicolumn{2}{c|}{\bf $\alpha = 0.1$} & 
    \multicolumn{2}{c}{\bf $\alpha = 1.0$} \\
    \cline{2-8}
    & 
    {\bf Algorithms}&
    \multicolumn{1}{c}{\bf Train} &
    \multicolumn{1}{c|}{\bf Test} &
    \multicolumn{1}{c}{\bf Train} &
    \multicolumn{1}{c|}{\bf Test} &
    \multicolumn{1}{c}{\bf Train} &
    \multicolumn{1}{c}{\bf Test} \\
    \hline
    \multirow{6}{*}{
\begin{tabular}{@{}c@{}} 
\addlinespace[1ex]
{\bf Non}-stationary  \\ 
({\bf Sine}) \\
\adjustbox{width=0.12\linewidth}{\includestandalone{elements/figure/table_fig/sine_curve}}
\end{tabular}} 
& 
\FedAPM~({\bf ours})& 
{\bf 82.5} $\pm$ 2.1 \%&
{\bf 82.5} $\pm$ 2.4 \%&

{\bf 85.7} $\pm$ 0.9 \%&
{\bf 85.6} $\pm$ 0.9 \%&

{\bf 90.6} $\pm$ 0.2 \%&
{\bf 89.7} $\pm$ 0.3 \%
\\
& 
\FedAvg~over {\em active} & 
78.9  $\pm$ 1.6 \%&
78.5  $\pm$ 1.8 \%&

82.1  $\pm$ 1.1 \%&
82.0  $\pm$ 1.3 \%&

88.3  $\pm$ 0.1 \%&
87.5  $\pm$ 0.1 \%
\\
& 
\FedAvg~over {\em all} & 
58.5  $\pm$ 3.0 \%&
58.5  $\pm$ 3.8 \%&

71.3 $\pm$ 2.5 \%&
71.3 $\pm$ 2.8 \%&

82.0  $\pm$ 0.7 \%&
81.9  $\pm$ 0.6 \%
\\
& 
\FedAU & 
\underline{79.5}  $\pm$ 1.6 \%&
\underline{79.5}  $\pm$ 1.7 \%&

\underline{82.5} $\pm$ 1.4 \%&
\underline{82.5} $\pm$ 1.3 \%&

\underline{88.4}  $\pm$ 0.1 \%&
\underline{87.6}  $\pm$ 0.2 \%
\\
& 
\FAST & 
78.9 $\pm$ 1.3 \%&
78.9 $\pm$ 1.3 \%&
82.3  $\pm$ 1.0 \%&
82.3  $\pm$ 1.0 \%&
87.6  $\pm$ 0.1 \%& 
87.0  $\pm$ 0.1 \%
\\
\noalign{\vspace{.5mm}}
\cline{2-8}
\noalign{\vspace{.5mm}}
& 
\FedAvg~with {\em known} $p_i^t$'s  & 
84.2  $\pm$ 1.0 \%&
83.5  $\pm$ 1.0 \%&
86.3  $\pm$ 1.0 \%&
86.0  $\pm$ 1.0 \%&
91.5  $\pm$ 0.3 \%&
90.5  $\pm$ 0.1 \%
\\
& 
\MIFA~({\em memory aided}) & 
82.6  $\pm$ 0.1 \%&
82.6  $\pm$ 0.0 \%&

84.2  $\pm$ 0.4 \%&
84.1  $\pm$ 0.4 \%&

88.4  $\pm$ 0.1 \%&
87.5  $\pm$ 0.1 \%
\\
\bottomrule
\end{tabular}
\end{footnotesize}
}
\end{table}
\textbf{Impact of system-design parameters.}
In this part,
we study the impact of system-design parameter
including 
the degree of non-stationarity $\gamma$
and data heterogeneity $\alpha$
under non-stationary with sine trajectory.
The results are in~\prettyref{tab: exp supp flu gamma} and~\prettyref{tab: exp supp dir alpha}.
Overall,~\FedAPM~keeps outperforming the algorithms not assisted by memories or known statistics.

In~\prettyref{tab: exp supp dir alpha},
clients' local data becomes more heterogeneous when $\alpha$ increases.
We can see a clear increase trend in accuracy.
However,
\FedAPM~remains to attain the best accuracies both train and test when compared to the algorithms not aided by memory or known statistics.
Moreover,
it outperforms~\MIFA,
which consumes a lot of storage space,
when $\alpha = 0.1$ and $1.0$.
The observations confirm the practicality of~\FedAPM.

\end{document}

%% file: elements/figure/table_fig/updated/flat_dynamic.tex
  \begin{overpic}[scale=1]{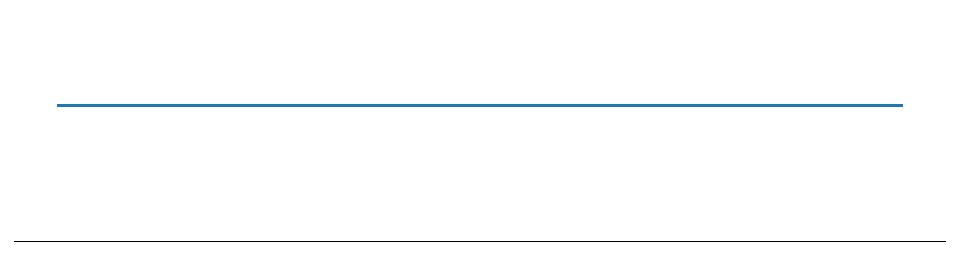} %
    \put(0,21){\fontsize{60}{70}\selectfont $p_i$}
    \put(95,5){\fontsize{40}{50}\selectfont $0$}
  \end{overpic}

%% file: elements/figure/table_fig/updated/staircase_dynamic.tex
  \begin{overpic}[scale=1]{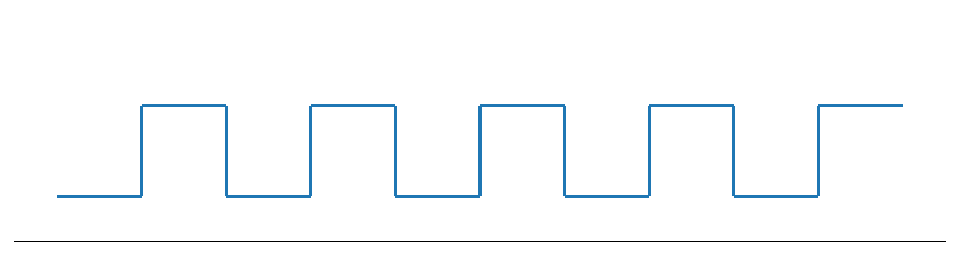} %
    \put(0,17){\fontsize{60}{70}\selectfont $p_i^t$}
    \put(95,5){\fontsize{40}{50}\selectfont $0$}
  \end{overpic}

%% file: elements/figure/table_fig/updated/sine_dynamic.tex
  \begin{overpic}[scale=1]{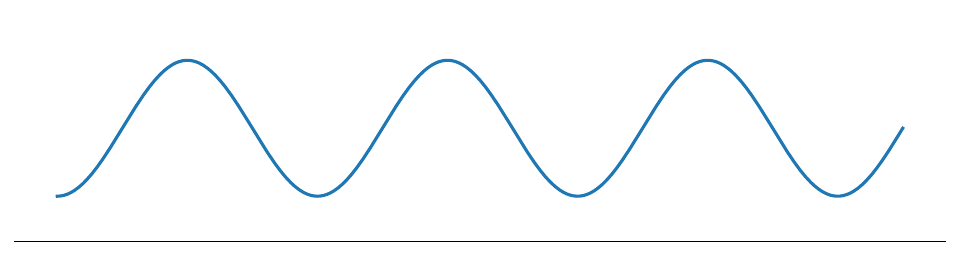} %
    \put(0,17){\fontsize{60}{70}\selectfont $p_i^t$}
    \put(95,5){\fontsize{40}{50}\selectfont $0$}
  \end{overpic}

%% file: elements/figure/table_fig/updated/interleaved_sine.tex
  \begin{overpic}[scale=1]{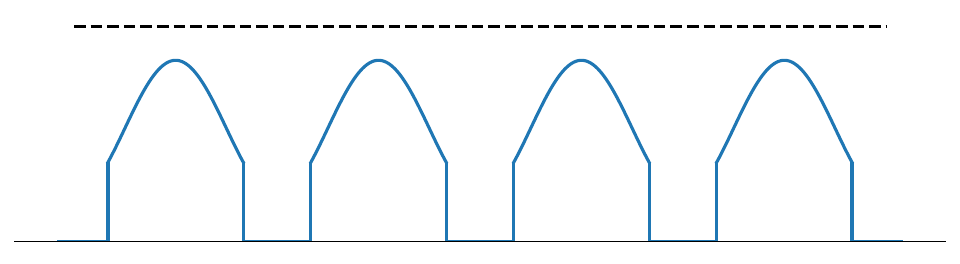} %
    \put(0,11){\fontsize{60}{70}\selectfont $p_i^t$}
    \put(95,5){\fontsize{40}{50}\selectfont $0$}
    \put(95,21){\fontsize{40}{50}\selectfont $1$}
  \end{overpic}

%% file: elements/figure/table_fig/sine_curve.tex
\begin{tikzpicture}
\begin{axis}[
    xmin=0, xmax=12*pi,
    ymin=-.6, ymax=1.3,
    domain=1.5*pi:11.5*pi,
    samples=500,
    axis lines=middle,
    hide y axis, %
    xtick=\empty, 
    ytick=\empty,
    ]
    \addplot[blue] {.3*sin(deg(.8*x)) + .6};
    \node[below left] at (axis cs:12*pi,-0.05) {\scalebox{3}{$ 0$}};  
    \node[below right] at (axis cs:0,1) {\scalebox{3}{$p_i^t$}};  
\end{axis}
\end{tikzpicture}